%% file: TVguarantees_paper_arXiv.tex
\documentclass[onefignum,onetabnum]{siamonline190516}

\usepackage{benstyle_SIAM}
\usepackage{lipsum}
\usepackage{amsfonts}
\usepackage{graphicx}
\usepackage{epstopdf}
\usepackage{algorithmic}

\usepackage{enumitem}

\newsiamremark{remark}{Remark}
\newsiamremark{hypothesis}{Hypothesis}
\crefname{hypothesis}{Hypothesis}{Hypotheses}
\newsiamthm{claim}{Claim}

\usepackage{amsopn}
\DeclareMathOperator{\diag}{diag}

\input{TVguarantees_shared}

\ifpdf
\hypersetup{
  pdftitle={Improved TV guarantees},
  pdfauthor={B. Adcock, N. Dexter and Q. Xu}
}
\fi

\begin{document}

\maketitle

\begin{abstract}
In this paper, we consider the use of Total Variation (TV) minimization for compressive imaging; that is, image reconstruction from subsampled measurements. Focusing on two important imaging modalities -- namely, Fourier imaging and structured binary imaging via the Walsh--Hadamard transform -- we derive uniform recovery guarantees asserting stable and robust recovery for arbitrary random sampling strategies. Using this, we then derive a class of theoretically-optimal sampling strategies. For Fourier sampling, we show recovery of an image with approximately $s$-sparse gradient from $m \gtrsim_d  s \cdot \log^2(s) \cdot \log^4(N)$ measurements, in $d \geq 1$ dimensions. When $d = 2$, this improves the current state-of-the-art result by a factor of $\log(s) \cdot \log(N)$. It also extends it to arbitrary dimensions $d \geq 2$. For Walsh sampling, we prove that $m \gtrsim_d s \cdot \log^2(s) \cdot \log^2(N/s) \cdot \log^3(N) $ measurements suffice in $d \geq 2$ dimensions. To the best of our knowledge, this is the first recovery guarantee for structured binary sampling with TV minimization.
\end{abstract}



\section{Introduction}

Total Variation (TV) minimization is an important technique in modern image processing \cite{ChambolleEtAlTVImage,ChamPockAN}, with a wide range of applications including denoising, deblurring and reconstruction. In this paper, we consider the latter problem. Specifically, given noisy, linear measurements $y = A x + e \in \bbC^m$ of an unknown $d$-dimensional image $x \in \bbC^{N^d}$, we study its reconstruction via the constrained TV minimization problem
\be{
\label{TVminintro}
\min_{z \in \bbC^{N^d}} \nm{z}_{\mathrm{TV}}\ \mbox{subject to $\nm{A z - y}_{\ell^2} \leq \eta$},
}
where $\nm{\cdot}_{\mathrm{TV}}$ is the TV semi-norm. Natural images have approximately sparse gradients. As is now well known, minimizing the TV semi-norm promotes this structure, often leading to high-quality reconstructions from a relatively small number of measurements. TV minimization has proved an extremely effective tool for image reconstruction, with many applications in medical, scientific and industrial modalities.

A fundamental issue in image reconstruction is choosing a measurement matrix $A$. The main goal of so-called \textit{compressive imaging} is to choose $A$ so as to deliver high-quality reconstructions from as few measurements $m$ as possible. Generally speaking, the possible choices are dictated by the physical sensing apparatus. In this paper, we consider two important image acquisition protocols: namely, Fourier sampling with the discrete Fourier transform and binary sampling via the Walsh--Hadamard transform. Arguably, these are two out of the three most important types of sampling encountered in imaging -- the other being the Radon transform. Fourier sampling arises in numerous applications, including Magnetic Resonance Imaging (MRI), Nuclear Magnetic Resonance (NMR) and radio interferometry, while binary sampling arises in numerous optical imaging modalities, such as lensless imaging, infrared imaging holography, fluorescence microscopy and so forth.

Once the acquisition protocol has been fixed, the task of selecting measurements reduces to that of designing a \textit{sampling strategy}, i.e.\ a specific choice of $m$ Fourier or Walsh frequencies to sample. The main objective of this paper is to develop sampling strategies for TV minimization in these scenarios. In tandem, we also derive sufficient conditions on the number of measurements $m$ under which the underlying image is accurately recovered via \eqref{TVminintro}. We do this by leveraging the theory of compressed sensing \cite{FoucartRauhutCSbook} to prove new recovery guarantees for TV minimization which relate the number of measurements $m$ to the approximate gradient sparsity $s$ of the underlying image.

\subsection{Previous work}

TV minimization was studied in some of the first papers on compressed sensing. In \cite{CandesRombergTao}, Cand\`es, Romberg \& Tao considered the recovery of a one-dimensional image $x \in \bbC^N$ with exactly $s$-sparse gradient from $m$ noiseless Fourier measurements taken uniformly and randomly. They showed that $x$ could be recovered exactly by solving \eqref{TVminintro} with $\eta =0$ with high probability, provided the number of measurements
$
m \gtrsim s \cdot  \log(N).
$
The first results asserting recovery for approximately sparse images from noisy measurements were shown by Needell \& Ward for the two-dimensional case in \cite{NeedellWardTV1}, and later for the $d$-dimensional case in \cite{NeedellWardTV2}. In particular, these works were the first to exploit (in the compressed sensing context) the important connection between the TV semi-norm and Haar wavelet coefficients. Neither of these works pertained directly to Fourier sampling. The first results on Fourier sampling were shown by Krahmer \& Ward \cite{KrahmerWardCSImaging} and Poon \cite{PoonTV}.  In the former, uniform recovery guarantees\footnote{In compressed sensing, a \textit{uniform} recovery guarantee states that a single random draw of a given measurement matrix suffices for recovery of all (approximately) sparse vectors. This is stronger than a nonuniform recovery guarantee, which asserts that a single random draw is sufficient for recovery of a fixed vector.} were shown for two-dimensional images from noisy Fourier measurements, with frequencies chosen randomly according to an \textit{inverse square law} density. Specifically, if
\be{
\label{KWmeascondintro}
m \gtrsim s \cdot \log^3(s) \cdot \log^5(N),
}
then with high probability, the recovered vector $\hat{x}$ satisfies
\be{
\label{KWrecintro}
\nm{x - \hat{x}}_{\ell^2} \lesssim \frac{\sigma_s(\nabla x)_{\ell^1}}{\sqrt{s}} + \eta,\qquad \sigma_{s}(\nabla x)_{\ell^1} = \min \{ \nm{\nabla x  - z }_{\ell^1} : \mbox{$z \in \bbC^{N^2}$ is $s$-sparse} \},
}
where $\eta$ is an upper bound on a certain weighted $\ell^2$-norm of the noise term $e$. Conversely, \cite{PoonTV} established nonuniform recovery guarantees in the one- and two-dimensional cases for both uniform random sampling and variable density sampling. Amongst other features,  \cite{PoonTV} was the first to prove results demonstrating the benefits of variable density sampling: namely, while both uniform random and variable density sampling recover the image gradient accurately, the latter leads to better recovery of the image itself. In comparison with \eqref{KWmeascondintro}--\eqref{KWrecintro}, in the two-dimensional case \cite{PoonTV} showed that if
\be{
\label{Poonmeascondintro}
m \gtrsim s \cdot \log(N),
}
Fourier samples were drawn using a combination of uniform random and inverse square law sampling, then, with high probability,
\be{
\label{Poonrecintro}
\nm{x - \hat{x}}_{\ell^2} \lesssim \log(s) \cdot \log(N^2/s) \log^{1/2}(N) \log^{1/2}(m) \left ( \log^{1/2}(m) \log(s) \frac{\sigma_{s}(\nabla x)_{\ell^{2,1}} }{\sqrt{s}} + \eta \right ),
}
where $\eta$ is an upper bound for the (unweighted) $\ell^2$-norm of the noise (the appearance of $\sigma(\cdot)_{\ell^{2,1}}$ here indicates that \cite{PoonTV} considered the isotropic TV norm, whereas \cite{KrahmerWardCSImaging} considered the anisotropic TV norm -- see later). In particular, this approach leads to a better measurement condition \eqref{Poonmeascondintro} than the measurement condition \eqref{KWmeascondintro}, but a correspondingly worse error bound \eqref{Poonrecintro} over \eqref{KWrecintro}.

\subsection{Contributions}

The above results of \cite{KrahmerWardCSImaging} and \cite{PoonTV} represent the state-of-the-art recovery guarantees for TV minimization in compressed sensing with Fourier sampling. In this paper we improve and generalize these results in the following ways:

{\bf 1.}\ We derive recovery guarantees in $d \geq 1$ dimensions, as opposed to $d=2$ in \cite{KrahmerWardCSImaging} and $d=1,2$ in \cite{PoonTV}. We consider both the isotropic (like in \cite{PoonTV}) and anisotropic (like in \cite{KrahmerWardCSImaging}) TV semi-norms. Also as in \cite{PoonTV} we examine both uniform random and variable density sampling.

{\bf 2.}\ As in \cite{KrahmerWardCSImaging}, our recovery guarantees are uniform, and when $d \geq 2$ they take the form
\be{
\label{introerrbd1}
\nm{x - \hat{x}}_{\ell^2} \lesssim_d \frac{\sigma_s(\nabla x)_{\ell^1}}{\sqrt{s}} + \sqrt{\log(N)} \eta,
}
for variable density sampling. Unlike \cite{KrahmerWardCSImaging}, we do not impose a weighted norm on the noise vector. This gives rise to the $\sqrt{\log(N)}$ factor in \eqref{introerrbd1}. As in \cite{PoonTV}, we also derive error bounds for the recovery of the image gradient $\nabla x$.

{\bf 3.}\ Unlike \cite{KrahmerWardCSImaging,PoonTV} we derive a recovery guarantee for arbitrary variable density sampling schemes in order to examine the effect of the sampling scheme on the measurement condition. 

{\bf 4.}\ We derive theoretically-optimal variable density sampling schemes in $d \geq 1$ dimensions. For such schemes, our measurement condition is
\bes{
m \gtrsim_{d} s \cdot \log^2(s) \cdot \log^4(N).
}
In particular, for the $d = 2$ case, we improve the current state-of-the-art measurement condition \eqref{KWmeascondintro} for uniform recovery by a factor of $\log(s) \cdot \log(N)$. When $d = 2$ we show that the inverse square law scheme of \cite{KrahmerWardCSImaging,PoonTV} is an instance of a theoretically-optimal scheme.

{\bf 5.}\ Interestingly, we show that the theoretically-optimal Fourier sampling scheme ceases to be radially-symmetric in $d \geq 3$ dimensions. We also derive a near-optimal sampling scheme based on so-called \textit{hyperbolic cross} sampling densities.

{\bf 6.}\ Finally, unlike \cite{KrahmerWardCSImaging,PoonTV} we also consider binary sampling with the Walsh--Hadamard transform. In this case, we prove a recovery guarantee of the form
\bes{
\nm{x - \hat{x}}_{\ell^2} \lesssim_d \frac{\sigma_s(\nabla x)_{\ell^1}}{\sqrt{s \log(N)}} + \sqrt{\log(N)} \eta,
}
and derive theoretically-optimal variable density sampling strategies for which the measurement condition reads
\bes{
m \gtrsim_{d} s\cdot \log^2(s) \cdot \log^2(N/s) \cdot \log^3(N) .
}
Unlike in the Fourier case, we show that certain radially-symmetric sampling schemes are theoretically optimal in any dimension for Walsh sampling. To the best of our knowledge, this is the first recovery guarantee for TV minimization with structured binary sampling. For results on binary sampling with wavelet sparsifying transforms, see \cite{AAHWalshWavelet,MoshtThesis,MoshtEtAlHadamardHaar}.

Note that our focus in this paper is on Fourier and Walsh sampling, since these acquisition protocols arise in many practical imaging settings. Although common in compressed sensing, we do not consider sampling with Gaussian or Bernoulli random matrices. These are generally infeasible for imaging, since they lead to dense, unstructured matrices. Moreover, even if they were, it is well known that they are highly suboptimal for imaging, being significantly outperformed by structured Fourier and Walsh sampling \cite{AHPRBreaking,BASBMKRCSwavelet,AsymptoticCS}. For recovery guarantees for TV minimization from Gaussian or Bernoulli measurements, see \cite{CaiXuTVGauss,KrahmerEtAlTVCS}.

\vspace{-1mm}
\subsection{Structure dependence}\label{s:limitations} 
Similar to \cite{KrahmerWardCSImaging,PoonTV}, the sampling schemes we develop in this paper are independent of the image (or class of images) being recovered. In particular, they exploit only the sparsity of $\nabla x$ and no further local, or geometric, properties of the edges of $x$. As has been well documented \cite{AHPRBreaking,PoonTV,AsymptoticCS}, optimal sampling strategies in practice should also take local properties into account: roughly speaking, an image with well separated edges should be sampled more densely at low frequencies than an image with edges that lie close to each other, even when the two images possess the same gradient sparsity. In the case of sparsity in orthonormal wavelets, it is well understood (from a theoretical and practical perspective) how to design sampling strategies that exploit such local structure \cite{AAHWalshWavelet,BASBMKRCSwavelet,AHPRBreaking,BastounisHansen,LiAdcockRIP}. Yet, this is not well understood for gradient sparsity. We shall not attempt to tackle this problem, although we do discuss it in the context of our numerical examples. We refer to \cite{ChambolleEtAlTVDenoise,PoonTV} for some further discussion on this topic. Nonetheless, as we show in our examples, good all-round performance across a range of images, resolutions and sampling percentages can be achieved with an (image independent) \textit{multilevel random sampling} strategy. This scheme was originally developed for wavelet sparsifying transforms in  \cite{AHPRBreaking}. We show that it also achieves similarly good performance for TV minimization.

\subsection{Outline}
We begin in \S \ref{s:prelims} with preliminaries. We state our main results for Fourier and Walsh sampling in \S \ref{s:Fourmain}--\S \ref{s:Foursamp} and \S \ref{s:BeyondFourierTV} respectively. In \S \ref{s:experiments} we present several numerical experiments. Finally, in \S \ref{s:proofsFourier}--\S\ref{s:proofsII} we give the proofs of the main results. The Supplementary Material contains some supporting material and proofs of several of the minor results.

\section{Preliminaries}\label{s:prelims}

We first introduce some notation and background material.

\subsection{Notation}

We denote the $\ell^p$-norm on $\bbC^N$ by $\nm{\cdot}_{\ell^p}$ and the $\ell^2$-inner product by $\ip{\cdot}{\cdot}$. For $1 \leq p ,q < \infty$, we define the $\ell^{p,q}$-norm on $\bbC^{N \times M}$ as
\bes{
\nm{X}_{\ell^{p,q}} = \left ( \sum^{N}_{i=1} \left ( \sum^{M}_{j=1} |x_{ij} |^p \right )^{q/p} \right )^{1/q},\qquad X = (x_{ij})^{N,M}_{i,j=1}.
}
Note that $\nm{X}_{\ell^{2,2}} = \nm{X}_{F}$ is the Frobenius norm of $F$.
We define the $\ell^0$-norm of a vector $x = (x_i)^{N}_{i=1}$ as $\nm{x}_{\ell^0} = | \mathrm{supp}(x) |$, where $\mathrm{supp}(x) = \{ i : x_i \neq 0\}$ is the support of $x$. For a matrix $X = (x_{ij})^{N,M}_{i,j=1} \in \bbC^{N \times M}$ we define the $\ell^{2,0}$-norm as
\bes{
\nm{X}_{\ell^{2,0}} = | \mathrm{supp}(X) |,\qquad \mathrm{supp}(X) = \left \{ i : \sum^{M}_{j=1} |x_{ij}|^2 \neq 0 \right \} .
}
Given a subset $\Delta \subseteq \{1,\ldots,N\}$, we write $P_{\Delta} \in \bbC^{N \times N}$ for the diagonal matrix corresponding to the orthogonal projection with range $\mathrm{span} \{ e_i : i \in \Delta \}$, where $\{ e_i \}^{N}_{i=1}$ is the canonical basis of $\bbC^N$. Note that for $x \in \bbC^N$, the vector $P_{\Delta} x$ is isometrically isomorphic to a vector in $\bbC^{|\Delta|}$, and similarly for $X \in \bbC^{N \times M}$, $P_{\Delta} X$ is isomorphic to an element of $\bbC^{|\Delta| \times M}$. On occasion we therefore slightly abuse notation and consider  $P_{\Delta} x$ as an element of $\bbC^{|\Delta|}$ or $P_{\Delta} X$ as an element of $\bbC^{|\Delta| \times N}$.

We write $C >0$ for a numerical constant and $C_x > 0$ for a constant depending on a variable $x$. We use the notation $a \lesssim b$ to mean there exists $C>0$ such that $a \leq C b$, and likewise for the symbol $\gtrsim$. We also write $a \lesssim_x b$ when $a \leq C_x b$ for some $C_x > 0$ depending on a variable $x$, and likewise for $\gtrsim_x$. We write $a \asymp b$ or $a \asymp_x b$ if $a \lesssim b \lesssim a$ or  $a \lesssim_x b \lesssim_x a$.

\subsection{Discrete images}

We consider discrete, $d$-dimensional complex images
\bes{
X = (X_{i_1,\ldots,i_d})^{N}_{i_1,\ldots,i_d = 1} \in \bbC^{N \times \cdots \times N}, 
}
where $N$ is its \textit{resolution}. The motivation to consider complex images stems primarily from MRI, where the images are often complex.
We assume throughout this paper that $N = 2^r$ is a power of two, where $r \geq 1$.
We often reshape $X$ into a vector using lexicographical ordering. Let $\varsigma:\{1,\ldots,N^d\}\to\{1,\ldots,N\}^d$ be the bijection corresponding to this ordering, defined via its inverse as
\begin{equation*}
\varsigma^{-1}(i_1,\ldots,i_d)=N^{d-1}i_1+N^{d-2}i_2+\ldots+i_d,\qquad (i_1,\ldots,i_d) \in \{1,\ldots,N\}^d.
\end{equation*}
Given $X$, we let $x = (x_i)^{N^d}_{i=1} \in \bbC^{N^d}$ be such that $x_i = X_{\varsigma(i)}$ and write $x = \mathrm{vec}(X)$.

\subsection{The Discrete Fourier Transform and recovery problem}

We order frequency from lowest to highest in absolute value. Define the bijection
\begin{equation}
\label{fourier1}
\varrho:\{1,\ldots,N\}\to\left\{-N/2+1,\ldots,N/2\right\},\ i \mapsto (-1)^i\left\lfloor i/2\right\rfloor.
\end{equation}
With this order, we define the one-dimensional Discrete Fourier Transform (DFT) matrix $F = F^{(1)} \in \bbC^{N \times N}$ as
\bes{
F_{ij} = \exp(-2 \pi \mathrm{i}\varrho(i)(j-1)/N ), \quad i,j = 1,\ldots,N,
}
(this differs from the usual DFT matrix by a row permutation and diagonal scaling, but is beneficial for our purposes as it orders frequencies from lowest to highest).

The $d$-dimensional DFT $F = F^{(d)} \in \bbC^{N^d \times N^d}$ is given by $F^{(d)} = F^{(1)} \otimes \cdots \otimes F^{(1)}$, where $\otimes$ denotes the Kronecker product. Notice that $N^{-d} F^* F = I$ is the identity matrix. 
The rows of $F^{(d)}$ correspond to the $d$-dimensional frequency space $\{-N/2+1,\ldots,N/2\}^d$. Specifically, let $\varrho = \varrho^{(d)} : \{1,\ldots,N^d\} \rightarrow \{-N/2+1,\ldots,N/2\}^d$ be the bijection defined by
\be{
\label{fourierd}
\varrho^{(d)}(i)  =  (\varrho(\varsigma(i)_1),\ldots,\varrho(\varsigma(i)_d)),\qquad i \in \{1,\ldots,N^d\},
}
where $\varsigma$ is the lexicographical ordering and $\varrho$ is the one-dimensional bijection \eqref{fourier1}. Then the $i^{\rth}$ row of $F^{(d)}$ corresponds to the frequency $\omega = \varrho(i)$.

In the first part of this paper, we consider the problem of recovering a vectorized image $x$ from $m$ of its Fourier frequencies. The choice of frequencies is variously referred to as a \textit{sampling scheme}, \textit{strategy}, \textit{map} or \textit{pattern}. We consider two main types of sampling schemes:

\defn{
[Uniform random sampling]
\label{d:unifsamppatt}
A $d$-dimensional \textit{uniform random sampling scheme of order $m$} is a subset of frequencies $\Omega = \{ \omega_1,\ldots,\omega_m \} \subseteq \{ - N/2+1,\ldots,N/2\}^d$ where the $\omega_i$ are chosen independently and uniformly from $\{ - N/2+1,\ldots,N/2\}^d$.
}

\defn{
[Variable density sampling]
\label{d:VDsamppatt}
Let $p= ( p_{\omega} )$ be a probability distribution on $\{ - N/2+1,\ldots,N/2\}^d$.  A $d$-dimensional \textit{variable density sampling scheme of order $m$} is a subset of frequencies $\Omega = \{ \omega_1,\ldots,\omega_m \} \subseteq \{ - N/2+1,\ldots,N/2\}^d$ where the $\omega_i$ are chosen i.i.d.\ according to $p$.  
}

Let $\Omega$ be given by one of these schemes. With slight abuse of notation, write $P_{\Omega} \in \bbC^{N^d \times N^d}$ for the orthogonal projection onto the indices in $\Omega$ (technically, this should be $P_{\varrho^{-1}(\Omega)}$ with $\varrho$ as in \eqref{fourierd}). Then we write $A = \frac{1}{\sqrt{m}} P_{\Omega} F \in \bbC^{m \times N^d}$ for the corresponding measurement matrix. This is an example of a subsampled DFT matrix: the vector $A x$ consists of the $m$ frequency values of the vectorized image $x$ from the set $\Omega$. We assume these values are also corrupted by noise, giving the vector of measurements
\bes{
y = A x + e \in \bbC^m,
}
where $e \in \bbC^m$ is a noise vector. With this in hand, the recovery problem we aim to solve is the following: \textit{given $y = A x + e$, recover $x$.}

\subsection{The Discrete Walsh--Hadamard Transform and recovery problem}\label{ss:DWT}

We now define the Discrete Walsh--Hadamard Transform. Recall that the one-dimensional (sequency-ordered) Walsh functions on $[0,1)$ are defined by
\bes{
v_n(x) = (-1)^{\sum^{\infty}_{i=1} (n_i + n_{i+1}) x_i},\qquad 0 \leq x < 1,\quad n \in \bbN_0,
}
where $(n_i)_{i \in \bbN} \in \{0,1\}^{\bbN}$ and $(x_i)_{i \in \bbN} \in \{0,1\}^{\bbN}$ are the dyadic expansions of $n$ and $x$ respectively (see \cite{Antun16,GaussWalsh,GolubovEtAlWalsh} and references therein for further information on Walsh functions). The number $n \in \bbN_0$ is the \textit{sequency} (number of sign changes) of the Walsh function; it is therefore analogous to the Fourier frequency.
The functions $v_n$ take values in $\{+1,-1\}$ and form an orthonormal basis of $L^2([0,1))$. When $N = 2^r$, the \textit{Discrete Walsh--Hadamard Transform} (DHT) arises by sampling this basis on an equispaced grid in $[0,1)$:
\bes{
H = H^{(1)} = \left ( v_m(n/N) \right )^{N-1}_{m,n=0} \in \{-1,1\}^{N \times N}.
}
Note that other orderings of the Walsh functions (or equivalently the rows of $H$) could be considered here, e.g.\ the Paley or ordinary orderings. The sequency ordering is convenient due to its connection to frequency; we therefore use it throughout.
When $d \geq 2$, we write $H^{(d)}  = H^{(1)} \otimes \cdots \otimes H^{(1)}$ for the $d$-dimensional DHT matrix. Note that in any dimension, $H$ is a symmetric matrix and is orthogonal up to a constant: specifically, $N^{-d} H^{\top} H = I$.

In $d \geq 1$ dimensions, the transform $x \mapsto H x$ computes the discrete Walsh--Hadamard measurements of a vectorized image $x$ corresponding to the frequencies in $\{ 0 ,\ldots,N-1\}^d$. Specifically, let $\varrho : \{1,\ldots,N^d \} \rightarrow \{0,\ldots,N-1\}^d$ be the bijection defined by
\bes{
\varrho(i) = (\varsigma(i)_1 - 1,\ldots,\varsigma(i)_d - 1),\qquad i \in \{1,\ldots,N^d\},
}
where $\varsigma$ is the lexicographical ordering.
Then the $i^{\rth}$ row of $H^{(d)}$ corresponds to the Walsh frequency $n = \varrho(i)$ with $i^{\rth}$ entry of $H x$ being the Walsh frequency of $x$.

Similar to the Fourier case, we consider sampling schemes $\Omega = \{ \omega_1,\ldots,\omega_m \} \subseteq \{0,\ldots,N-1\}^d$. Uniform random and variable density sampling schemes are all defined in the analogous manner, with the notable difference that in Walsh--Hadamard sampling the frequencies are nonnegative numbers only, as opposed to arbitrary integers. As in the Fourier case, we write $A = \frac{1}{\sqrt{m}} P_{\Omega} H \in \bbR^{m \times N^d}$ for the subsampled DHT matrix. Hence the noisy measurements are given by $y = A x + e \in \bbR^m$ and the recovery problem is to recover $x$ from $y$.

\subsection{Gradient operators and TV semi-norms}

We consider periodic gradient operators. The one-dimensional discrete gradient operator $\nabla:\mathbb{C}^N\to\mathbb{C}^N$ is defined by
\bes{
(\nabla x)_i=x_{i+1}-x_i,\quad i = 1,\ldots,N,
}
where $x=(x_i)_{i=1}^N$ and $x_{N+1}=x_1$. The one-dimensional \textit{Total Variation} semi-norm $\nm{\cdot}_{\mathrm{\mathrm{TV}}}$ is $\nm{x}_{\mathrm{\mathrm{TV}}}=\nm{\nabla x}_{\ell^1}$. Note that $\nabla$ is the circulant matrix generated by the vector $(-1,0\ldots,0,1)$.

In $d$ dimensions, we define the $j^{\rth}$ partial derivative operator $\nabla_j: \mathbb{C}^{N^d}\to \mathbb{C}^{N^d}$ as 
\begin{equation*}
\nabla_j=\underbrace{I\otimes \cdots \otimes I}_{d-j}\otimes\nabla \otimes \underbrace{I\otimes \cdots \otimes I}_{j-1},
\end{equation*}
where $\nabla$ is the one-dimensional discrete gradient operator and $I\in\mathbb{C}^{N\times N}$ is the identity matrix. 
When $d \geq 2$, there is more than one way to define the TV semi-norm. We define the $d$-dimensional \textit{isotropic} discrete gradient operator as
\begin{equation}
\label{idgo}
\nabla:\mathbb{C}^{N^d}\to \mathbb{C}^{N^d\times d},\ x \mapsto \nabla x=\left(\begin{array}{cccc} \nabla_1x & \cdots&\nabla_dx \end{array} \right).
\end{equation}
The $d$-dimensional \textit{isotropic} TV semi-norm is 
$
\nm{x}_{\mathrm{\mathrm{TV}}_i}=\nm{\nabla x}_{\ell^{2,1}},
$
where  $\nabla x\in\mathbb{C}^{N^d\times d}$ is as in \eqref{idgo}.
Alternatively, the $d$-dimensional \textit{anisotropic} discrete gradient operator is
\begin{equation}
\label{adgo}
\nabla:\mathbb{C}^{N^d}\to \mathbb{C}^{d N^d},\ x \mapsto \nabla x=\left( \begin{array}{c} \nabla_1x\\   \vdots \\ \nabla_dx \end{array} \right),
\end{equation}
and $d$-dimensional \textit{anisotropic} TV semi-norm is 
$
\nm{x}_{\mathrm{\mathrm{TV}}_a}= \nm{\nabla x}_{\ell^1},
$
where  $\nabla x\in\mathbb{C}^{d N^d}$ is as in \eqref{adgo}.
Notice that these semi-norms are equivalent up to a constant:
\be{
\label{TVairelate}
\nm{x}_{\mathrm{TV}_i} \leq \nm{x}_{\mathrm{TV}_a} \leq \sqrt{d} \nm{x}_{\mathrm{TV}_i} .
}

\subsection{TV minimization problem}

Let $X \in \bbC^{N \times \cdots \times N}$ be an image, $x \in \bbC^{N^d}$ be its vectorization, $A \in \bbC^{m \times N^d}$ be a measurement matrix, as defined above, and $y = A x +e$ be noisy measurements. To recover $x$ from $y$, we consider the constrained TV minimization problem
\be{
\label{TVminprob}
\min_{z \in \bbC^{N^d}} \nm{z}_{\mathrm{\mathrm{TV}}}\ \mbox{subject to $\nm{A z - y}_{\ell^2} \leq \eta$},
}
where $\eta \geq \nm{e}_{\ell^2}$ is an upper bound on the noise level, and $\nm{\cdot}_{\mathrm{\mathrm{TV}}}$ denotes either the isotropic or anisotropic TV norm.  We write $\hat{x} \in \bbC^{N^d}$ for a minimizer of this problem, which is the reconstruction of $x$, and $\widehat{X} \in \bbC^{N \times \cdots \times N}$ for the corresponding reconstruction of $X$.

\subsection{Gradient sparsity and best $s$-term approximation}

In what follows, we derive conditions on $\Omega$ and $m$ under which the error $\nm{x - \hat{x}}_{\ell^2} = \nmu{X - \widehat{X}}_{\ell^{2,2}}$ satisfies a bound depending on the gradient sparsity of the image. To this end, we define the $\ell^1$-norm best $s$-term approximation error of a vector $x \in \bbC^N$ as
\bes{
\sigma_s(x)_{\ell^{1}}=\min\{\nm{x-z}_{\ell^{1}}: \nm{z}_{\ell^{0}}\leq s\}.
}
Similarly, we define the $\ell^{2,1}$-norm best $s$-term approximation error of a matrix $X \in \bbC^{N \times M}$ as
\begin{equation*}
\sigma_s(X)_{\ell^{2,1}}=\min\{\nm{X-Z}_{\ell^{2,1}}: \nm{Z}_{\ell^{2,0}}\leq s\}.
\end{equation*}

\section{Main results on Fourier sampling}\label{s:Fourmain}

We now present our main results on Fourier sampling. We consider Walsh sampling in \S \ref{s:BeyondFourierTV}.

\subsection{Uniform random Fourier sampling}
Based on \cite{PoonTV}, we first consider uniform random Fourier sampling, as in Definition \ref{d:unifsamppatt}. For reasons that will become clear, we separate our results into the $d = 1$ and $d \geq 2$ cases:

\thm{
[Uniform Fourier sampling, one dimension]
\label{t:TVuniform1D}
Let $d = 1$, $0 < \varepsilon < 1$, $2 \leq s , m \leq N$ and $\Omega = \Omega_1 \cup \Omega_2$, where $\Omega_1 \subseteq \{-N/2+1,\ldots,N/2\}$ is a uniform random sampling scheme of order $m-1$ and $\Omega_2 = \{ 0\}$.  Let $A = \frac{1}{\sqrt{m}}P_{\Omega} F\in \bbC^{m \times N}$ and
\bes{
m \gtrsim s \cdot \log(s) \cdot \left ( \log(s) \cdot \log(N) + \log(\varepsilon^{-1}) \right ).
}
Then the following holds with probability at least $1-\varepsilon$.  For all $x \in \bbC^N$ and $y = A x + e \in \bbC^m$, where $\nm{e}_{\ell^2} \leq \eta$ for some $\eta \geq 0$, every minimizer $\hat{x}$ of \eqref{TVminprob} satisfies
\be{
\label{TVUnif1Dgraderr}
\nm{\nabla x - \nabla \hat{x}}_{\ell^2} \lesssim \frac{\sigma_s \left ( \nabla x \right )_{\ell^1}}{\sqrt{s}}  + \eta,\quad 
\nm{x -\hat{x}}_{\mathrm{TV}} \lesssim \sigma_s \left ( \nabla x \right )_{\ell^1}  + \sqrt{s} \eta,
}
and
\be{
\label{TVUnif1Dsigerr}
\frac{\nm{x - \hat{x}}_{\ell^2}}{\sqrt{N}} \lesssim \sigma_s \left ( \nabla x \right )_{\ell^1} + \sqrt{s} \eta.
}
}

\begin{theorem}
[Uniform Fourier sampling, $d \geq 2$ dimensions]
\label{t:TVuniformdD}
Let, $d \geq 2$, $0<\varepsilon<1$, $2\leq s,m\leq N^d$ and $\Omega=\Omega_1\cup\Omega_2$, where $\Omega_1$ is a $d$-dimensional uniform random sampling map of order $m-1$ and $\Omega_2=\{(0,0,\ldots,0)\}$. Let $A=\frac{1}{\sqrt{m}}P_{\Omega}F\in\mathbb{C}^{m\times N^d}$ and 
\begin{equation}
m\gtrsim s\cdot \log(s) \cdot (d\cdot \log(s)\cdot\log(N)+\log(\varepsilon^{-1})),
\end{equation}
Then the following holds with probability at least $1-\varepsilon$.  For all $x \in \bbC^{N^d}$ and $y = A x + e \in \bbC^m$, where $\nm{e}_{\ell^2} \leq \eta$ for some $\eta \geq 0$, every minimizer $\hat{x}$ of \eqref{TVminprob} satisfies
\begin{equation}
\label{uge2}
\nm{\nabla x - \nabla \hat{x}}_{\ell^2} \lesssim \frac{\sigma_s(\nabla x)_{\ell^1}}{\sqrt{s}}+d\eta,\quad \nm{x - \hat{x}}_{\mathrm{\mathrm{TV}}_a} \lesssim \sigma_s(\nabla x)_{\ell^1} +\sqrt{s} d\eta \quad \mathrm{(anisotropic)},
\end{equation}
\begin{equation}
\label{uge22}
\nmu{\nabla\hat{x}-\nabla x}_{\ell^{2,2}}\lesssim \frac{\sigma_s(\nabla x)_{\ell^{2,1}}}{\sqrt{s}}+\sqrt{d}\eta ,\quad \nm{x - \hat{x}}_{\mathrm{\mathrm{TV}}_i} \lesssim \sigma_s(\nabla x)_{\ell^{2,1}} +\sqrt{s}\sqrt{d} \eta
\quad \mathrm{(isotropic)},
\end{equation}
and
\begin{equation}
\label{unia}
\nmu{\hat{x}-x}_{\ell^2}\lesssim 2^{-d/2}\sigma_s(\nabla x)_{\ell^1}+(1+2^{-d/2}\sqrt{s}d)\eta\quad \mathrm{(anisotropic)},
\end{equation}
\begin{equation}
\label{unii}
\nmu{\hat{x}-x}_{\ell^2}\lesssim2^{-d/2}\sqrt{d}\sigma_s(\nabla x)_{\ell^{2,1}}+(1+2^{-d/2}\sqrt{s}d)\eta\quad \mathrm{(isotropic)}.
\end{equation}
\end{theorem}

These results assert recovery of $x$ from roughly $s  \cdot \log^2(s) \cdot \log(N)$ measurements for fixed $d$, i.e.\ linear in $s$ up to the log factors. The gradient error bound in the $\ell^2$-norm (or $\ell^{2,2}$-norm in the case of the anisotropic TV semi-norm) is the typical stable and robust recovery guarantee found ubiquitously in compressed sensing \cite{FoucartRauhutCSbook}. Specifically, the error depends on a best $s$-term approximation error $\sigma_{s}(\nabla x)_{\ell^1} / \sqrt{s}$ (stability) and the noise level $\eta$ (robustness). 

Conversely, the recovery of the image $x$ is worse by a factor of $\sqrt{s}$ than the recovery of its gradient -- compare, for example, \eqref{uge2} with \eqref{unia}. As observed previously in \cite{PoonTV}, this is due to the choice of a uniform random sampling sampling scheme. In the next section we improve the stability and robustness of the image recovery by adding samples drawn from a variable density. We remark in passing that the one-dimensional signal recovery bound \eqref{TVUnif1Dsigerr} involves a factor of $1/\sqrt{N}$. This factor is natural when considering $x$ as the discretization of a continuous image \cite[Rem.\ 2.1]{PoonTV}.

As noted, nonuniform recovery guarantees for uniform random Fourier sampling were shown in \cite{PoonTV}. In one dimension, \cite[Thm.\ 2.3]{PoonTV} asserts that
\bes{
m \gtrsim s \cdot \log(N) \cdot (1 + \log(\varepsilon^{-1}) ),
}
measurements are sufficient for an error bound of the form
\bes{
\frac{\nm{x - \hat{x}}_{\ell^2}}{\sqrt{N}} \lesssim \log^{1/2}(m) \log(s) \sigma_{s}(\nabla x)_{\ell^1} + \eta \sqrt{s}  .
}
Our uniform recovery guarantee (Theorem \ref{t:TVuniform1D}) imposes a higher sample complexity (by a factor of $\log^2(s)$), but obtains an improved error bound \eqref{TVUnif1Dsigerr}, in which no log factors appear. The same comparison can be made in $d = 2$ dimensions. See Theorem \ref{t:TVuniformdD} and \cite[Thm.\ 2.4]{PoonTV}.

\subsection{Variable density Fourier sampling}

Using an idea of \cite{PoonTV}, we now consider a sampling strategy where the uniform random samples (which are sufficient to recover the gradient stably and robustly) are augmented by a set of variable density Fourier samples to enhance the image recovery. Following Definition \ref{d:VDsamppatt}, let $p = (p_{\omega})$ be a probability distribution on $\{ - N/2+1,\ldots,N/2\}^d$. We also require several additional concepts. First, if $\omega \in \bbR$, we let $\overline{\omega} = \max \{ 1 , | \omega | \}$. Second, if $\omega = (\omega_1,\ldots,\omega_d) \in \bbR^d$, we let $\pi : \{1,\ldots,d\} \rightarrow \{1,\ldots,d\}$ be a bijection such that $\overline{\omega_{\pi(1)}} \geq \overline{\omega_{\pi(2)}} \geq \ldots \geq \overline{\omega_{\pi(d)}}$. Next, we define $q = (q_{\omega} )$ by
\be{
\label{qomdef1}
q_{\omega} =\overline{\omega_{\pi(1)}} \cdots \overline{\omega_{\pi(d/2)}} ,\qquad  \mbox{$d$ even},
}
and
\be{
\label{qomdef2}
q_{\omega} =\overline{\omega_{\pi(1)}} \cdots \overline{\omega_{\pi((d-1)/2)}} \sqrt{\overline{\omega_{\pi((d+1)/2)}}} ,\qquad \mbox{$d$ odd}.
}
Finally, we let $\Gamma(p)$ be the smallest positive constant such that
\be{
\label{Cpdef}
(q_{\omega})^{-2} \leq \Gamma(p) p_{\omega},\qquad \forall \omega \in \{ - N/2+1,\ldots,N/2\}^d.
}
Notice that $\Gamma(p) \geq 1$, since $p$ is a probability distribution and $q_{0} =1$.

\thm{[Variable density Fourier sampling, one dimension]
\label{t:TVVDS1D}
Let $d = 1$, $0 < \varepsilon < 1$, $2 \leq s , m \leq N$ and $\Omega = \Omega_1 \cup \Omega_2 \subseteq \{-N/2+1,\ldots,N/2\}$, where $\Omega_1$ is a uniform random sampling scheme of order $m/2$ and $\Omega_2$ is a variable density sampling scheme of order $m/2$ corresponding to a probability distribution $p = (p_{\omega})$.  Let $A = \frac{1}{\sqrt{m}}P_{\Omega} F\in \bbC^{m \times N}$ and 
\be{
\label{m1DVDS}
m \gtrsim \Gamma(p) \cdot s \cdot \log(\Gamma(p) s) \cdot \left ( \log(\Gamma(p) s) \cdot \log(N) + \log(2\varepsilon^{-1}) \right ).
}
Then the following holds with probability at least $1-\varepsilon$.  For all $x \in \bbC^N$ and $y = A x + e \in \bbC^m$, where $\nm{e}_{\ell^2} \leq \eta$ for some $\eta \geq 0$, every minimizer $\hat{x}$ of \eqref{TVminprob} satisfies
\be{
\label{TVVDS1Dgraderr}
\nm{\nabla x - \nabla \hat{x}}_{\ell^2} \lesssim \frac{\sigma_s \left ( \nabla x \right )_{\ell^1}}{\sqrt{s}} + \eta ,\quad
\nm{ x -  \hat{x}}_{\mathrm{TV}} \lesssim \sigma_s \left ( \nabla x \right )_{\ell^1}  + \sqrt{s} \eta,
}
and
\be{
\label{TVVDS1Dsigerr}
\frac{\nm{x - \hat{x}}_{\ell^2}}{\sqrt{N}} \lesssim \frac{\sigma_{s}(\nabla x)_{\ell^1}}{s} + \left ( \sqrt{\frac{\Gamma(p)}{N}} + \frac{1}{\sqrt{s}} \right ) \eta.
}
}

\begin{theorem}
[Variable density Fourier sampling, $d \geq 2$ dimension]
\label{t:TVVDSdD}
Let, $d \geq 2$, $0<\varepsilon<1$, $2\leq s,m\leq N^d$ and $\Omega=\Omega_1\cup\Omega_2$, where $\Omega_1$ is a $d$-dimensional uniform random sampling pattern of order $m/2$ and $\Omega_2$ is a variable density sampling scheme of order $m/2$ corresponding to a probability distribution $p = (p_{\omega})$. Let $A = \frac{1}{\sqrt{m}}P_{\Omega} F\in \bbC^{m \times N^d}$ and 
\begin{equation}
\label{nomvds}
m \gtrsim_d \Gamma(p) \cdot s \cdot \log^2(N) \cdot \log(\Gamma(p) \log(N) s) \cdot \left ( \log(\Gamma(p) \log(N) s) \cdot \log(N) + \log(2\epsilon^{-1}) \right ),
\end{equation}
where $\Gamma(p)$ is as in \eqref{Cpdef}. Then the following holds with probability at least $1-\varepsilon$.  For all $x \in \bbC^{N^d}$ and $y = A x + e \in \bbC^m$, where $\nm{e}_{\ell^2} \leq \eta$ for some $\eta \geq 0$, every minimizer $\hat{x}$ of \eqref{TVminprob} satisfies
\begin{equation}
\label{grani}
\nm{\nabla x - \nabla \hat{x}}_{\ell^2} \lesssim \frac{\sigma_s(\nabla x)_{\ell^1}}{\sqrt{s}}+d\eta,\quad \nm{x - \hat{x}}_{\mathrm{\mathrm{TV}}_a} \lesssim \sigma_s(\nabla x)_{\ell^1} +\sqrt{s} d\eta \quad \mathrm{(anisotropic)},
\end{equation}
\begin{equation}
\label{gri}
\nmu{\nabla\hat{x}-\nabla x}_{\ell^{2,2}}\lesssim \frac{\sigma_s(\nabla x)_{\ell^{2,1}}}{\sqrt{s}}+\sqrt{d}\eta ,\quad \nm{x - \hat{x}}_{\mathrm{\mathrm{TV}}_i} \lesssim \sigma_s(\nabla x)_{\ell^{2,1}} +\sqrt{s}\sqrt{d} \eta
\quad \mathrm{(isotropic)},
\end{equation}
and
\begin{equation}
\label{srani}
\nmu{\hat{x}-x}_{\ell^2}\lesssim \frac{\sigma_s(\nabla x)_{\ell^1}}{\sqrt{s}} + \left (\sqrt{\Gamma(p)}+d \right )\eta \quad \mathrm{(anisotropic)},
\end{equation}
\begin{equation}
\label{sri}
\nmu{\hat{x}-x}_{\ell^2}\lesssim \sqrt{d} \frac{\sigma_s(\nabla x)_{\ell^{2,1}}}{\sqrt{s}} + \left (\sqrt{\Gamma(p)}+d \right ) \eta \quad \mathrm{(isotropic)}.
\end{equation}
\end{theorem}

These results are general in the sense that they permit any variable density sampling scheme. Moreover, the effect of the density $p$ is seen clearly through the constant $\Gamma(p)$: the smaller $\Gamma(p)$, the better the measurement conditions \eqref{m1DVDS} and \eqref{nomvds} and the image recovery bounds \eqref{TVVDS1Dsigerr}, \eqref{srani} and \eqref{sri}. In the next section, we discuss the choice of $p$. Specifically, we identify densities for which $\Gamma(p)$ satisfies the optimal bound $\Gamma(p) \lesssim \log_d(N)$.

With this in mind, these results can be understood as follows. Suppose that $p$ is chosen so that $\Gamma(p) \lesssim_d \log(N)$. Then by incorporating variable density samples we achieve better stability and robustness in the image recovery by a factor of $\sqrt{s}$ over the case when only uniform random samples are used (Theorems \ref{t:TVuniform1D} and \ref{t:TVuniformdD}). In particular, the image error bounds, up to the factor of $\Gamma(p)$, depend on $\sigma_{s}(\nabla x)_{\ell^1} / \sqrt{s}$ and $\eta$, exactly as in the gradient error bounds. Moreover, to achieve these estimates we need a number of measurements scaling linearly in $s$, up to log factors. We note also that the anisotropic and isotropic TV semi-norms give the same recovery guarantees, up to factors in $d$.

\subsection{Discussion}

To illustrate this difference, in Fig.\ \ref{f:GradRecSNR} we compare the stability and robustness of the recovery of a two-dimensional image and its gradient. 
In this figure, we perturb either the image $x$ (to study stability) or the measurements $y$ (to study robustness) and compute the error in the reconstructed image and its gradient. 
We use the standard Shepp--Logan phantom, since its gradient is exactly sparse, and compare the recovery from uniform random and variable density samples.

For both types of perturbations, observe that the image recovery error is better for variable density samples than uniform random samples, whereas the gradient recovery errors are very similar.  This confirms the results of the previous section, which assert that uniform random sampling provides adequate recovery of the image gradient, matching the stability and robustness of variable density sampling, but that the image recovery error is worse by a factor of $\sqrt{s}$.

\begin{figure}
{\small 
\begin{center}
\begin{tabular}{@{\hspace{0pt}}c@{\hspace{10pt}}c@{\hspace{10pt}}c@{\hspace{10pt}}c@{\hspace{0pt}}}
Stability (image) & Stability (gradient) & Robustness (image) & Robustness (gradient)
\\[1pt]
\includegraphics[width=0.16\paperwidth,clip=true,trim=0mm 0mm 0mm 0mm
]{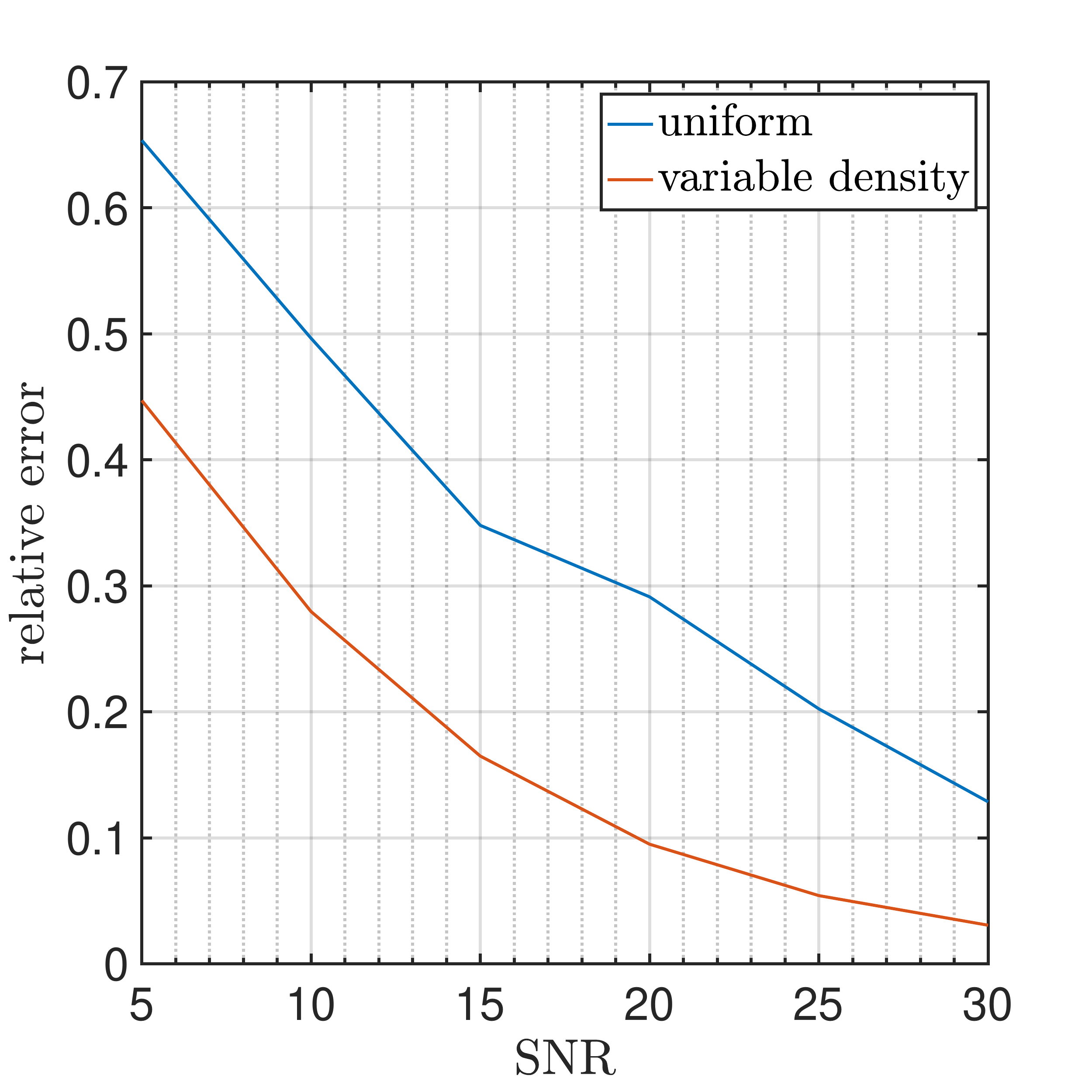}
&
\includegraphics[width=0.16\paperwidth,clip=true,trim=0mm 0mm 0mm 0mm
]{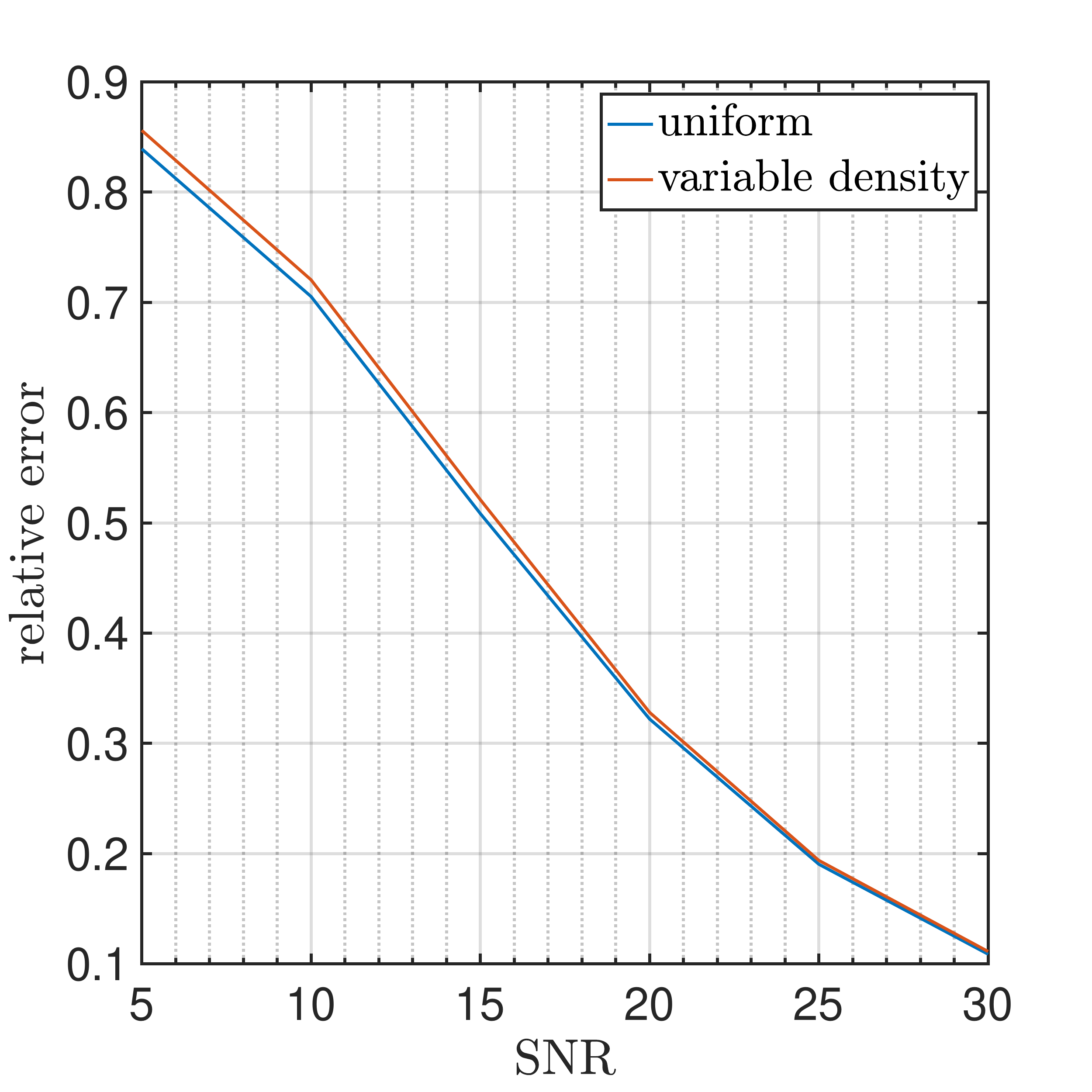}
& 
\includegraphics[width=0.16\paperwidth,clip=true,trim=0mm 0mm 0mm 0mm
]{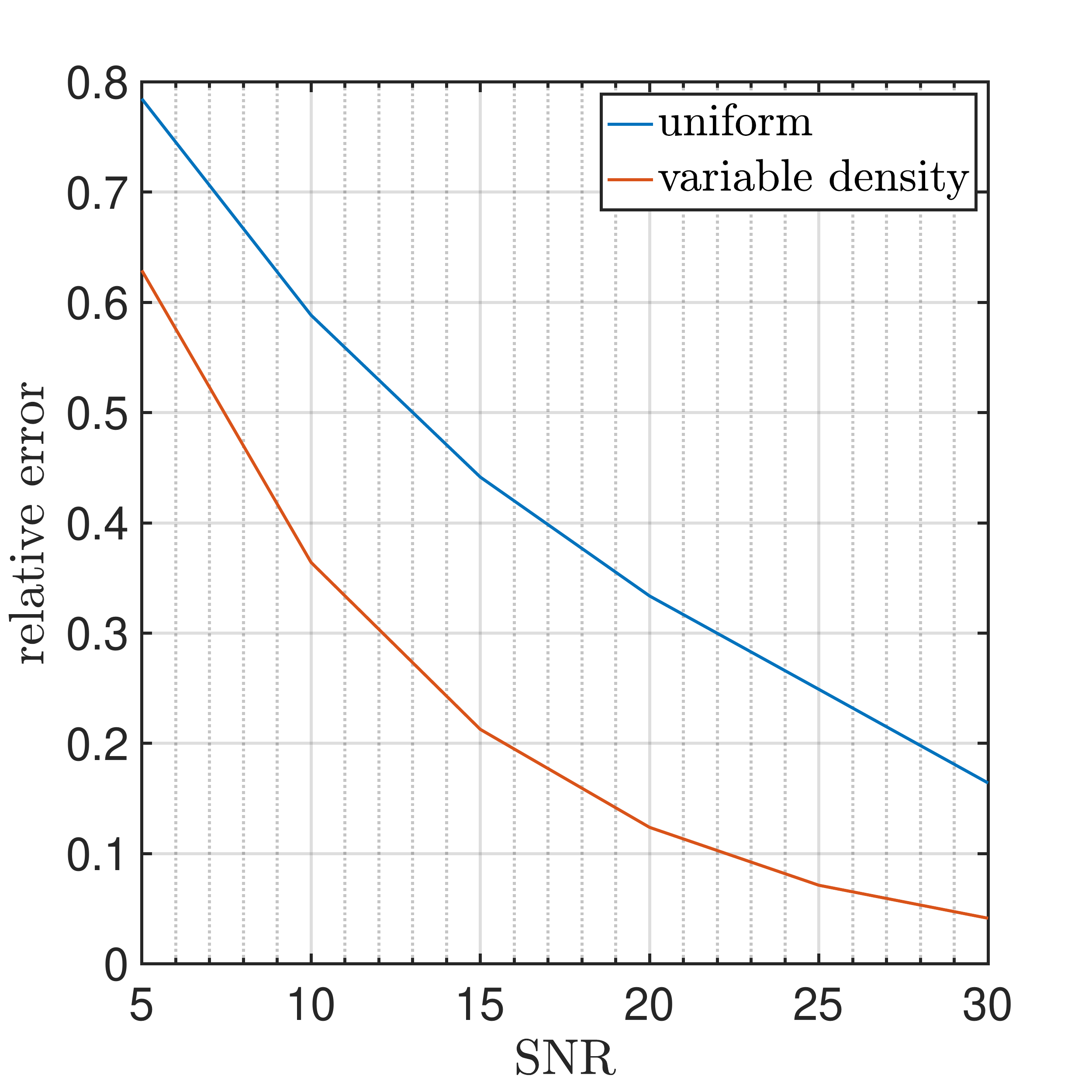}
&
\includegraphics[width=0.16\paperwidth,clip=true,trim=0mm 0mm 0mm 0mm
]{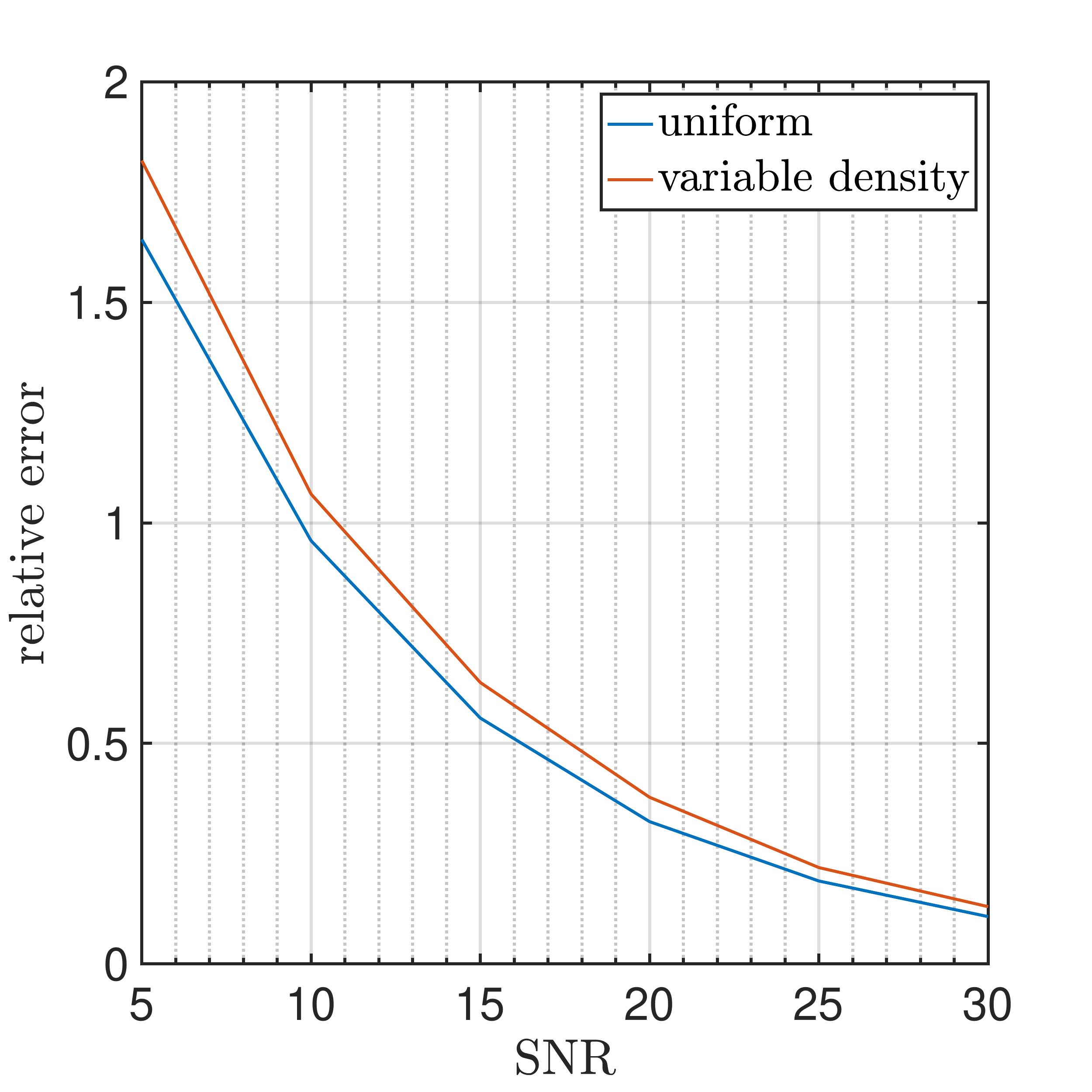}
 \end{tabular}
 \end{center}
}
\vspace{-2mm}
\caption{Recovery of the discrete $256^2$ Shepp--Logan phantom from $25\%$ Fourier measurements using either uniform random sampling or variable density sampling according to \eqref{thopt2Dalt1}.  The horizontal axis shows the signal-to-noise ratio (SNR) of the perturbation and the vertical axis shows the relative error in the recovered image or recovered image gradient.
For the stability experiment (left), the image $x$ is perturbed to $x+h$.  The SNR and relative error are defined as $20 \log_{10}(\nm{x}_{\ell^2} / \nm{h}_{\ell^2})$ and $\nm{z - (x+h)}_{\ell^2} / \nm{x+h}_{\ell^2}$ or $\nm{\nabla(z - (x+h))}_{\ell^2} / \nm{\nabla(x+h)}_{\ell^2}$ respectively, where $z$ is the reconstruction of $x+h$.  For the robustness experiment (right), the measurements $y$ are perturbed to $y + h$.  The SNR and relative error are defined as $20 \log_{10}(\nm{y}_{\ell^2}/\nm{h}_{\ell^2})$ and $\nm{z - x}_{\ell^2} / \nm{x}_{\ell^2}$ or $\nm{\nabla(z - x)}_{\ell^2} / \nm{\nabla(x)}_{\ell^2}$  respectively, where $z$ is the reconstruction obtained from measurements $y + h$.} 
\label{f:GradRecSNR}
\end{figure}

The intuition for this discrepancy is quite straightforward. Since it has periodic boundary conditions, the gradient operator commutes with the DFT matrix (see Lemma \ref{l:commuting}). Hence recovery of the image gradient is equivalent to recovering a sparse vector from samples of its Fourier transform. It is well known that uniform random sampling is a suitable (in fact, optimal) sampling strategy for recovering a sparse vector from samples of its Fourier transform. Hence, we expect adequate recovery of the gradient from such measurements. 
On the other hand, since the constant vector lies in the null space of the gradient operator, it is impossible to recover $x$ from $\nabla x$. This is why the zero frequency is added in Theorems \ref{t:TVuniform1D} and \ref{t:TVuniformdD}. However, the stability and robustness of the image recovery is worse, since the gradient operator $\nabla$ is ill-conditioned for large $N$. In particular, smooth functions (i.e.\ image textures) lie approximately in its null space. Yet, the Fourier transform of a smooth function decays rapidly with increasing frequency. Hence, variable density sampling overcomes this issue by sampling more densely near the origin, thus stabilizing the recovery of the smooth image components.

\section{Choice of Fourier sampling pattern}\label{s:Foursamp}

As noted above, Theorems \ref{t:TVVDS1D} and \ref{t:TVVDSdD} allow for any variable density sampling scheme. We now discuss this choice in more detail.

\subsection{Theoretically-optimal sampling patterns}

We commence by deriving sampling patterns that are theoretically optimal, in the sense that they give the optimal scaling of $\Gamma(p)$ with respect to $N$ (for fixed $d$):

\lem{
\label{l:optimalrule}
Let $p = (p_\omega)$ be a probability distribution and $\Gamma(p)$ be as in \eqref{Cpdef}. Then $\Gamma(p) \gtrsim \log(N)$. Moreover, if
\bes{
p_{\omega} =  \frac{C_{N,d}}{(q_{\omega})^2},\qquad \omega \in \{-N/2+1,\ldots,N/2\}^d,
}
where $q_{\omega}$ is as in \eqref{qomdef1}--\eqref{qomdef2}, then $\Gamma(p) \lesssim_d \log(N)$. 
}

Using this, we immediately deduce the following:

\cor{
[Theoretically-optimal variable density Fourier sampling, one dimension]
\label{c:thopt1D}
Consider the setup of Theorem \ref{t:TVVDS1D} with $p = (p_{\omega})$ given by
\be{
\label{1Doptdensity}
p_{\omega} = \frac{C_N}{\max \{ 1 , |\omega| \} },\qquad \omega \in \{ -N/2+1,\ldots,N/2\},
}
and $s \gtrsim \log(N)$. Then the conclusions of Theorem \ref{t:TVVDS1D} hold (with $\Gamma(p) \lesssim \log(N)$ in the case of \eqref{TVVDS1Dsigerr}), provided $m$ satisfies
\be{
\label{thopt1Dmeascond}
m \gtrsim s \cdot \log(s) \cdot \log(N) \cdot \left ( \log(s) \cdot \log(N) + \log(2 \varepsilon^{-1}) \right ).
}
}

Note that the condition $s \gtrsim \log(N)$ is imposed merely to simplify the measurement condition (it allows one to replace terms such as $\log(\log(N) s)$ by $\log(s)$). 
It is informative to compare this result with Theorem \ref{t:TVuniform1D}. The measurement condition \eqref{thopt1Dmeascond} prescribes an additional $\log(N)$ samples over Theorem \ref{t:TVuniform1D}, taken according to the density \eqref{1Doptdensity}. However, this leads to an improved signal recovery error of the form
\be{
\label{elkisland}
\frac{\nmu{\hat{x} - x}_{\ell^2}}{\sqrt{N}} \lesssim \frac{\sigma_{s}(\nabla x)_{\ell^1}}{s} + \left ( \sqrt{\frac{\log(N)}{N}} + \frac{1}{\sqrt{s}} \right ) \eta.
}
Note that a nonuniform recovery guarantee of similar flavour to Corollary \ref{c:thopt1D} was first proved in \cite[Thm.\ 2.1]{PoonTV}. Therein
$
m \gtrsim s \cdot \log(N) \cdot ( 1 + \log(\varepsilon^{-1}))
$
samples taken in the same way (in particular, with the same variable density \eqref{1Doptdensity}) were shown to give a recovery error
\bes{
\frac{\nmu{\hat{x} - x}_{\ell^2}}{\sqrt{N}} \lesssim \log^2(s) \log(N) \log(m) \left ( \log(s) \log^{1/2}(m) \frac{\sigma_{s}(\nabla x)_{\ell^1}}{s} + \frac{1}{\sqrt{s}}  \eta \right ).
}
By contrast, Corollary \ref{c:thopt1D} is a uniform recovery guarantee. While it imposes a more stringent measurement condition \eqref{thopt1Dmeascond}, specifically, by a factor of $\log^2(s) \log(N)$, it leads to an improved recovery guarantee \eqref{elkisland}. For instance, the best $s$-term approximation error term $\sigma_s(\nabla x)_{\ell^1} / s$ is improved by a factor of $\log^3(s) \log(N) \log^{3/2}(m)$.

\cor{
[Theoretically-optimal variable density Fourier sampling, two dimensions]
\label{c:thopt2D}
Let $d = 2$ and consider the setup of Theorem \ref{t:TVVDSdD} with $p = (p_{\omega})$ given by
\be{
\label{thopt2D}
p_{\omega} = \frac{C_N}{\left (\max \{ 1 , |\omega_1 | , |\omega_2 | \}\right)^{2}},\qquad \omega = (\omega_1,\omega_2)  \in \{ -N/2+1,\ldots,N/2\}^{2},
}
and $s \gtrsim \log(N)$.
Then the conclusions of Theorem \ref{t:TVVDSdD} hold (with $\Gamma(p) \lesssim \log(N)$ in the case of \eqref{srani} and \eqref{sri}), provided $m$ satisfies
\be{
\label{thopt2Dmeascond}
m \gtrsim s \cdot \log(s) \cdot \log^3(N) \cdot \left ( \log(s) \cdot \log(N) + \log(2 \varepsilon^{-1}) \right )
}
Furthermore, the same conclusion holds (with possibly different numerical constant) if \eqref{thopt2D} is replaced by
\be{
\label{thopt2Dalt1}
p_{\omega} = \frac{C_N}{1+(\omega_1)^2 + (\omega_2)^2 },\qquad \omega = (\omega_1,\omega_2) \in \{ -N/2+1,\ldots,N/2\}^{2},
}
or more generally, if $\nm{\cdot}$ is any norm on $\bbR^2$, by
\be{
\label{thopt2Dalt2}
p_{\omega} = \frac{C_N}{1+\nm{\omega}^2},\qquad  \omega \in \{ -N/2+1,\ldots,N/2\}^{2}.
}
}

Note that \eqref{thopt2D} follows immediately from the observation that $q_{\omega} = \max \{ 1 , |\omega_1| , |\omega_2 | \}$ when $d = 2$. The results for \eqref{thopt2Dalt1} and \eqref{thopt2Dalt2} follow in turn simply because of the equivalence of norms on a finite-dimensional vector space.

The scheme \eqref{thopt2Dalt1} is known as \textit{inverse square law} sampling. It is a standard and well-known variable density sampling strategy for compressed sensing recovery from Fourier measurements \cite{KrahmerWardCSImaging,PoonTV}. Interesting, this result also shows that there are many different sampling strategies that give the same recovery guarantees up to constants. The critical factor is the asymptotic decay rate as $\omega \rightarrow \infty$. Fig.\ \ref{f:2d_level_curves} visualizes the level curves of several such sampling strategies. Notice that the schemes \eqref{thopt2Dalt2} depend on the distance of $\omega$ from the zero frequency (with respect to some norm). We therefore informally refer to them as \textit{radially symmetric}.

\begin{figure}[t]
\includegraphics[width=0.23\paperwidth,clip=true,trim=0mm 0mm 0mm 0mm
]{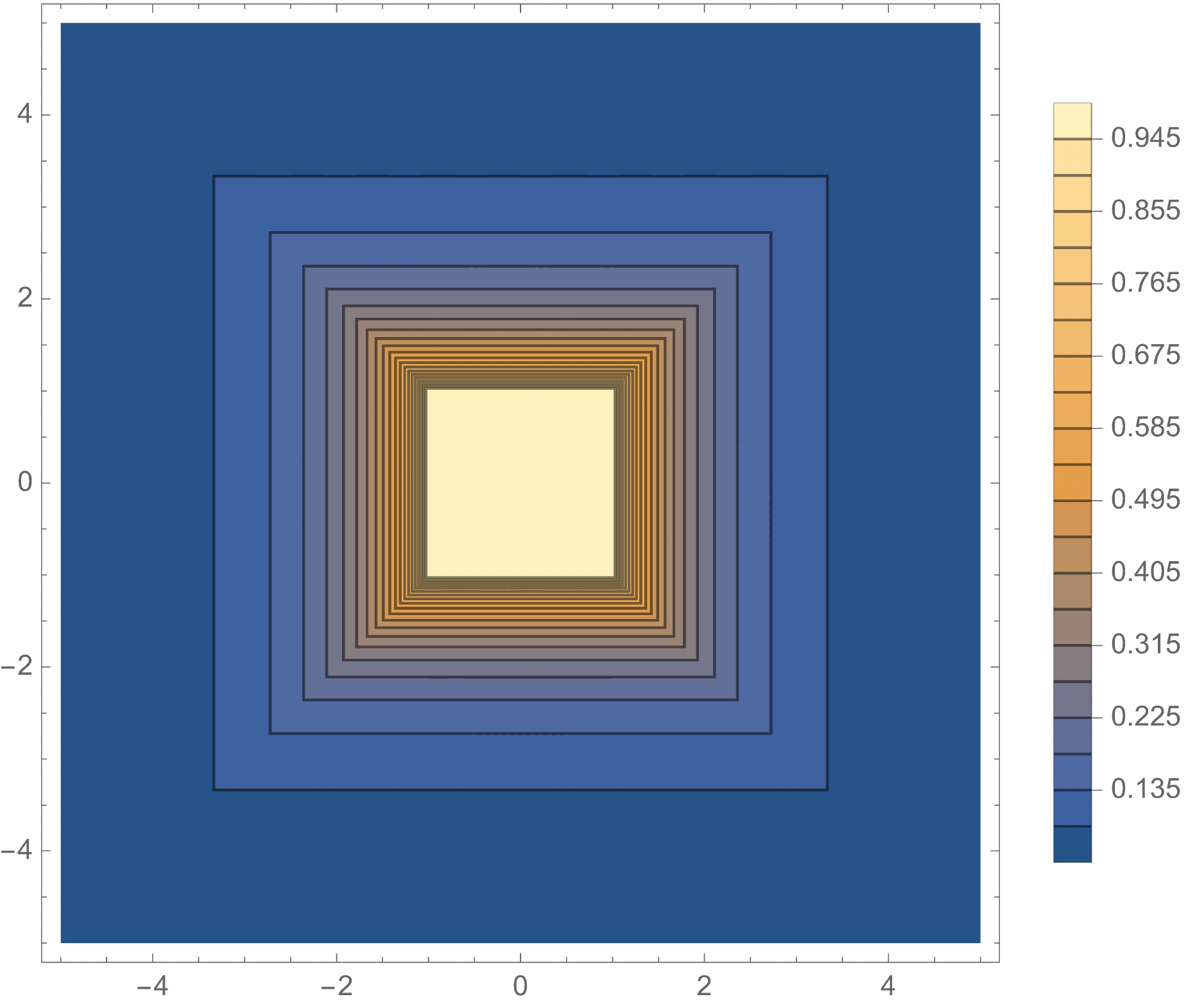}
\includegraphics[width=0.23\paperwidth,clip=true,trim=0mm 0mm 0mm 0mm
]{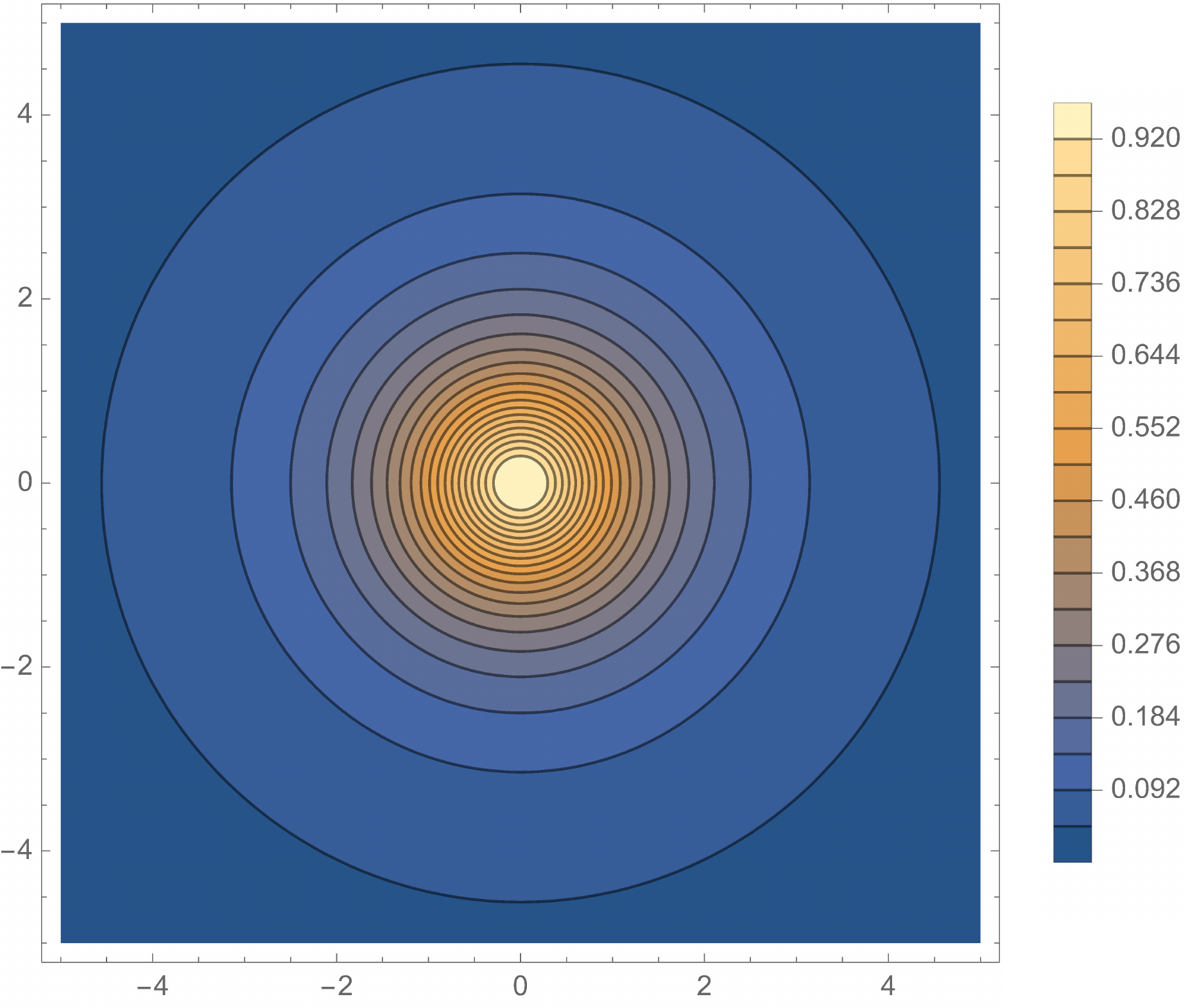}
\includegraphics[width=0.23\paperwidth,clip=true,trim=0mm 0mm 0mm 0mm
]{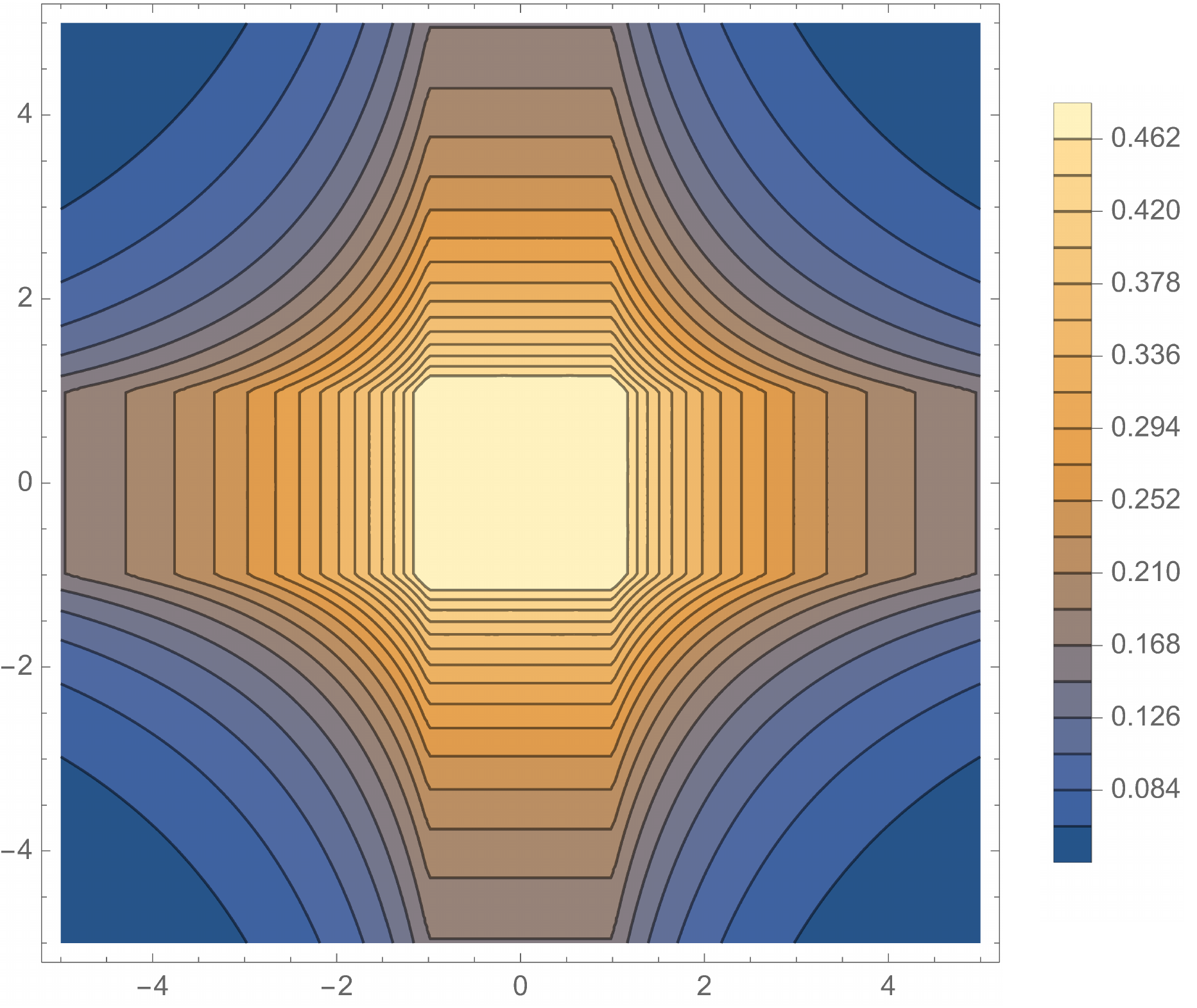}
\vspace{-2mm}
\caption{Level curves for the 2D {\bf(left)} theoretically optimal \eqref{thopt2D}, {\bf(middle)} inverse square \eqref{thopt2Dalt1} and {\bf(right)} hyperbolic cross \eqref{hypcrossdD} densities.}
\label{f:2d_level_curves}
\end{figure}

Similar results to Corollary \ref{c:thopt2D} were shown in \cite{KrahmerWardCSImaging,PoonTV}. In \cite[Thm.\ 1]{KrahmerWardCSImaging} a uniform recovery guarantee was proved for inverse square law sampling \eqref{thopt2Dalt1}, with the measurement condition
\be{
\label{KWmeascond}
m \gtrsim s \cdot \log^3(s) \cdot \log^5(N),
}
implying a image recovery bound
\bes{
\nmu{x - \hat{x}}_{\ell^2} \lesssim \frac{\sigma_{s}(\nabla x)_{\ell^1}}{\sqrt{s}} + \eta,
}
for the anisotropic TV semi-norm with a particular probability, where $\eta$ is a bound for the noise in a certain weighted $\ell^2$-norm. Corollary \ref{c:thopt2D} improves on this result in several ways. First, the log factors in the measurement condition \eqref{thopt2Dmeascond} are reduced by a factor of $\log(s) \cdot \log(N)$ over \eqref{KWmeascond}. Second, this result gives a robustness bound where the noise is measured in an unweighted $\ell^2$-norm. Third, this result establishes the same recovery guarantee for the family of sampling schemes \eqref{thopt2Dalt2}, as opposed to just the inverse square law \eqref{thopt2Dalt1}.

On the other hand, a nonuniform recovery guarantee was shown in \cite[Thm.\ 2.2]{PoonTV}. Therein
\be{
\label{poon2DTVmeas}
m \gtrsim s \cdot \log(N) \cdot ( 1 + \log(\varepsilon^{-1})),
}
taken in the same way (in particular, using the inverse square law) were shown to yield a image recovery bound of the form
\bes{
\nmu{\hat{x} - x}_{\ell^2} \lesssim\log(s) \log(N^2/s) \log^{1/2}(N) \log^{1/2}(m)  \left ( \log^{1/2}(m) \log(s) \frac{\sigma_{s}(\nabla x)_{\ell^{2,1}}}{\sqrt{s}} +  \eta \right ),
}
for the isotropic TV semi-norm. In comparison, Corollary \ref{c:thopt2D} is a uniform recovery guarantee with an image recovery error bound of the form
\bes{
\nmu{\hat{x} - x}_{\ell^2} \lesssim \frac{\sigma_{s}(\nabla x)_{\ell^{2,1}}}{\sqrt{s}} + \sqrt{\log(N)} \eta ,
}
for the isotropic TV semi-norm.
As in the one-dimensional case, the tradeoff for the worse log term in the measurement condition \eqref{thopt2Dmeascond} (by a factor of $\log^2(s) \log^3(N)$ over \eqref{poon2DTVmeas}) is a better image recovery bound by several log factors.

Finally, we consider the case of $d \geq 2$ dimensions. Note that this problem was not considered in either \cite{KrahmerWardCSImaging} or \cite{PoonTV}:

\cor{
[Theoretically-optimal variable density Fourier samples, $d \geq 2$ dimensions]
\label{c:thoptdD}
Let $d \geq 2$ and consider the setup of Theorem \ref{t:TVVDSdD} with $p = (p_{\omega})$ given by
\be{
\label{thoptdD}
p_{\omega} = \frac{C_{N,d}}{(q_{\omega})^2},\qquad \omega \in \{-N/2+1,\ldots,N/2\}^d,
}
and $s \gtrsim \log(N)$.
Then the conclusions of Theorem \ref{t:TVVDSdD} hold (with $\Gamma(p) \lesssim_d \log(N)$ in the case of \eqref{srani} and \eqref{sri}), provided $m$ satisfies
\be{
\label{thoptdDmeascond}
m \gtrsim_d s \cdot \log(s) \cdot \log^3(N) \cdot \left ( \log(s) \cdot \log(N) + \log(2 \varepsilon^{-1}) \right ).
}
In particular, when $d = 3$, \eqref{thoptdD} can be expressed as
\be{
\label{thopt3D}
p_{\omega} = C_N \left ( \left ( \max_{i=1,2,3} \{ \overline{\omega_i} \} \right )^2 \left ( \sum^{3}_{i=1} \overline{\omega}_i - \max_{i=1,2,3} \{ \overline{\omega_i} \} -\min_{i=1,2,3} \{ \overline{\omega_i} \} \right ) \right )^{-1}. 
}
}

Several remarks are in order. First, the measurement condition \eqref{thoptdDmeascond} and recovery error bounds are exactly the same as the two-dimensional measurement condition \eqref{thopt2Dmeascond} and error bounds, except possibly for $d$-dependent constants. Second, as shown by \eqref{thopt3D}, theoretically-optimal sampling strategies cease to be radially-symmetric in $d \geq 3$ dimensions. We shall discuss this further in the next section. But first, it is interesting to visualize the shape of the density \eqref{thopt3D}. Fig.\ \ref{fig:3D_pattern_plots} plots a typical level set of this function. We observe in particular the axis-aligned spikes, and the nonsmooth transitions along the edges of the cube.

\begin{figure}[t]
\begin{center}
\includegraphics[width=0.25\paperwidth,clip=true,trim=0mm 0mm 0mm 0mm]{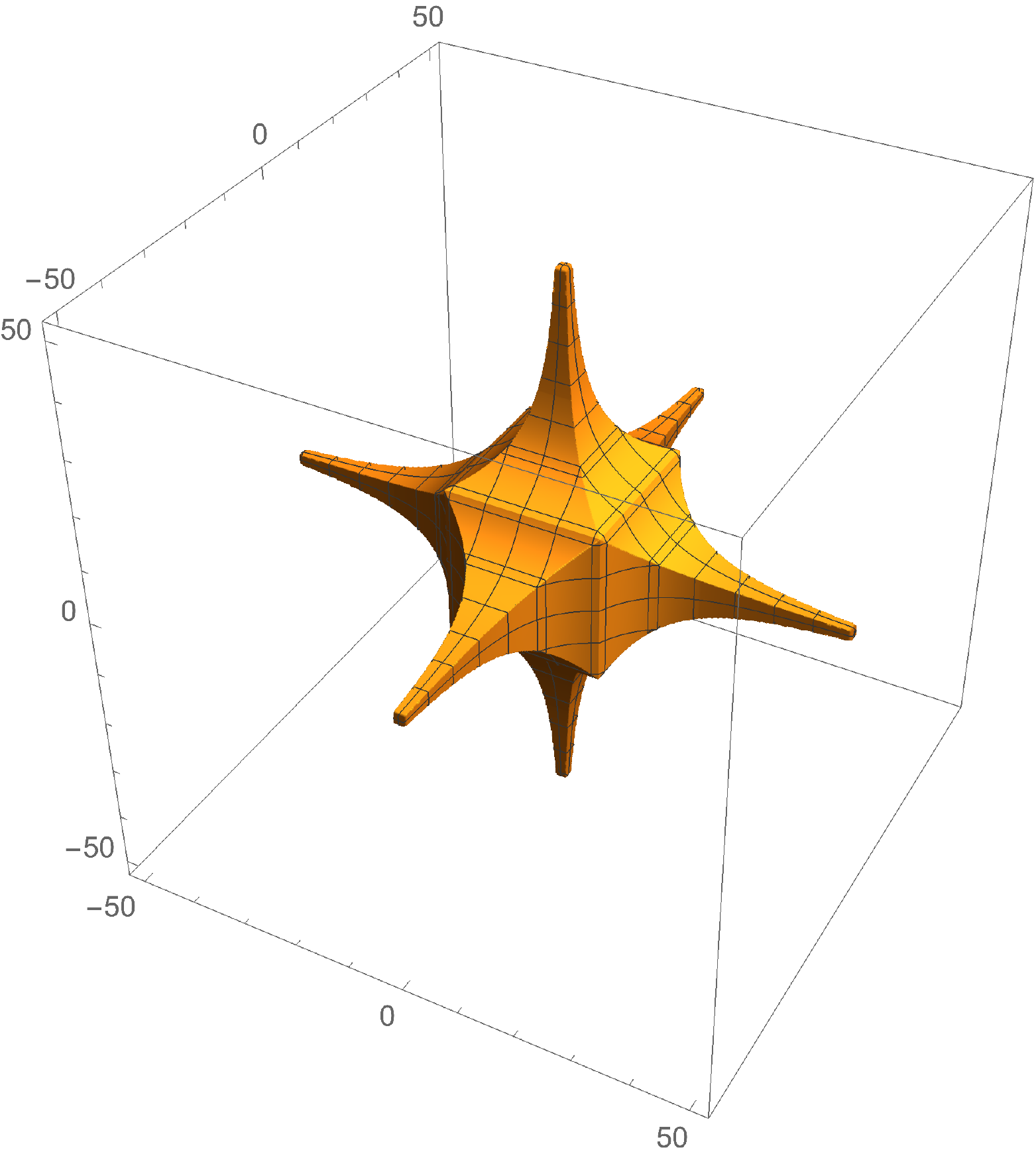}
\includegraphics[width=0.25\paperwidth,clip=true,trim=0mm 0mm 0mm 0mm]{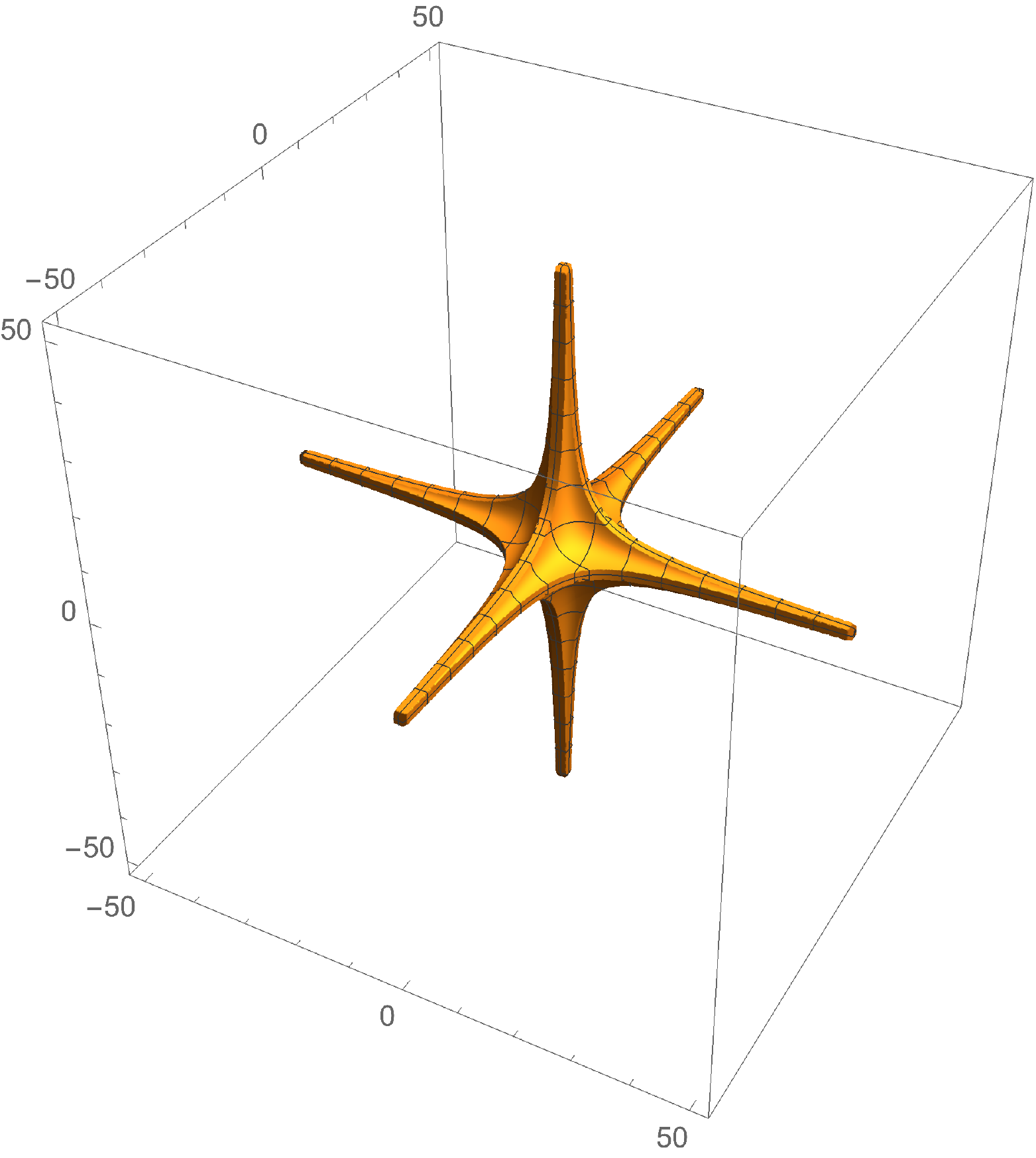}
\end{center}
\vspace{-2mm}
\caption{Level sets of the {\bf(left)} theoretically-optimal \eqref{thopt3D} and {\bf(right)} hyperbolic cross \eqref{hypcrossdD} densities in three dimensions.}
\label{fig:3D_pattern_plots}
\end{figure}

\subsection{Sub-optimality of radially-symmetric sampling}\label{ss:suboptimalradial}

As shown in Corollary \ref{c:thopt2D}, radially-symmetric sampling schemes are theoretically optimal in $d = 2$ dimensions.  We now show that this ceases to be the case when $d \geq 3$. 

\lem{
\label{l:radialsubopt}
Let $d \geq 2$ and $p = (p_{\omega})$ be defined by
\bes{
p_{\omega} = \frac{C_{N,d,\alpha}}{(1+\nm{\omega})^{\alpha}},\qquad \omega \in \{-N/2+1,\ldots,N/2\}^d,
}
where $\nm{\cdot}$ is any norm on $\bbR^d$ and $\alpha >0$. Then
\bes{
\Gamma(p) \asymp_{d,\alpha} \left \{ 
\begin{array}{cc} 
N^{d-\alpha} & \alpha < 2 \\ 
N^{d-2} & 2 \leq \alpha < d  \\
N^{d-2} \log(N) & \alpha = d \\
N^{\alpha-2}& \alpha > d 
\end{array} \right . ,
}
(note that the second case is only possible when $d \geq 3$). In particular, the best scaling for $\Gamma(p)$ is $\Gamma(p) \asymp \log(N)$ when $d = 2$ and $\Gamma(p) \asymp N^{d-2}$ when $d \geq 3$, and these correspond to the choice $\alpha = 2$.
}

In particular, this result means that in $d = 3$ dimensions any radially-symmetric sampling pattern will yield a measurement condition that scales linearly with $N$. This, in view of Corollary \ref{c:thoptdD} is theoretically suboptimal.

\begin{remark}[Why radially-symmetric Fourier sampling is suboptimal]
\label{r:whysuboptimal}
This arises from the proof of Theorem \ref{t:TVVDSdD}, which, following \cite{NeedellWardTV2,NeedellWardTV1}, relies on Haar wavelets. This proof relates the recovery properties of a variable-density scheme for gradient sparse images to its recovery properties for images which are sparse in the discrete Haar wavelet basis. The study of Fourier sampling with wavelets has been considered extensively in \cite{AHPRBreaking,KrahmerWardCSImaging,LiAdcockRIP} and elsewhere. In essence, the optimal variable density scheme is determined by the behaviour of Haar wavelets in frequency space. In one or two dimensions, the Fourier transform of a Haar wavelet decays sufficiently rapidly in all directions to allow for radially-symmetric sampling strategies to be optimal. However, as shown in \cite{AdcockEtAlCISampStrat}, in three or more dimensions, the slow decay of the Fourier transform of a multi-dimensional Haar wavelet means that the optimal sampling scheme is no longer, as termed therein, \textit{isotropic} (i.e.\ radially symmetric), but rather \textit{anisotropic}, similar to what is described in Corollary \ref{c:thoptdD}.
\end{remark}

\subsection{Near-optimal sampling using hyperbolic cross densities}

In $d \geq 3$ dimensions, it is interesting to determine other densities which offer theoretically optimal or near-optimal performance. As seen in Fig.\ \ref{fig:3D_pattern_plots}, the three-dimensional theoretically-optimal density \eqref{thopt3D} has level curves that fail to be smooth at certain points. To conclude this section, we now identify a different density possessing smooth level curves which is optimal up to the log factor. This is based on hyperbolic cross sampling:

\cor{
[Near-optimal hyperbolic cross Fourier sampling, $d \geq 2$ dimensions]
\label{c:HCnearoptimal}
Let $d \geq 2$ and consider the setup of Theorem \ref{t:TVVDSdD} with $p = (p_{\omega})$ given by
\be{
\label{hypcrossdD}
p_{\omega} = \frac{C_{N,d}}{\overline{\omega_1} \cdots \overline{\omega_d}} ,\qquad \omega \in \{-N/2+1,\ldots,N/2\}^d,
}
and $s \gtrsim \log(N)$.
Then the conclusions of Theorem \ref{t:TVVDSdD} hold (with $\Gamma(p) \lesssim_d \log^d(N)$ in the case of \eqref{srani} and \eqref{sri}), provided $m$ satisfies
\be{
\label{hypcrossmeascond}
m \gtrsim_d s \cdot \log(s) \cdot \log^{d+2}(N) \cdot \left ( \log(s) \cdot \log(N) + \log(2 \varepsilon^{-1}) \right ).
}
}

This result shows that hyperbolic cross sampling is near optimal. In particular, the measurement condition \eqref{hypcrossmeascond} is worse than the optimal condition \eqref{thoptdDmeascond} only by a factor of $\log^{d-1}(N)$.
Fig.\ \ref{f:2d_level_curves} plots the level curves of two-dimensional hyperbolic cross sampling and Fig.\ \ref{fig:3D_pattern_plots} shows a three-dimensional level set. Notice that this strategy mimics the spikes of the theoretically-optimal pattern, but is less dense near the centre. However, its is a smooth function of $\overline{\omega_1},\ldots,\overline{\omega_d}$, unlike in the case of the latter. We note in passing that the hyperbolic cross is a well-known object in multivariate approximation theory \cite{TemylakovHC}, where it is used to overcome the curse of dimensionality.

\section{Main results on Walsh sampling}
\label{s:BeyondFourierTV}

We now consider Walsh sampling. The major difference between this and the previous case is that the Walsh--Hadamard transform does not commute with the discrete gradient operator. For this reason, we do not provide gradient recovery estimates, we only consider variable density sampling and we assume throughout that $d \geq 2$ (see \S \ref{s:proofsII} for some further discussion on this point).
For simplicity, we state our results for anisotropic TV only in this section. However, results for isotropic TV can be readily proved as well.

Recall from \S \ref{ss:DWT} that Walsh frequencies are indexed over $\{ 0 ,\ldots,N-1\}^d$. Thus, we now consider variable density sampling according to probability distributions $p = (p_i)_{i \in \{0,\ldots,N-1\}^d}$ over this set. We let $\Gamma(p) \geq 0$ be the smallest constant such that
\be{
\label{CpdefWalsh}
\left ( 1 + \nmu{i}^{d}_{\ell^\infty} \right )^{-1} \leq \Gamma(p) p_i,\qquad \forall i \in \{0,\ldots,N-1\}^d.
}
Once more we notice that $\Gamma(p) \geq 1$, since the $p$ is a probability distribution and the left-hand side is equal to one when $i = (0,\ldots,0)$. Our main result is the following:

\thm{
[Variable density Walsh sampling, $d \geq 2$ dimensions]
\label{t:WalshTV}
Let $d \geq 2$, $0 < \varepsilon < 1$, $2 \leq s , m \leq N^d$ and $\Omega \subseteq \{0,\ldots,N-1\}^{d}$ be a variable density sampling scheme of order $m$ corresponding to a probability distribution $p = (p_i)$.  Let $A = \frac{1}{\sqrt{m}} P_{\Omega} H$ and suppose that
\bes{
m \gtrsim_{d} \Gamma(p) \cdot s \cdot \log^2(N/s) \cdot \log(N) \cdot \log(\Gamma(p) s) \cdot \left ( \log(\Gamma(p) s) \cdot \log(N) + \log(\varepsilon^{-1})  \right ),
}
where $\Gamma(p)$ is as in \eqref{CpdefWalsh}.
Then the following holds with probability at least $1-\varepsilon$.  For all $x \in \bbC^{N^d}$ and $y = A x + e \in \bbC^m$, where $\nm{e}_{\ell^2} \leq \eta$ for some $\eta \geq 0$, every minimizer $\hat{x}$ of \eqref{TVminprob} satisfies
\be{
\label{Walsherr}
\nm{x - \hat{x}}_{\ell^2} \lesssim \frac{\sigma_{s}(\nabla x)_{\ell^1}}{\sqrt{s \log(N)}} + \sqrt{\Gamma(p)} \eta .
}
}
Similar to Fourier sampling, this result asserts stable and robust recovery of the image $x$ from Walsh measurements, up to log factors, taken according to the appropriate variable density strategy.  We now consider the choice of sampling strategy:

\lem{
\label{l:GammaWalsh}
Let $p = (p_i)$ be a probability distribution and $\Gamma(p)$ be as in \eqref{CpdefWalsh}. Then $\Gamma(p) \gtrsim_d \log(N)$. Moreover, if 
\bes{
p_{i} = \frac{C_{N,d}}{1+\nm{i}^{d}},\qquad i \in \{0,\ldots,N-1\}^d,
}
where $\nm{\cdot}$ is any norm, then $\Gamma(p) \lesssim_d \log(N)$.
}

\cor{[Theoretically-optimal variable density Walsh sampling, $d \geq 2$ dimensions]
\label{c:thoptdDWalsh}
Consider the setup of Theorem \ref{t:WalshTV} with $p = (p_i)$ given by
\bes{
p_i = \frac{C_{N,d}}{1+\nm{i}^{d}},\qquad i \in \{0,\ldots,N-1\}^d,
}
where $\nm{\cdot}$ is any norm on $\bbR^d$, and $s \gtrsim \log(N)$. Then the conclusions of Theorem \ref{t:WalshTV} hold (with $\Gamma(p) \lesssim_d \log(N)$ in \eqref{Walsherr}), provided $m$ satisfies
\bes{
m \gtrsim_{d} s \cdot \log(s) \cdot \log^2(N/s) \cdot \log^2(N)  \left ( \log(s) \cdot \log(N) + \log(\varepsilon^{-1}) \right ).
}
}
Much like with Fourier sampling (Corollary \ref{c:thoptdD}), this result asserts a class of theoretically-optimal sampling strategies which ensure stable and robust recovery in $d \geq 2$ dimensions from Walsh measurements. We are unaware of any similar result in the literature. 
It is notable, however, that the optimal sampling strategy is radially symmetric in all dimensions, unlike in the Fourier case. See Remark \ref{r:WalshHaar} below.
We also note that the log term in Corollary \ref{c:thoptdDWalsh} is worse by a factor of $\log^2(N/s) / \log(N)$ than that of Corollary \ref{c:thoptdD}. This stems from the proof strategy, and specifically the different technique that is used in the Walsh case in the absence of the commuting property.

\begin{remark}
\label{r:WalshHaar}
Similar to the Fourier case (Remark \ref{r:whysuboptimal}), the explanation for why radially-symmetric sampling works in any dimensions for Walsh sampling can be traced to the use of Haar wavelets in the proof. Haar wavelets and Walsh functions are intimately related, see \eqref{WalshHaar}. This means that the Walsh transform of a Haar wavelet behaves far more nicely than its Fourier transform, which in turn allows one to use a radially-symmetric sampling pattern in any dimension. See also \cite{AdcockEtAlCISampStrat}. By contrast, as shown in \S \ref{ss:suboptimalradial} the use of a radially-symmetric sampling pattern in the Fourier case leads to a measurement condition with a factor of $N^{d-2}$.
\end{remark}

\section{Experiments and discussion}\label{s:experiments}

We now show a series of further numerical experiments.

\subsection{Experimental setup}

\begin{figure}
	\begin{center}
		\includegraphics[width=0.35\paperwidth,clip=true,trim=30mm 80mm 20mm 80mm]{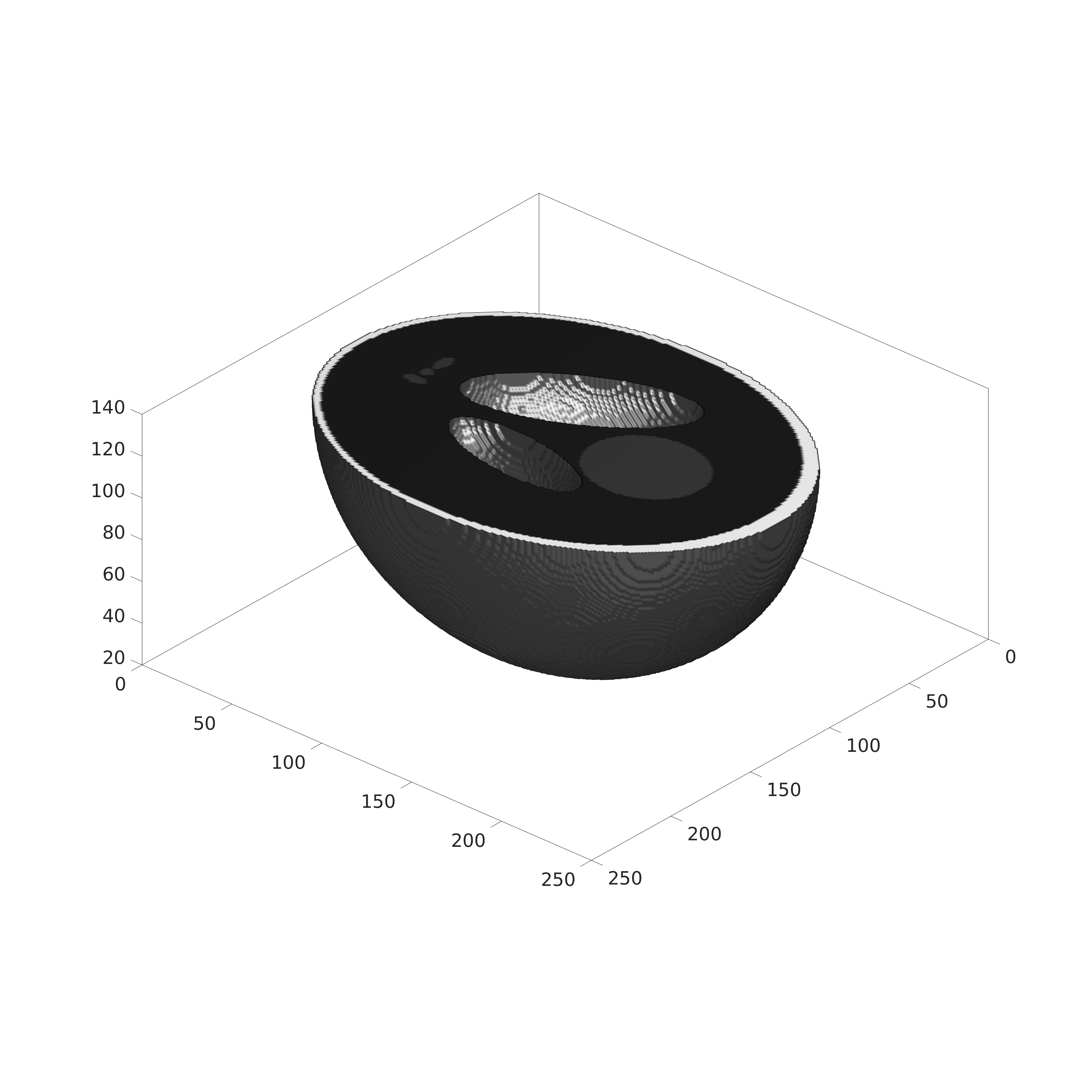}
		\includegraphics[width=0.35\paperwidth,clip=true,trim=30mm 80mm 20mm 80mm]{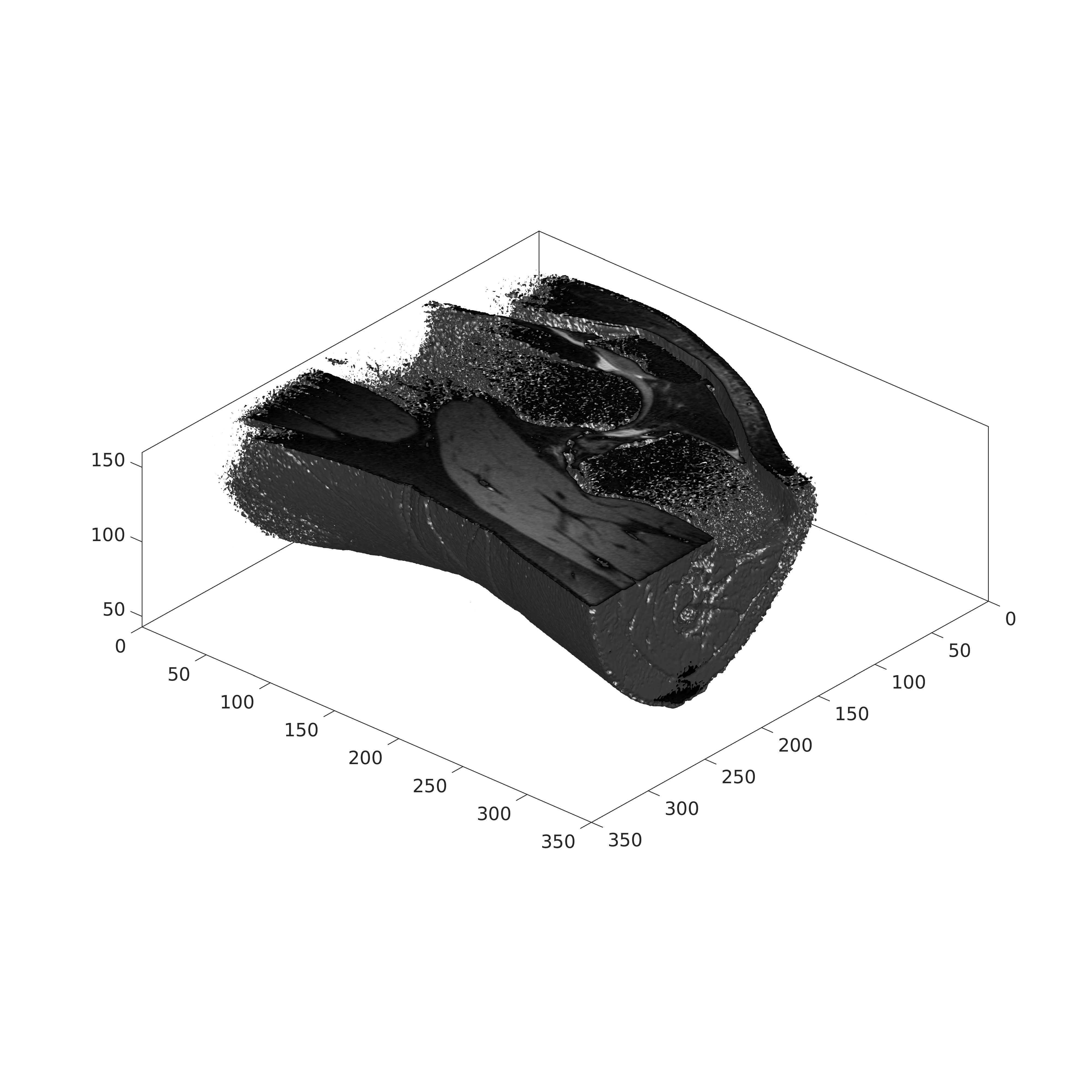}
	\end{center}
	\vspace{-2mm}
	\caption{The {\bf(left)} Shepp--Logan phantom (size $256^3$) generated with \url{https://www.mathworks.com/matlabcentral/fileexchange/9416-3d-shepp-logan-phantom} and {\bf(right)} ``knee MRI'' (size $320^2\times 256$) three-dimensional test images for Fourier sampling. The ``knee MRI'' test image is generated from the MRI data from case 11 of the ``Stanford Fullysampled 3D FSE Knees'' dataset available at \url{https://mridata.org}, and was zero-padded to obtain a test image of size $320^3$. 
	}
	\label{fig:3D_test_images}
\end{figure}

\begin{figure}
	\begin{center}
		\includegraphics[width=0.23\paperwidth,clip=true,trim=0mm 0mm 0mm 0mm]{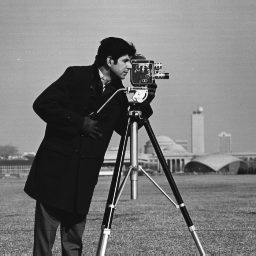}
		\includegraphics[width=0.23\paperwidth,clip=true,trim=0mm 0mm 0mm 0mm]{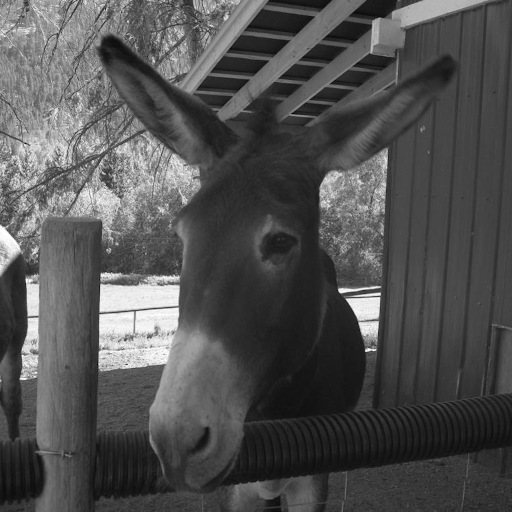}
		\includegraphics[width=0.23\paperwidth,clip=true,trim=0mm 0mm 0mm 0mm]{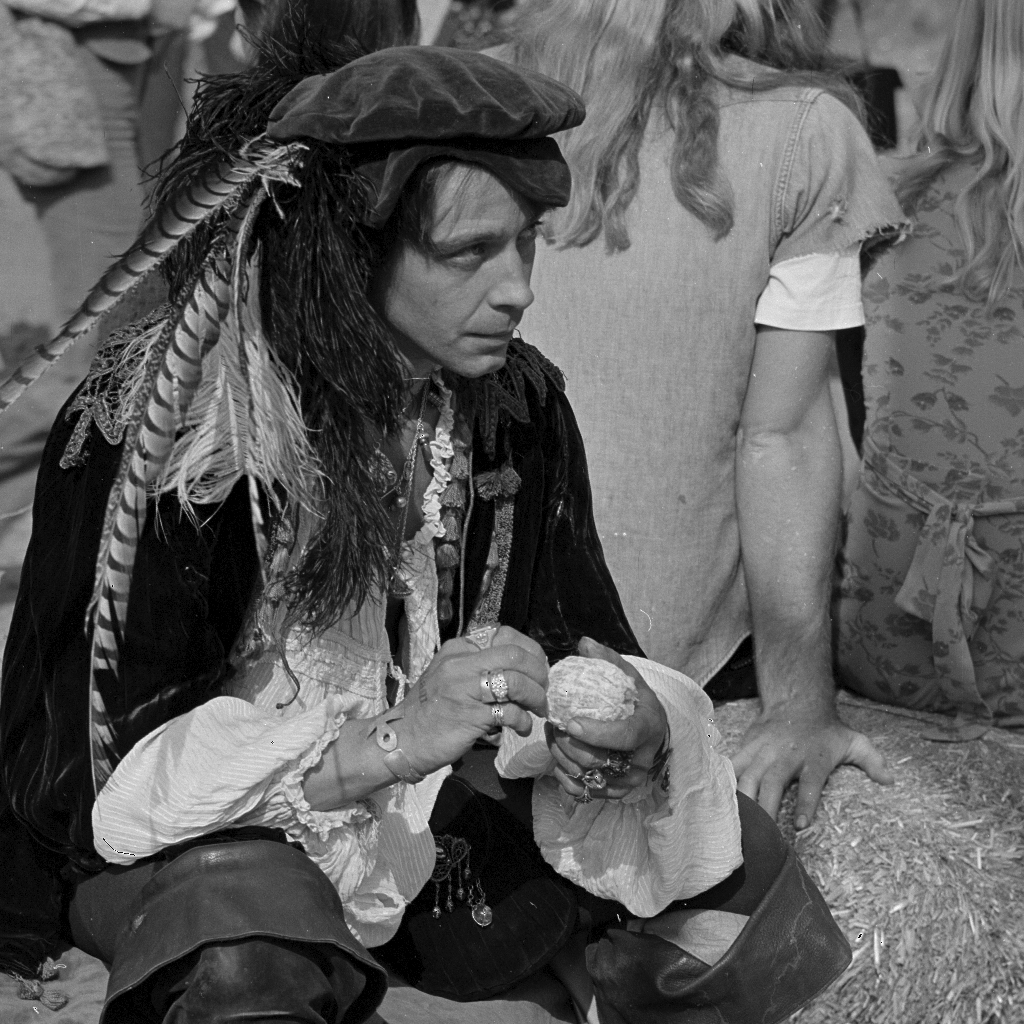}
	\end{center}
	\vspace{-2mm}
	\caption{The {\bf(left)} ``cameraman'' (size $256^2$), {\bf(middle)} ``donkey'' (size $512^2$) and {\bf(right)} ``man'' (available in sizes $256^2$, $512^2$ and $1024^2$) test images for Walsh--Hadamard sampling. 
	}
	\label{fig:Walsh_2D_images}
\end{figure}

We first describe the details of these experiments. We focus on reconstructing either three-dimensional MRI or test data, Fig.\ \ref{fig:3D_test_images}, with Fourier sampling or two-dimensional natural images with Walsh sampling, Fig.\ \ref{fig:Walsh_2D_images}. For each of our experiments, we run 20 trials of reconstructing the given image using a modified version of the {\tt NESTA} solver \cite{Becker2011} which allows for reconstruction of two- or three-dimensional images via TV-minimization. The NESTA parameters used are designed for images whose values lie in the range $[0,100]$, and therefore we rescale all images to this range. These parameters are $\mu = 0.2$, $5$ outer iterations, $5000$ inner iterations, a tolerance of $10^{-5}$ and $\delta = 10^{-5}$. We run 20 random trials, each with a different seed, and plot the average PSNR values.

We consider six sampling patterns, four of which have already been introduced in this paper. These are: \textit{uniform random}, \textit{hyperbolic cross} \eqref{hypcrossdD}, the \textit{theoretically-optimal} pattern (see Corollaries \ref{c:thoptdD} and \ref{c:thoptdDWalsh} for Fourier and Walsh--Hadamard respectively) and the \textit{inverse square law}. We also consider two further sampling patterns, \textit{half-half} sampling and \textit{multilevel random} subsampling. The former fully samples the lowest $m/2$ frequencies and then randomly subsamples the remainder. The latter was introduced in \cite{AHPRBreaking}. In this scheme, one first divides frequency space into $r$ annular regions $B_1,\ldots,B_r$ of equal width. Next, one defines a decreasing sampling fraction $p_k = m_k / |B_k|$ as
\bes{
p_k = 1,\quad k = 1,\ldots,r_0,\qquad p_k = \exp \left ( - \left ( \frac{b(k-r_0)}{r-r_0} \right )^a \right ),\quad k = r_0+1,\ldots,r,
}
where $r_0$ and $a$ are parameters, and $b$ is chosen so that $\sum^{r}_{k=1} m_k = m$. Finally, within each region $B_k$ one selects $m_k$ samples uniformly and randomly. We refer to \cite{AHPRBreaking} for further details.

\subsection{Fourier sampling}

Fig.\ \ref{fig:3D_PSNR} displays the PSNR values for reconstructing the two Fourier test images shown in Fig.\ \ref{fig:3D_test_images}. Note that the reconstruction is performed in three dimensions, while the Fig.\ shows the PSNR versus frame number in the $z$-direction.

\begin{figure}
	\begin{center}
		\includegraphics[width=0.14\paperwidth,clip=true,trim=12mm 0mm 20mm 0mm]{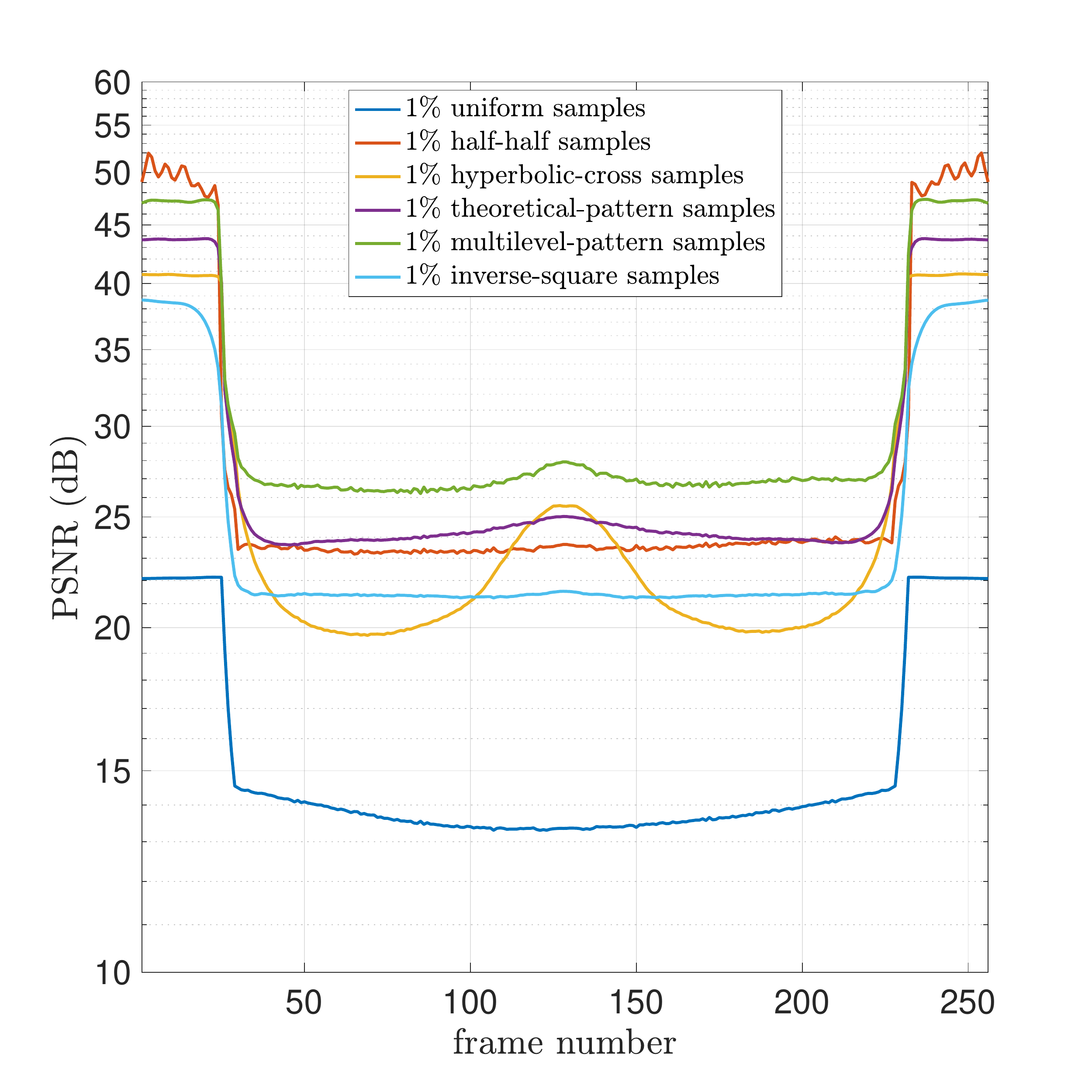}
		\includegraphics[width=0.14\paperwidth,clip=true,trim=12mm 0mm 20mm 0mm]{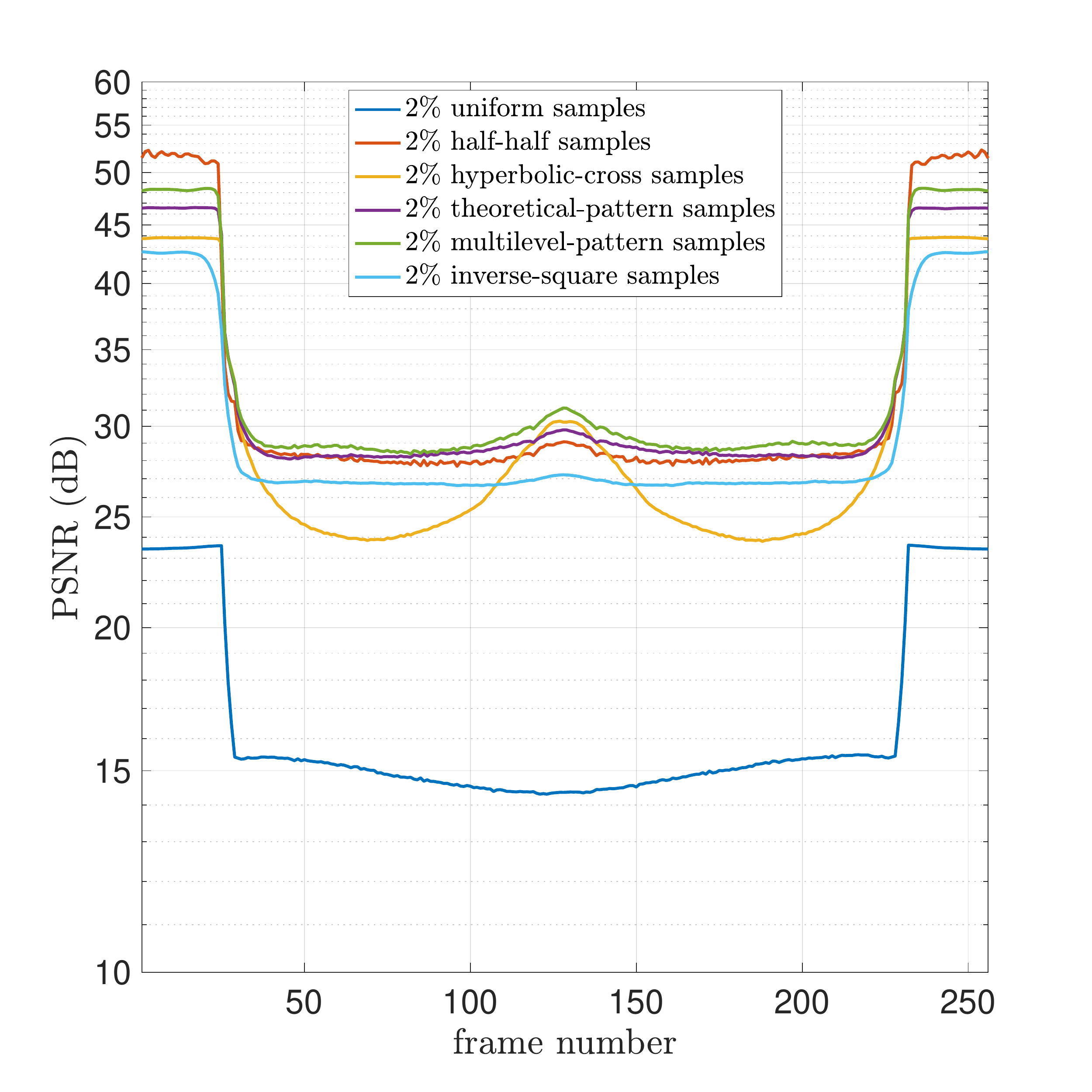}
		\includegraphics[width=0.14\paperwidth,clip=true,trim=12mm 0mm 20mm 0mm]{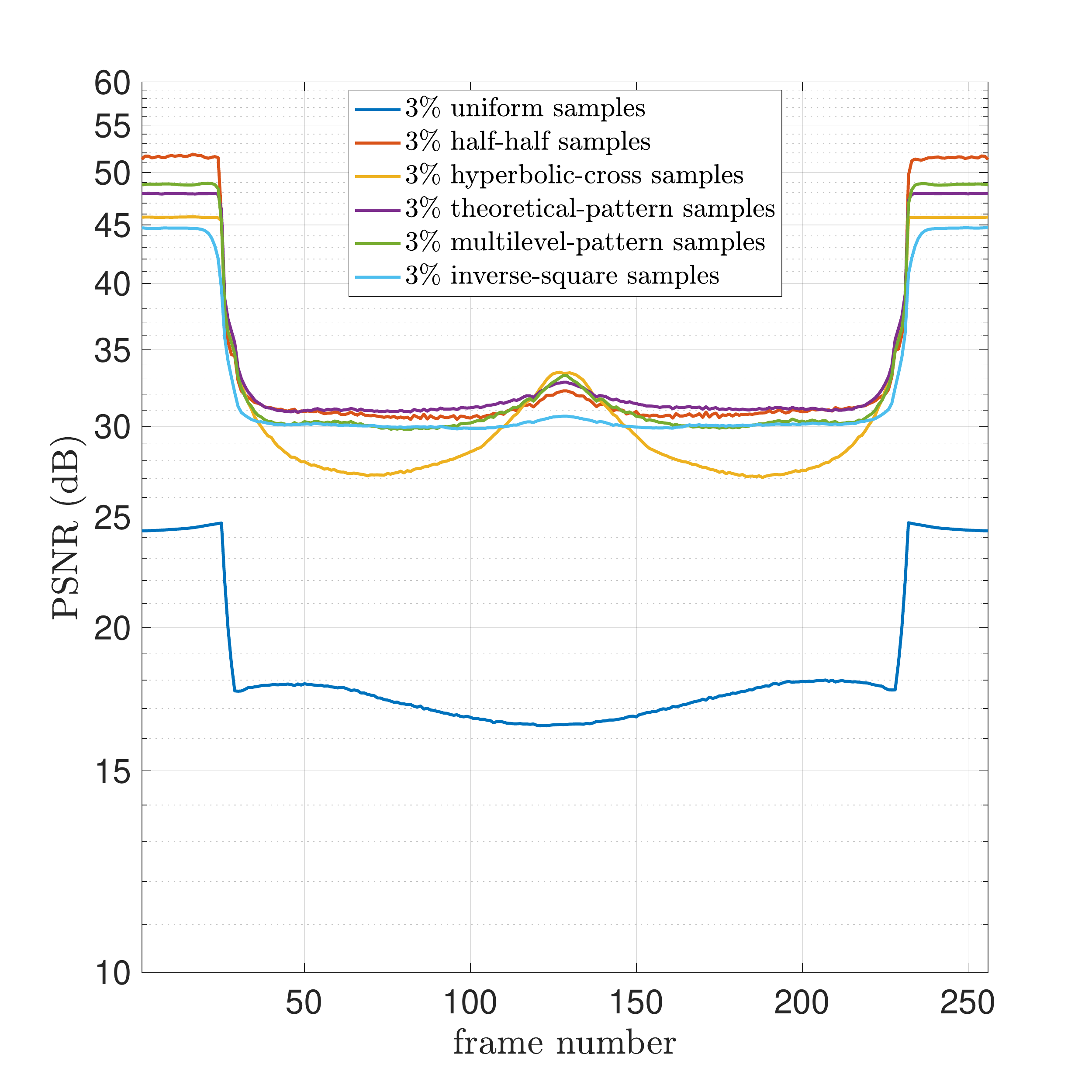}
		\includegraphics[width=0.14\paperwidth,clip=true,trim=12mm 0mm 20mm 0mm]{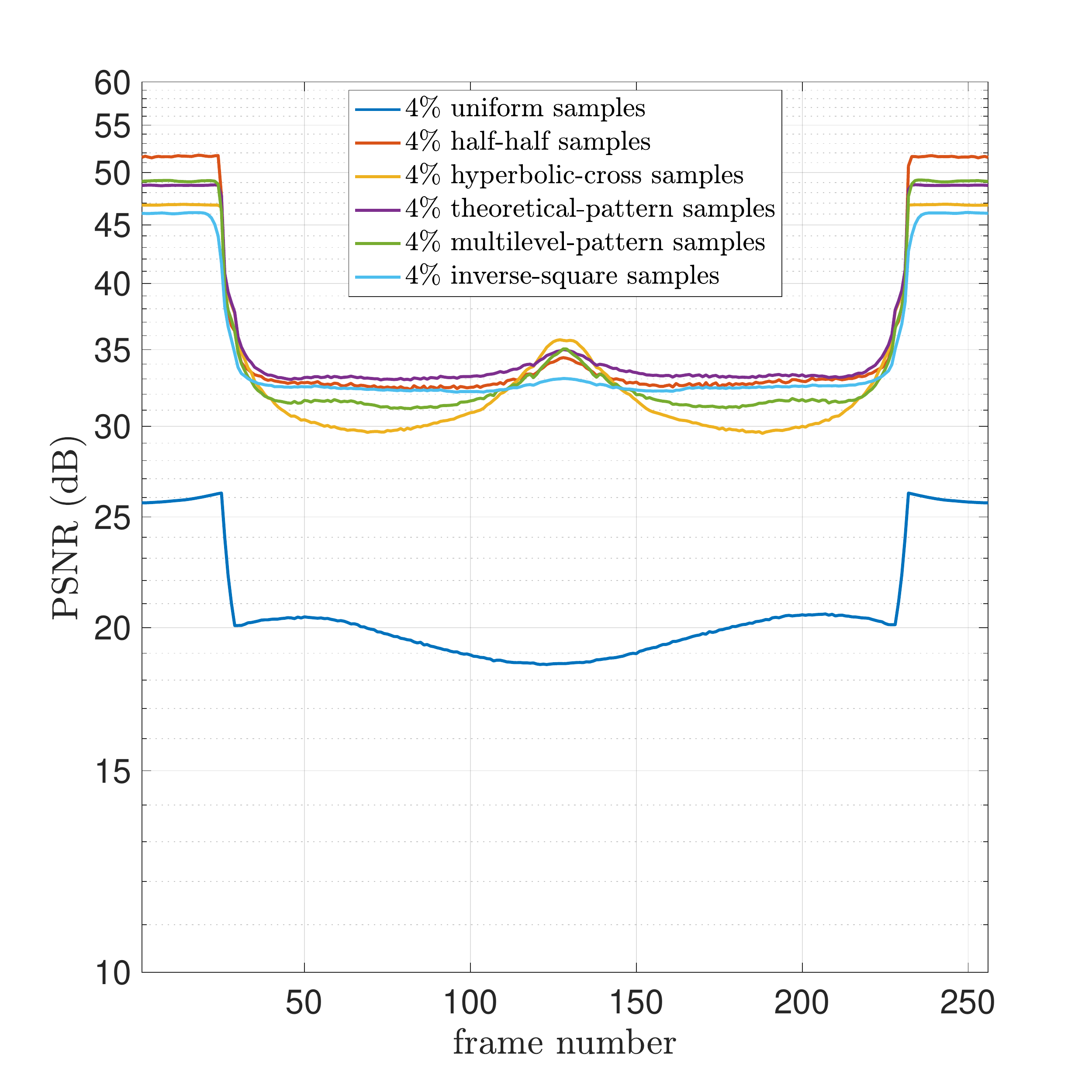}
		\includegraphics[width=0.14\paperwidth,clip=true,trim=12mm 0mm 20mm 0mm]{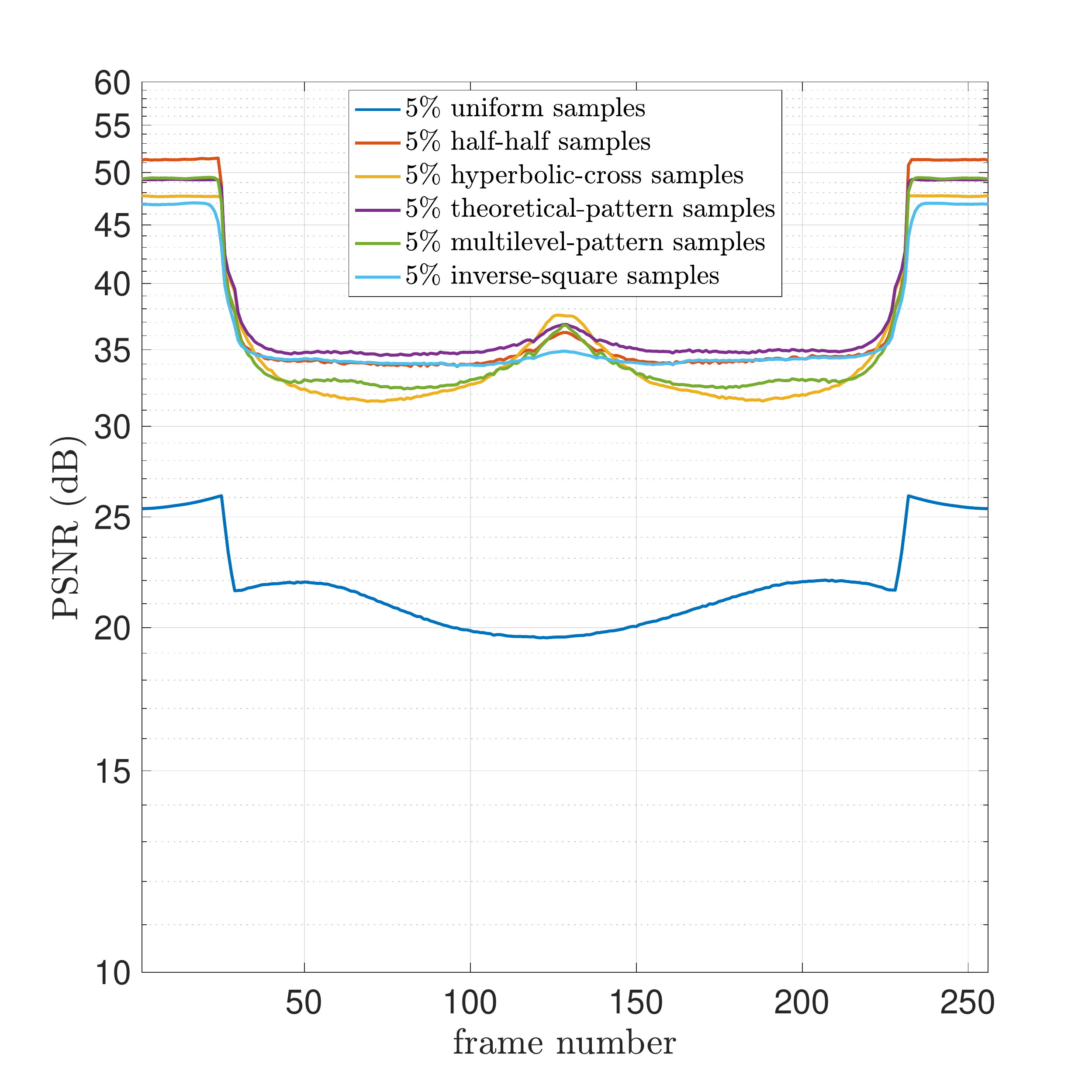} \\
		\includegraphics[width=0.14\paperwidth,clip=true,trim=12mm 0mm 20mm 0mm]{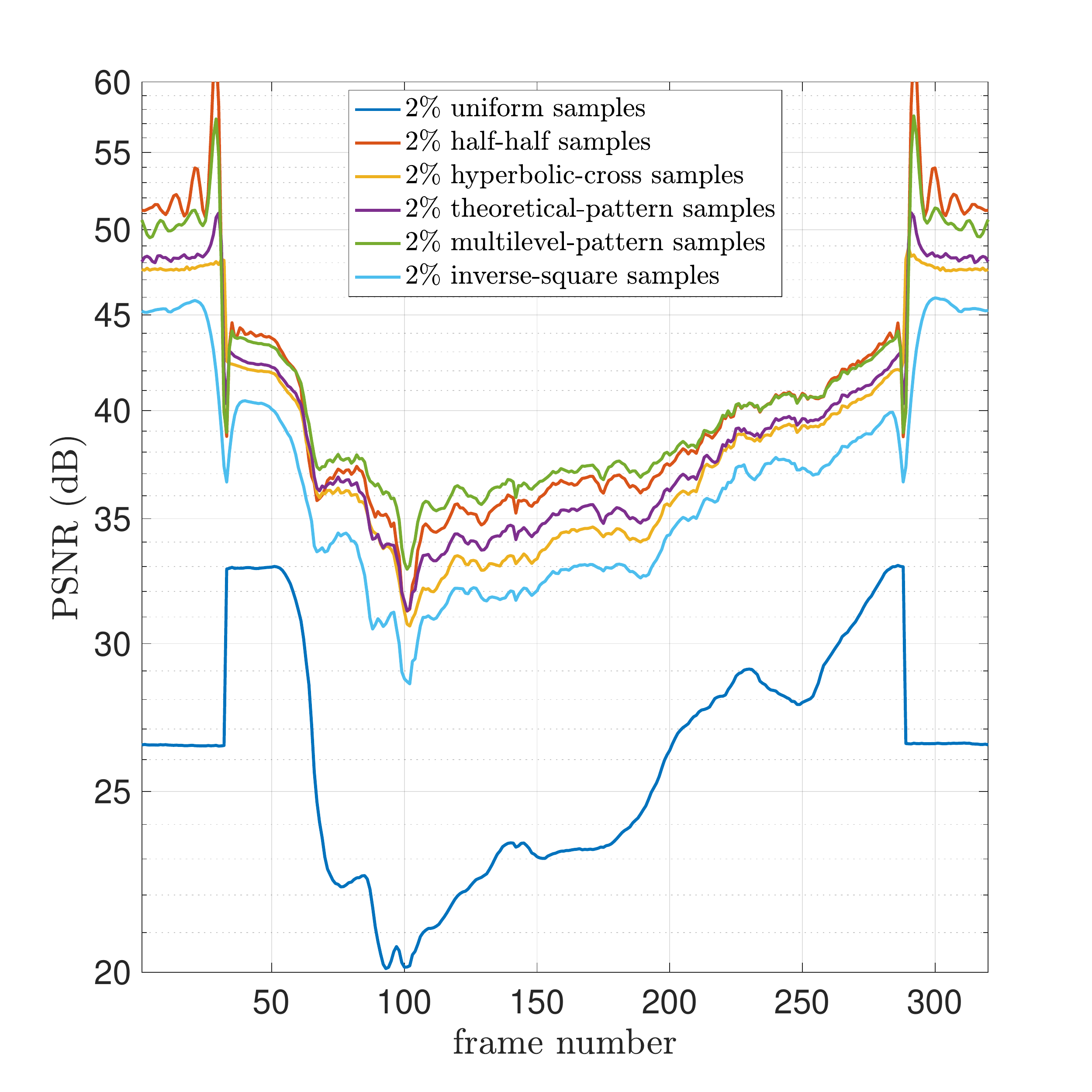}
		\includegraphics[width=0.14\paperwidth,clip=true,trim=12mm 0mm 20mm 0mm]{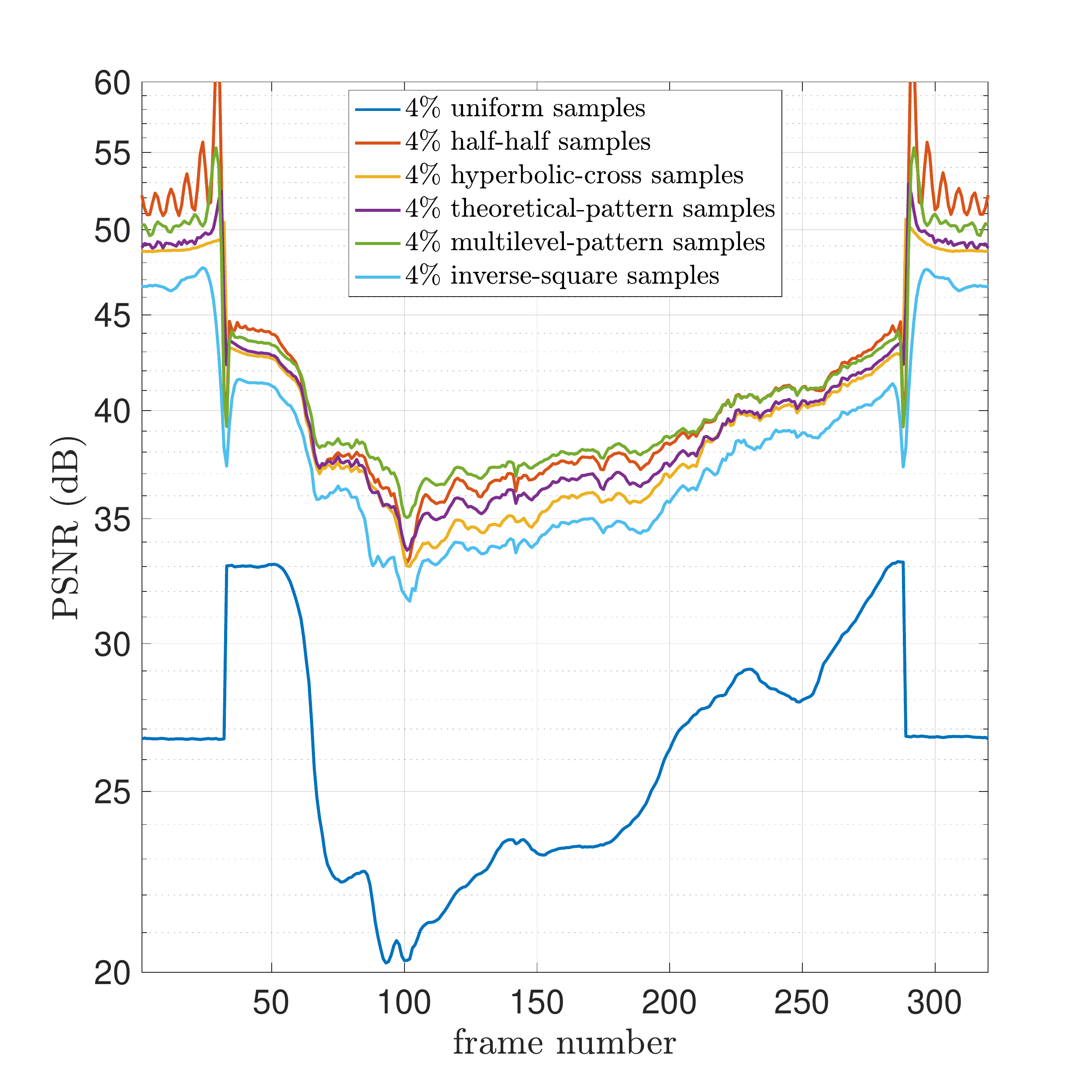}
		\includegraphics[width=0.14\paperwidth,clip=true,trim=12mm 0mm 20mm 0mm]{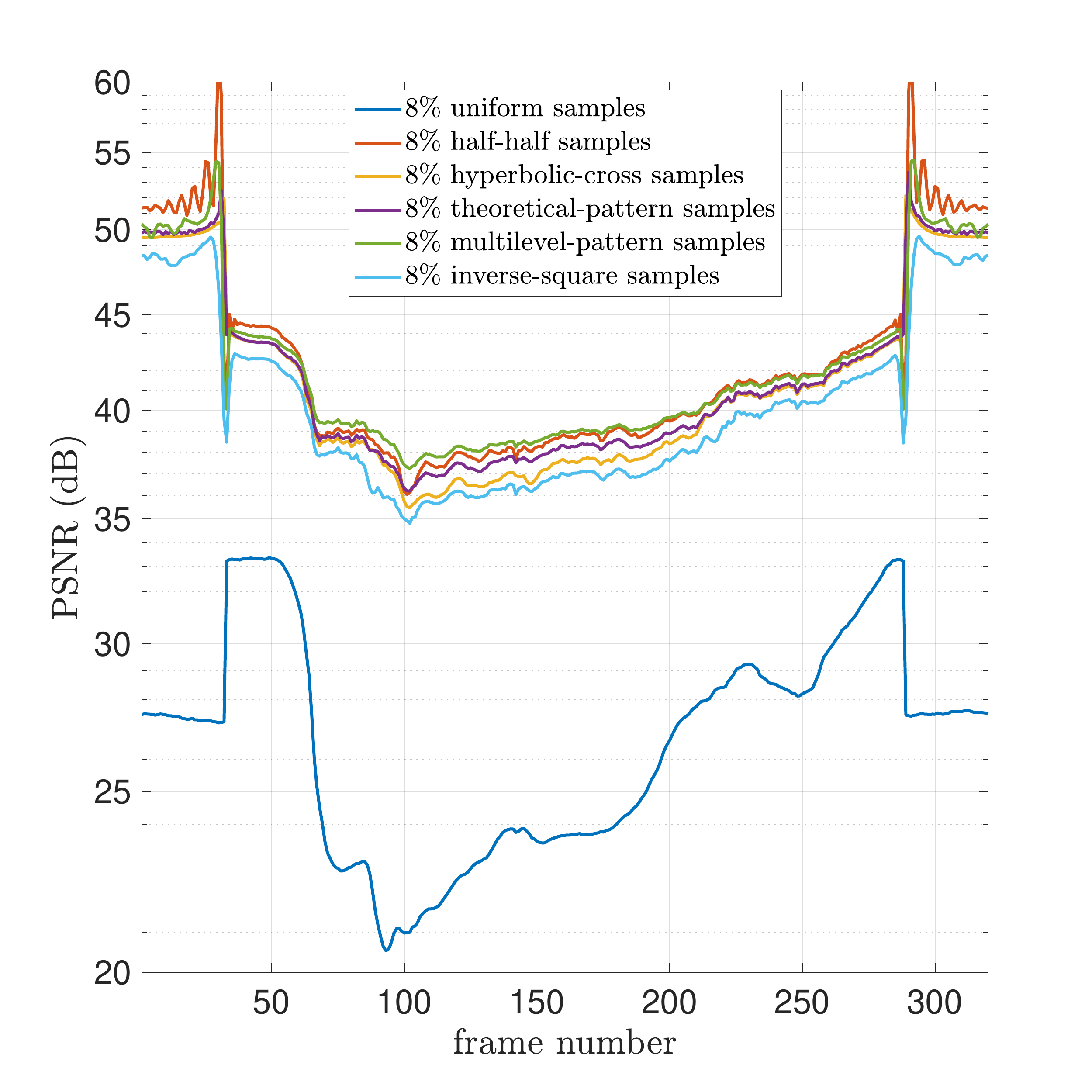}
		\includegraphics[width=0.14\paperwidth,clip=true,trim=12mm 0mm 20mm 0mm]{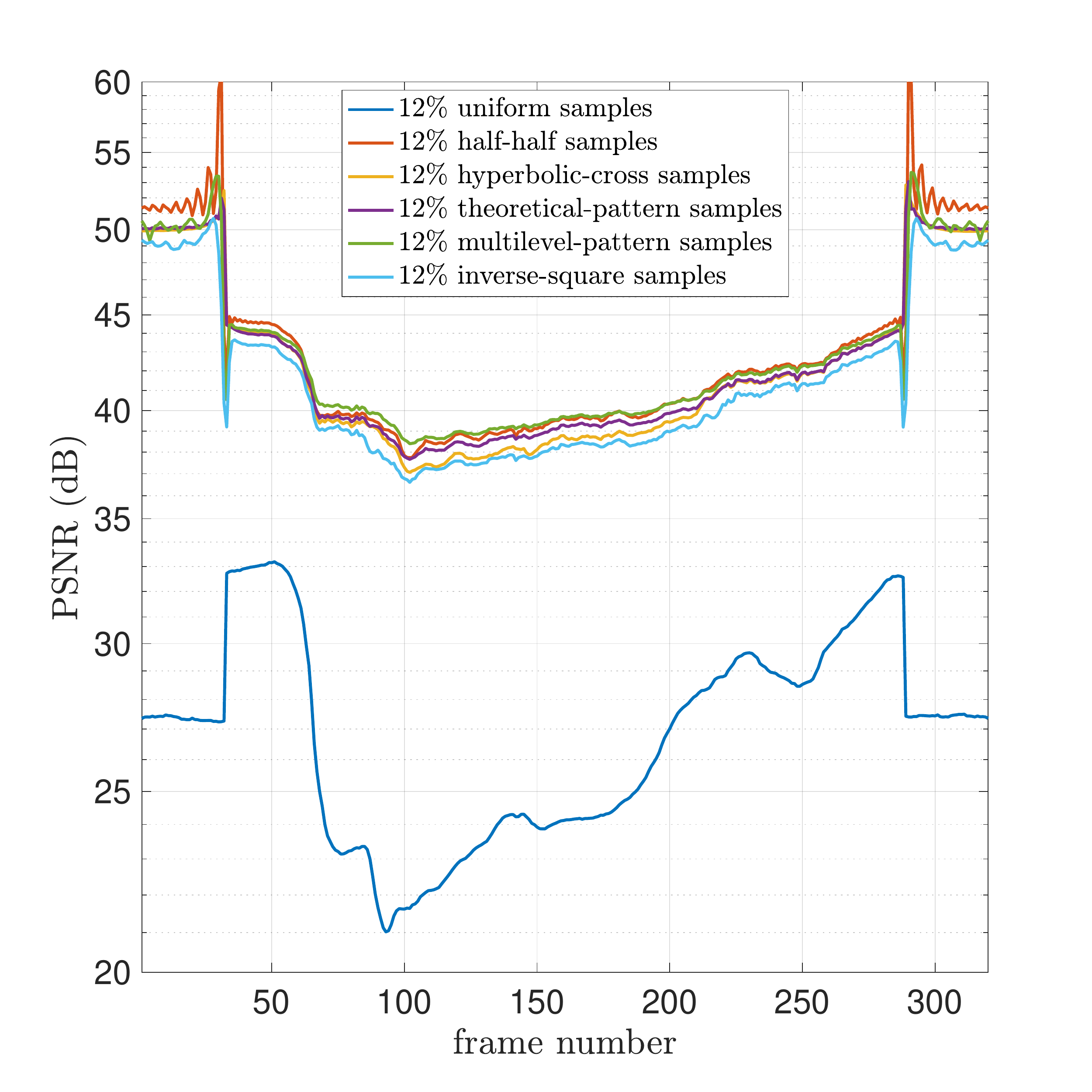}
		\includegraphics[width=0.14\paperwidth,clip=true,trim=12mm 0mm 20mm 0mm]{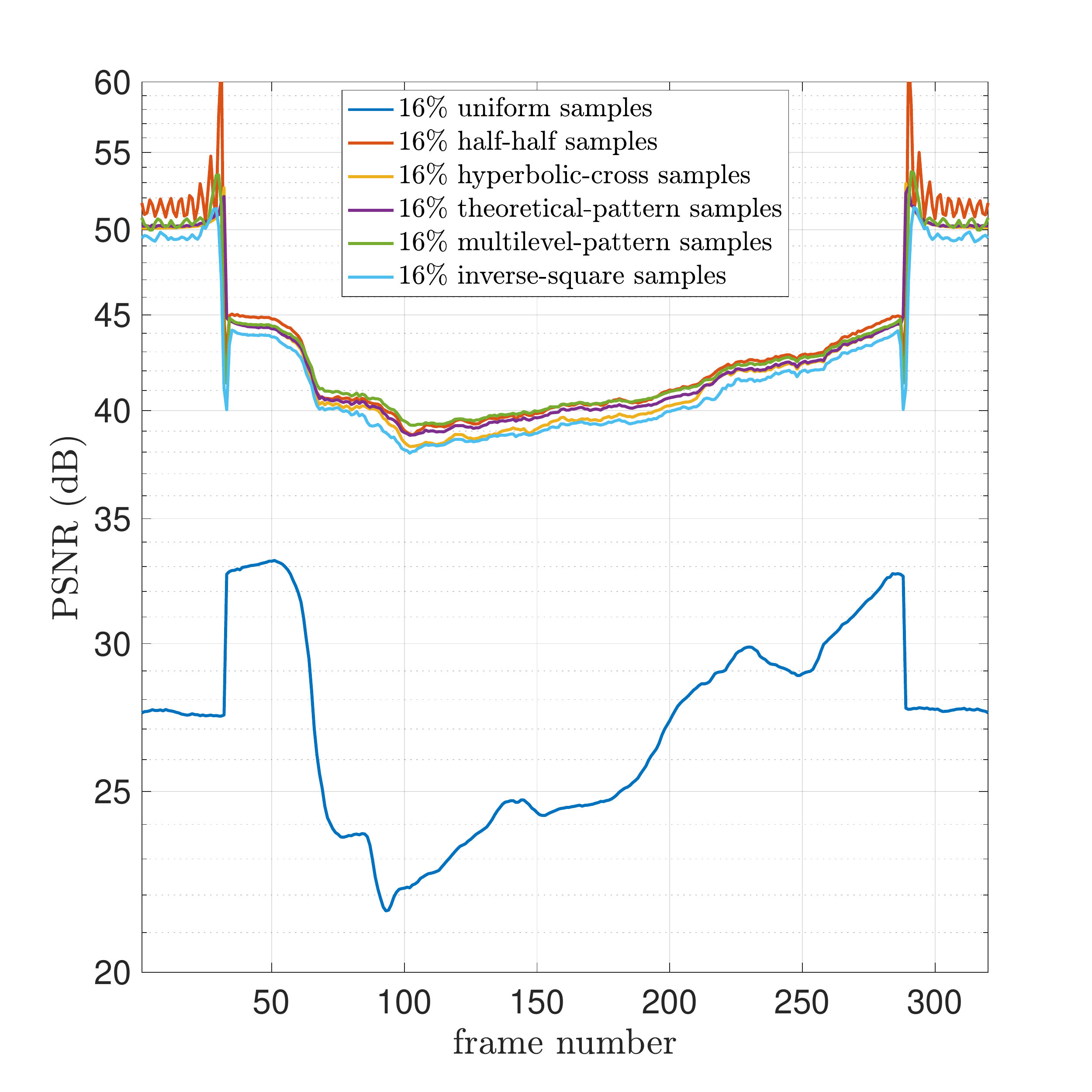}
	\end{center}
	\vspace{-2mm}
	\caption{Average PSNR values over 20 trials in reconstructing each frame of the Shepp--Logan phantom (top) and knee MRI (bottom) test image from Fig.\ \ref{fig:3D_test_images} with different Fourier sampling strategies as the sampling percentage is increased from left to right. Here the multilevel sampling is performed with $a=1$, $r=20$ and $r_0=1$ (top) or $r_0 = 2$ (bottom).}
	\label{fig:3D_PSNR}
\end{figure}

As expected, uniform random sampling performs very poorly in comparison to all other schemes. Similar, as predicted in \S \ref{ss:suboptimalradial}, the inverse-square law generally performs relatively poorly in comparison to the others, especially for the more complicated knee MRI image. 

Interestingly, the multilevel scheme performs amongst the best, especially at low sampling percentages. Often, it outperforms the theoretically-optimal pattern. This is in spite of the fact that the multilevel scheme is radially symmetric, whereas it was argued in \S \ref{ss:suboptimalradial} radially-symmetric patterns, at least those that draw samples from a single density, are theoretically suboptimal in three dimensions. 

Typically, in the experiments, the second and third best performers are the theoretically-optimal and half-half schemes. It should come as little surprise that the latter performs worse than multilevel random sampling: full sampling followed by uniform random sampling is a relatively crude strategy. Interestingly, the behaviour of the hyperbolic cross scheme is much more heavily dependent on the frame for the Shepp--Logan phantom -- it is clearly too anisotropic to recover the details in some of the frames -- than the other patterns. But its relative frame-by-frame performance on the knee MRI image is similar to the other patterns.  

In Fig.\ \ref{fig:knee_MRI_frame_102} we show the recovery of an individual frame for two different sampling percentages. In both cases the half-half and multilevel patterns give a slightly sharper image in comparison to the theoretical pattern, which is slightly more blurred. As one would expect, the hyperbolic cross and inverse-square law both present substantial additional artefacts.

\begin{figure}
\begin{center}
\includegraphics[width=0.6100\paperwidth,clip=true,trim=11mm 3mm 11mm 0mm]{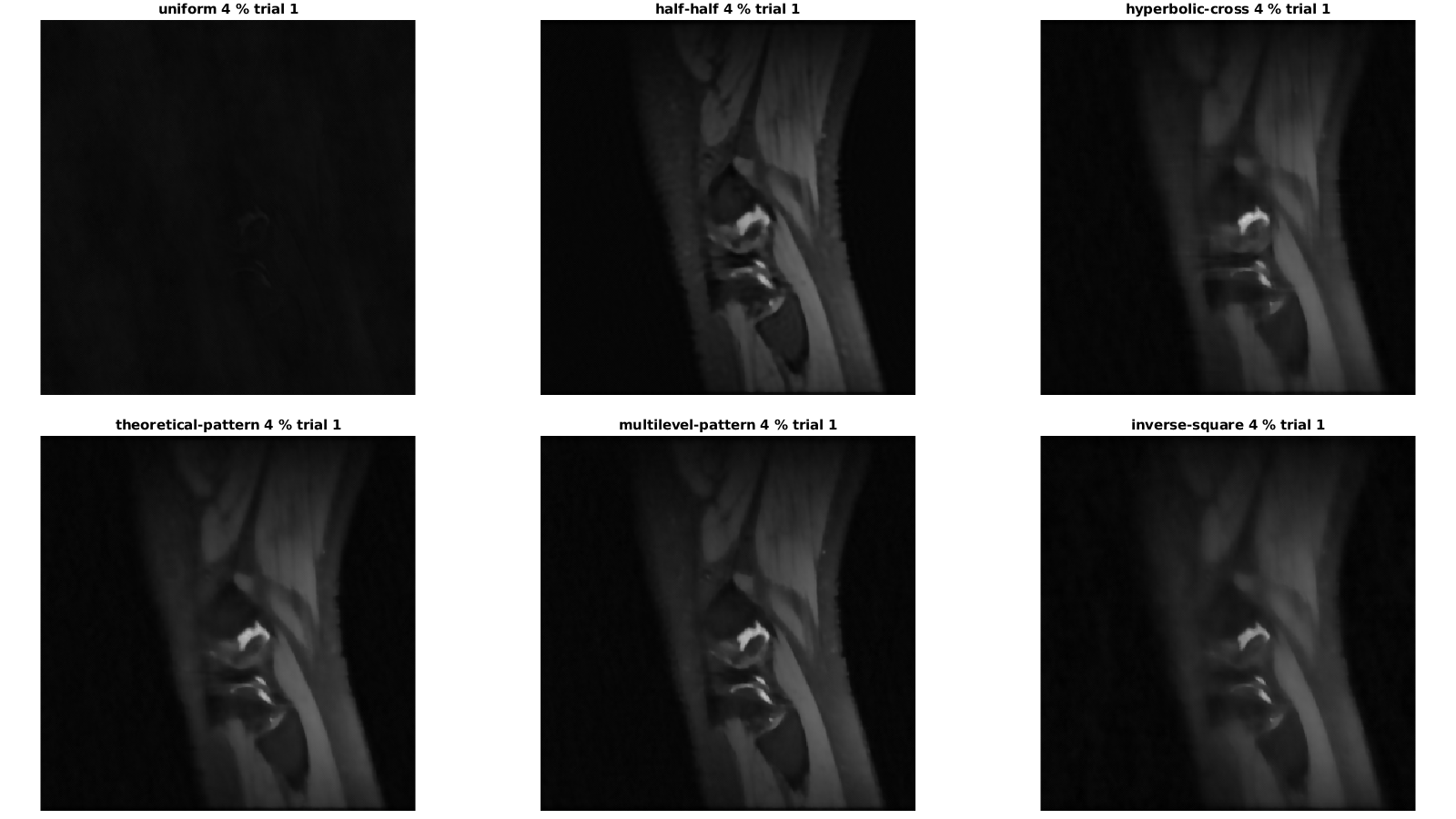}\\[0.2cm]
\includegraphics[width=0.6100\paperwidth,clip=true,trim=11mm 3mm 11mm 0mm]{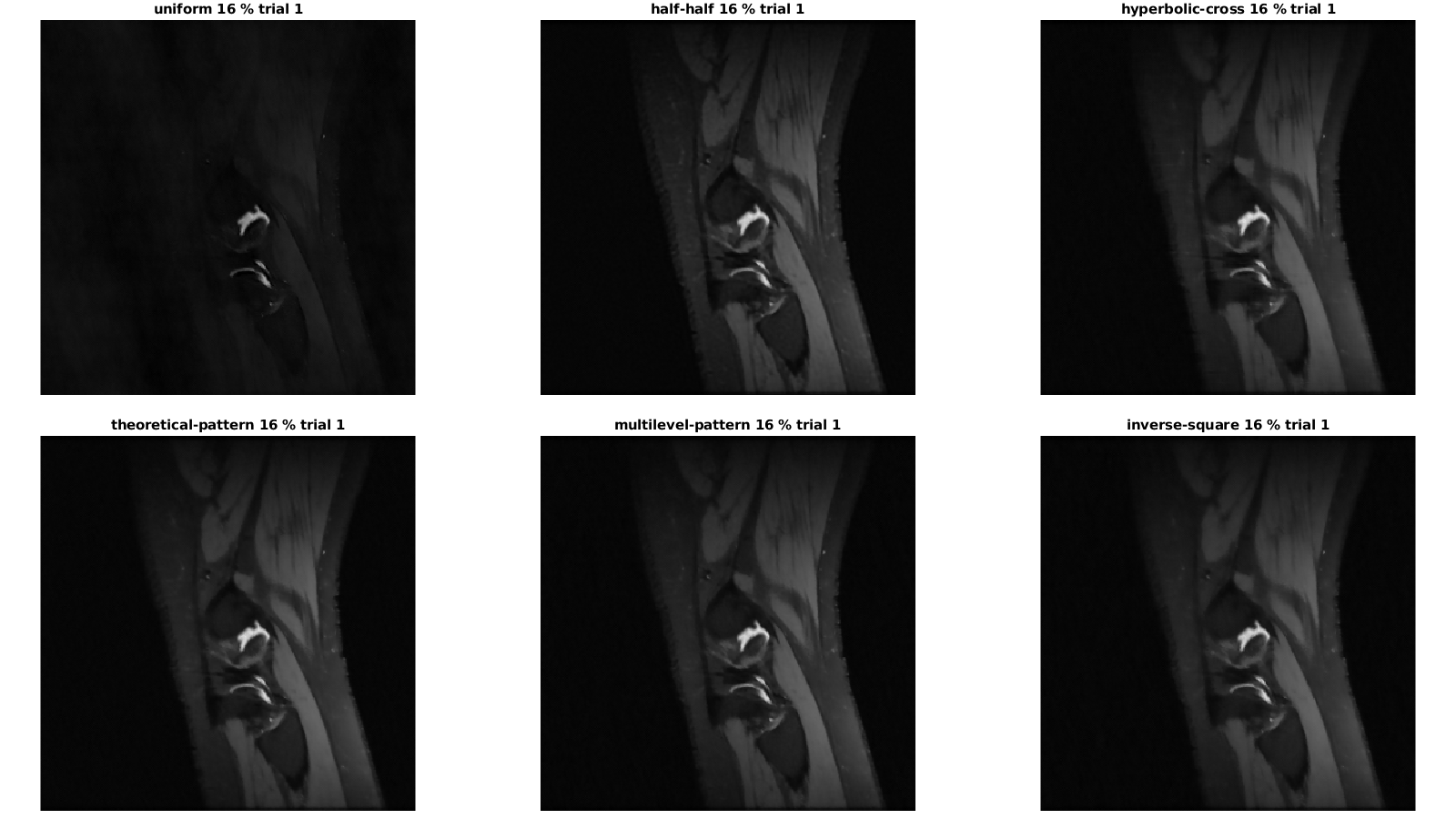}
\end{center}
\vspace{-2mm}
\caption{Comparison of reconstructions from trial 1 of 20 of frame 102 from the zero-padded ``knee MRI'' data with each method at {\bf(rows 1 \& 2)} 4\% and {\bf(rows 3 \& 4)} 16\% subsampling. 
}
\label{fig:knee_MRI_frame_102}
\end{figure}

\subsection{Walsh sampling}

In Fig.\ \ref{fig:Walsh_samp_comp} we consider two-dimensional Walsh sampling for the images in Fig.\ \ref{fig:Walsh_2D_images}. Across all images and all sampling percentages, the multilevel scheme consistently performs amongst the best, with generally the theoretically-optimal or half-half pattern performing second best. The relative performance of the half-half scheme is quite heavily dependent on the image, with it performing worse on the ``cameraman'' image but better on the ``donkey'' and ``man'' images. This is not surprising. The ``cameraman'' image is relatively simple, meaning the half-half scheme likely oversamples the high frequency regime. Conversely, the ``donkey'' and ``man'' images are more complex, meaning more sampling is needed at higher frequencies to resolve the fine details. This effect can be further examined by considering the ``man'' image at different resolutions, as we do in Fig.\ \ref{fig:Walsh_man_img_size_comp}. At low resolution the half-half scheme actually outperforms the multilevel scheme whenever the sampling percentage is greater than 12\%, whereas at higher resolution this only occurs after 21\%. This can once more be traced to the properties of the image. At low resolution, the edges of the image are relatively closer together, thus requiring more higher-frequency samples to resolve, whereas at higher resolutions they are relatively better separated.

\begin{figure}
	\begin{center}
		\includegraphics[width=0.23\paperwidth,clip=true,trim=8mm 6mm 18mm 15mm]{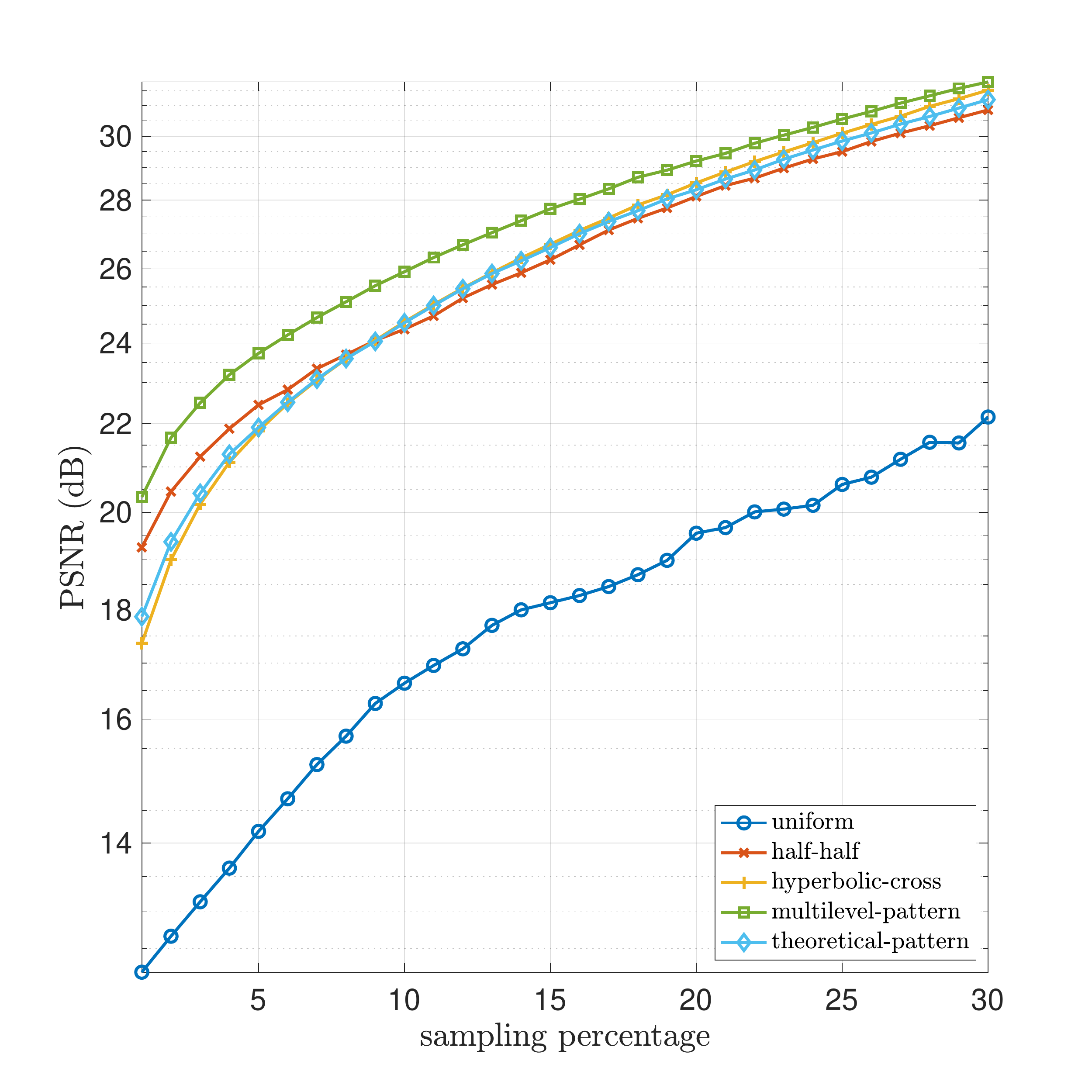}
		\includegraphics[width=0.23\paperwidth,clip=true,trim=8mm 6mm 18mm 15mm]{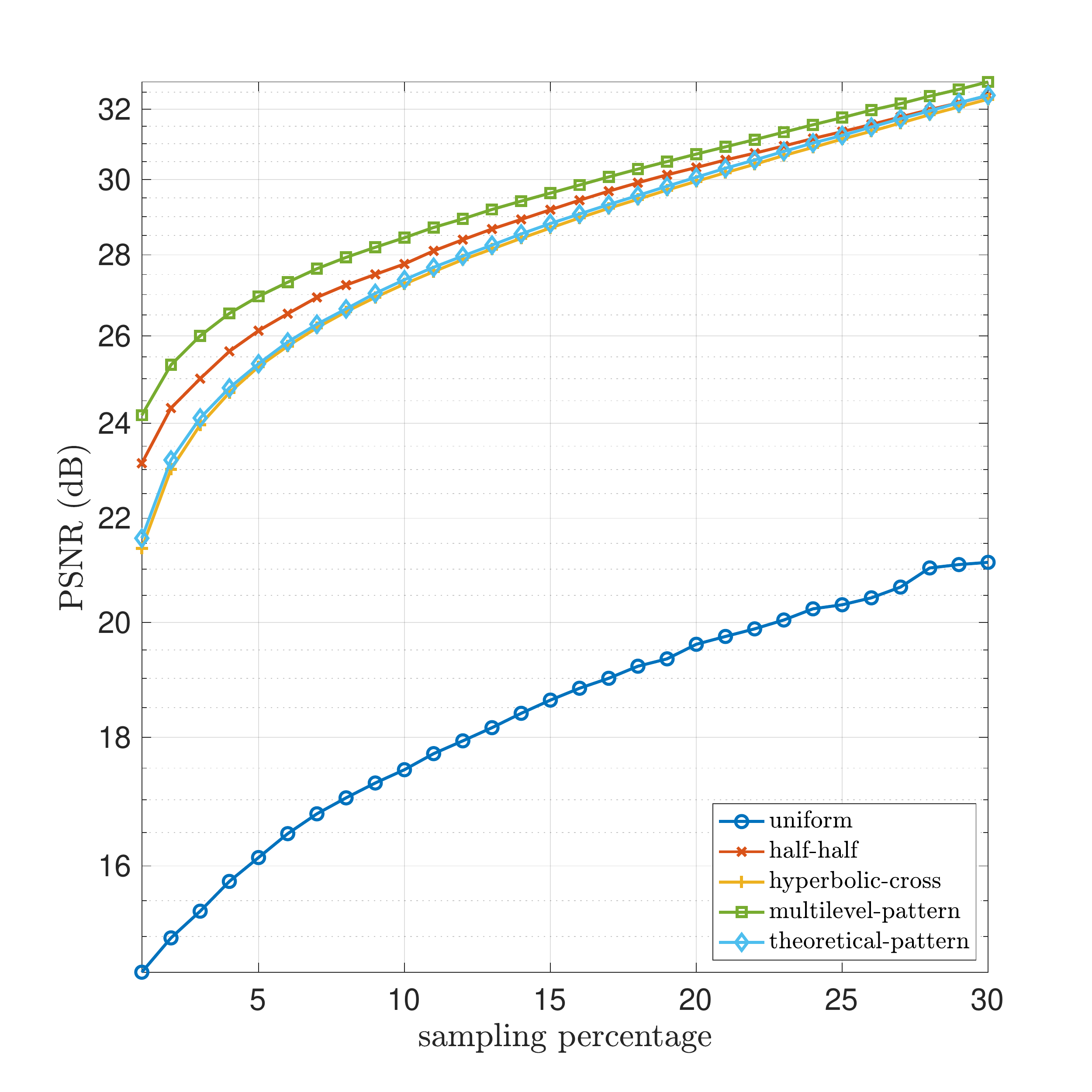}
		\includegraphics[width=0.23\paperwidth,clip=true,trim=8mm 6mm 18mm 15mm]{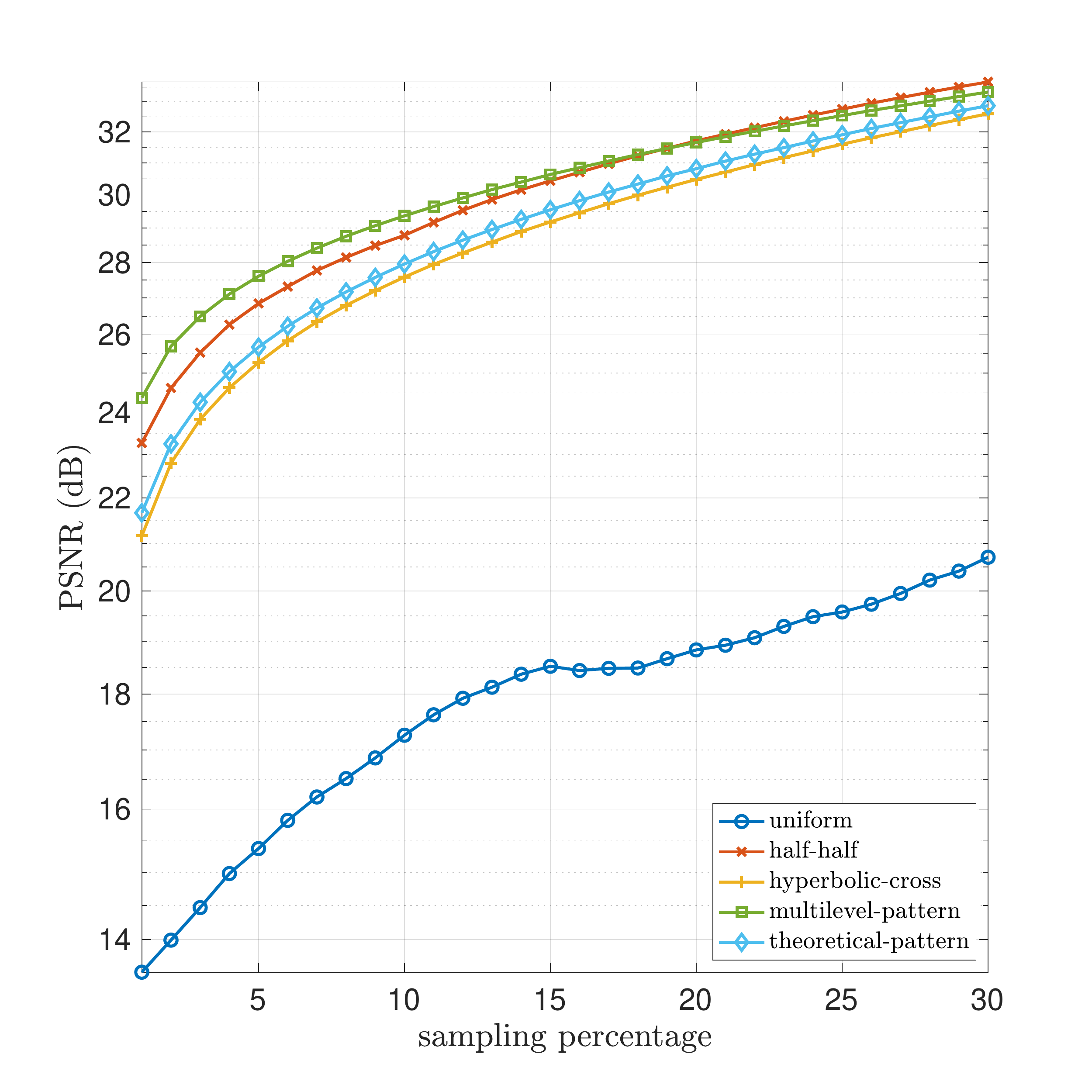}
	\end{center}
	\vspace{-2mm}
	\caption{Comparison of the average PSNR values over 20 trials for various sampling patterns in reconstructing the {\bf(left)} ``cameraman,'' {\bf(middle)} ``donkey,'' and {\bf(right)} ``man'' test images with Walsh sampling. Here the multilevel sampling is performed with $a=2$, $r=30$ and $r_0=2$.}
	\label{fig:Walsh_samp_comp}
\end{figure}

This observation is related to the previous discussion in \S \ref{s:limitations}. As originally considered in \cite{PoonTV}, the optimal sampling strategy in practice depends on the image, resolution and sampling percentage -- in particular, the geometry of its edges. This is not reflected in our theoretically-optimal sampling strategies (which are image independent). Yet it is notable that good all-round performance can be achieved with the multilevel random sampling strategy.

\begin{figure}
	\begin{center}
		\includegraphics[width=0.23\paperwidth,clip=true,trim=8mm 6mm 18mm 15mm]{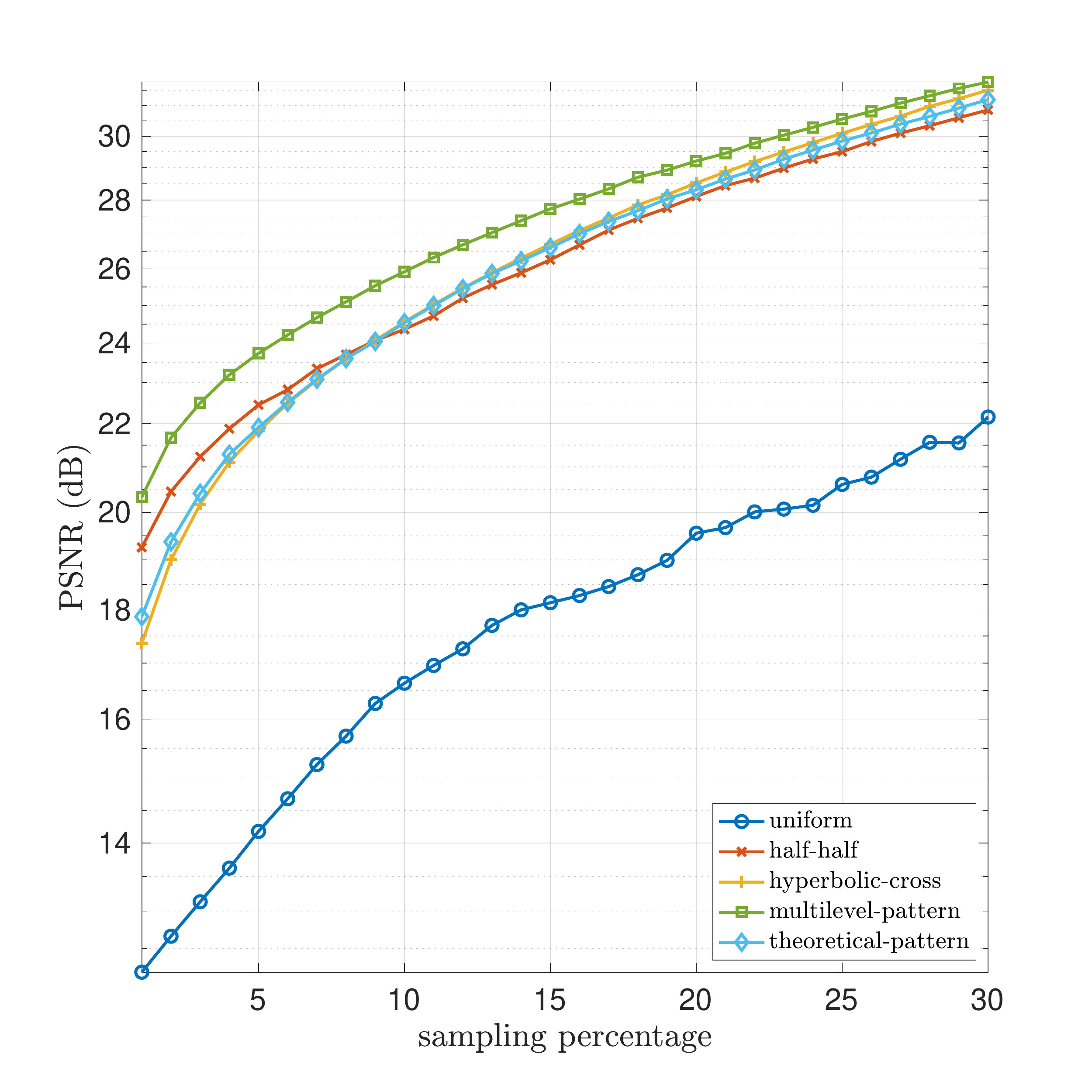}
		\includegraphics[width=0.23\paperwidth,clip=true,trim=8mm 6mm 18mm 15mm]{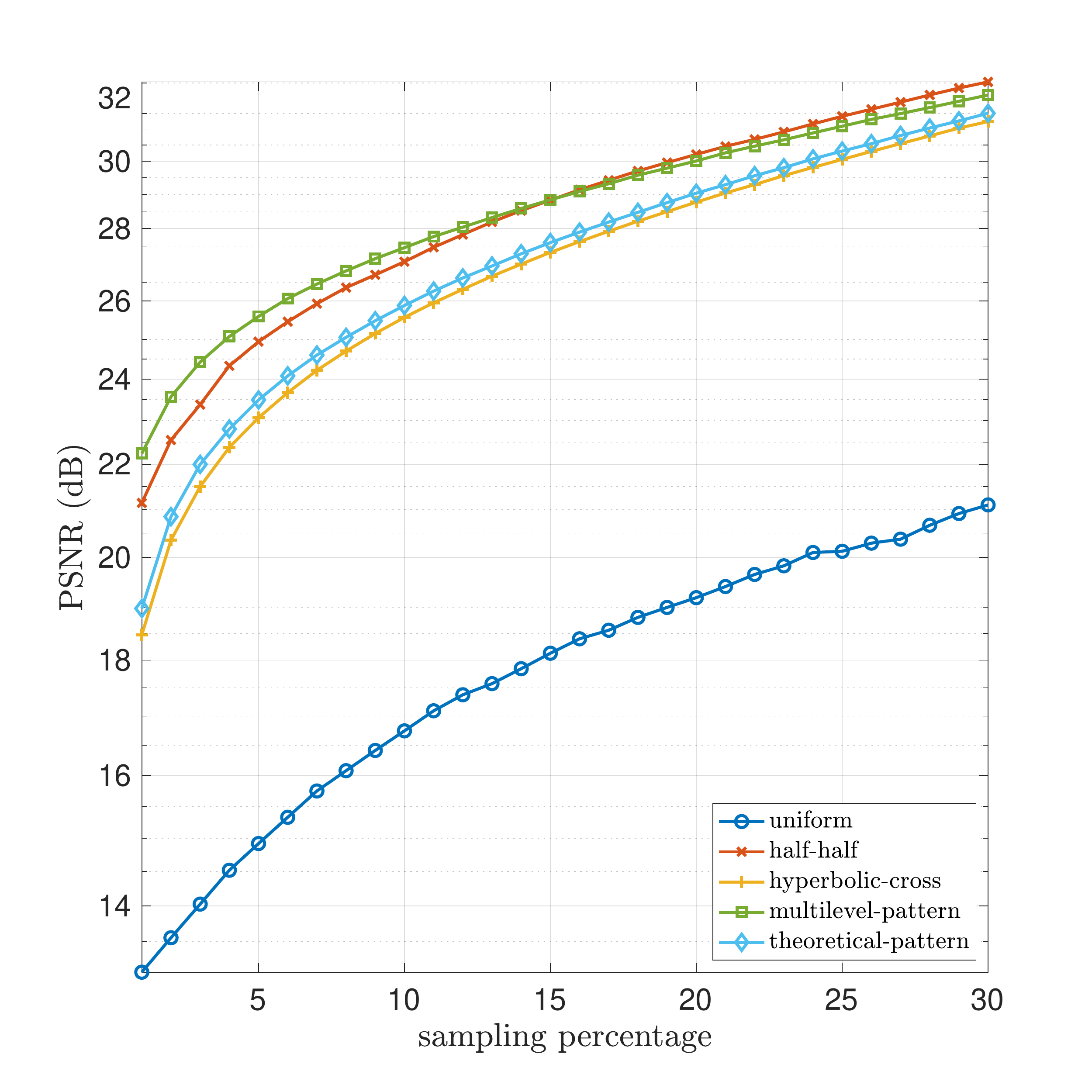}
		\includegraphics[width=0.23\paperwidth,clip=true,trim=8mm 6mm 18mm 15mm]{man_1024_Walsh_experiment_PSNR_plots.pdf}
	\end{center}
	\vspace{-2mm}
	\caption{Comparison of the average PSNR values over 20 trials for various sampling patterns in reconstructing  the ``man'' test image with {\bf(left)} $N=256$ {\bf(middle)} $N=512$, and {\bf(right)} $N=1024$ with Walsh sampling. Here the multilevel sampling is performed with $a=2$, $r=30$ and $r_0=2$.}
	\label{fig:Walsh_man_img_size_comp}
\end{figure}

\section{Proofs Part I: Theorems \ref{t:TVuniform1D}--\ref{t:TVVDSdD}}\label{s:proofsFourier}

The remainder of this paper is devoted to the proofs of the main results. We divide this into two sections: Fourier sampling in this section and Walsh sampling in the next.  Note that in both these sections we rely on some background results which are found in the Supplementary Materials. In \S \ref{s:proofsrest} we also prove several of the ancillary lemmas stated in previously. 

\subsection{Overview}

Our proof is divided into three parts.  First, in \S \ref{ss:gradrecovproof}, we assert stable and robust recovery of the gradient $\nabla x$.  Following \cite{PoonTV}, this made possible by the uniform random samples $\Omega_1$ and relies crucially on the commuting property of the Fourier transform and the gradient operator (Lemma \ref{l:commuting}). 

Next, in the \S \ref{ss:imageunif} and \S \ref{ss:imagevds}, we address the recovery of the image itself.  In the case of uniform random sampling, we follow \cite{PoonTV} and use the following discrete \textit{Poincar\'e inequality}
\be{
\label{discPoin}
\begin{split}
\nm{z}_{\ell^2} &\leq \nm{A z}_{\ell^2} + \sqrt{N} \nm{z}_{\mathrm{TV}},\quad \forall z \in \bbC^N,
\\
\nm{z}_{\ell^2} &\leq \nm{A z}_{\ell^2} + 2^{1-d/2}\nm{z}_{\mathrm{TV}_a},\quad \forall z \in \bbC^{N^d}.
\end{split}
}
See Lemma \ref{mpi}. The estimates for $\nm{\hat{x} - x}_{\ell^2}$ then follow by setting $z = \hat{x} - x$ and using the existing gradient error bounds.
For variable density sampling in \S \ref{ss:imagevds}, based on ideas of \cite{NeedellWardTV2,NeedellWardTV1}, we derive a strengthened Poincar\'e inequality for any measurement matrix that is incoherent with Haar wavelets (Lemma \ref{pihim}). The rest of the proof is then devoted to showing that variable density Fourier samples are sufficiently incoherent with Haar wavelets. For this we use tools from \S \ref{ss:BOS}.

\subsection{Gradient recovery}\label{ss:gradrecovproof}

In this section, we prove the error bounds \eqref{TVUnif1Dgraderr}, \eqref{uge2}, \eqref{uge22}, \eqref{TVVDS1Dgraderr}, \eqref{grani} and \eqref{gri}  for gradient recovery using uniform random and variable density Fourier sampling. This relies on the commuting property: 

\begin{lemma}
[Commuting property]
\label{l:commuting}
Let $F^{(d)}$ be the $d$-dimensional DFT matrix and $\nabla_j$ be the $j^{\rth}$ partial derivative operator. Then
\begin{equation*}
F^{(d)}\nabla_j=(I^{(d-j)}\otimes\mathrm{diag}(\lambda)\otimes I^{(j-1)})F^{(d)},
\end{equation*}
where $\lambda=(\lambda_j)_{j=1}^N\in\mathbb{C}^N$ has entries $\lambda_j=\exp ( 2\pi \mathrm{i}\varrho(j)/N )-1$ and $\varrho$ is defined in \eqref{fourier1}). $I^{(d-j)}=\underbrace{I\otimes I\otimes\cdots\otimes I}_{d-j}$, $I^{(j-1)}= \underbrace{I\otimes\cdots\otimes I}_{j-1}$, and $I\in\mathbb{C}^{N\times N}$ is the identity matrix.
\end{lemma}
\begin{proof}
The $d = 1$ case is a simple exercise.
Now consider the $d \geq 2$ case. We have
\begin{align*}
F^{(d)}\nabla_j&=(\underbrace{F\otimes F\otimes\cdots\otimes F}_{d})(\underbrace{I\otimes I\otimes\cdots\otimes I}_{d-j}\otimes \nabla \otimes \underbrace{I\otimes\cdots\otimes I}_{j-1})\\
&=F^{(d-j)}\otimes((\mathrm{diag}(\lambda)F)\otimes (F^{(j-1)}  I^{(j-1)}))\\
&=(I^{(d-j)} F^{(d-j)})\otimes(\mathrm{diag}(\lambda)\otimes I^{(j-1)} F^{(j)}) =I^{(d-j)}\otimes\mathrm{diag}(\lambda)\otimes I^{(j-1)}F^{(d)},
\end{align*}
as required.
\end{proof}

Next, since the sampling map has the form $\Omega=\Omega_1\cup\Omega_2$ in all cases, we can write $A$ as 
\begin{equation*}
A=\left(\begin{array}{c} A_1\\A_2 \end{array} \right)\qquad A_i=\frac{1}{\sqrt{m}}P_{\Omega_i}F,\quad i=1,2.
\end{equation*}
The matrix $N^{-d/2} F$ is unitary, and therefore $A_1$ is a randomly-subsampled unitary matrix (see \S \ref{s:CSpreliminaries}) with the uniform probability distribution $q = (1/N^d)^{N^d}_{i=1}$. Since $|F_{jk}| = 1$, the bounded orthonormal system constant $\Theta = 1$. Hence, by \eqref{BOSRIPcond}, $A_1$ satisfies the RIP of order $2s$ with $\delta_{2s} \leq 1/2$ (this factor is arbitrary, any number less than $4 /\sqrt{41}$ will do) and probability at least $1-\varepsilon$, provided (after simplifying the log factor using the fact that $N^d \geq s \geq 2$),
\bes{
m \gtrsim s \cdot \log(s) \cdot ( \log(s) \cdot \log(N^d) + \log(\varepsilon^{-1}) ).
}
For the next steps of the proof, we split into the anisotropic and isotropic cases.

\subsubsection{Anisotropic TV: \eqref{TVUnif1Dgraderr},\eqref{TVVDS1Dgraderr}, \eqref{uge2} and \eqref{grani}}

\begin{proof}[Proof of \eqref{TVUnif1Dgraderr}, \eqref{TVVDS1Dgraderr}, \eqref{uge2}and \eqref{grani}]
We use Lemma \ref{l:blockArNSP}
Let $B = B_1$ be as in this lemma with $A = A_1$. Since $A_1$ has the RIP of order $2s$ with constant $\delta_{2s} \leq 1/2$, it also has the rNSP of order $s$.
We now apply Lemma \ref{l:rNSPdistance} to $B_1$ with the vectors $\nabla\hat{x}$ and $\nabla x$ and use the fact that $\nmu{\nabla \hat{x}}_1=\nmu{\hat{x}}_{\mathrm{\mathrm{TV}}_a}\leq \nmu{x}_{\mathrm{\mathrm{TV}}_a}= \nmu{\nabla x}_{\ell^1}$ to get
\eas{
\nm{\nabla \hat{x}-\nabla x }_{\ell^1} &\lesssim \sigma_s(\nabla x)_{\ell^1}+\sqrt{sd} \nm{ B_1\nabla (\hat{x}-x) }_{\ell^2},
\\
\nm{ \nabla\hat{x}-\nabla x }_{\ell^2} &\lesssim \frac{\sigma_s(\nabla x)_{\ell^1}}{\sqrt{s}} + \sqrt{d} \nm{ B_1\nabla (\hat{x}-x) }_{\ell^2}.
}
For the second term, we use the fact that $F^{(d)}\nabla_i=(I^{(d-i)}\otimes\mathrm{diag}(\lambda)\otimes I^{(i-1)})F^{(d)}$ (Lemma \ref{l:commuting}) and the bound $\nm{\lambda}_{\ell^{\infty}}\leq 2$ to get
\begin{align*}
\nmu{B_1\nabla (\hat{x}-x)}_{\ell^2}&=\sqrt{\sum_{i=1}^d \nmu{A_1\nabla_i(\hat{x}-x)}^2_{\ell^2}}=\sqrt{\sum_{i=1}^d \nm{\frac{1}{\sqrt{m}}P_{\Omega_1}F^{(d)}\nabla_i (\hat{x}-x) }^2_{\ell^2}}\\
&\leq 2\sqrt{d} \nm{\frac{1}{\sqrt{m}}P_{\Omega_1}F^{(d)}(\hat{x}-x) }_{\ell^2}\\
&=2\sqrt{d} \nmu{A_1(\hat{x}-x) }_{\ell^2} \leq 2\sqrt{d} \nmu{A(\hat{x}-x) }_{\ell^2} \leq4\sqrt{d}\eta.
\end{align*}
Note that in the last step we have used the fact that  $\hat{x}$ and $x$ are feasible for \eqref{TVminprob}. Substituting this into the previous estimates now gives the result.
\end{proof}

\subsubsection{Isotropic TV: \eqref{uge22} and \eqref{gri}}

For isotropic TV, we use the matrix recovery techniques from \S \ref{ss:matrixrecov}.

\begin{proof}[Proof of \eqref{uge22} and \eqref{gri}]

The matrix $A_1$ satisfies the RIP of order $2s$ with constant $\delta_{2s} \leq 1/2$. Hence, by \cite[Prop.\ 4.3]{dexter2018mixed} it also has the $\ell^{2,2}$-rNSP of order $s$ with constants $\rho$ and $\gamma$ depending on $\delta_{2s}$. Using $\nabla\hat{x}$ and $\nabla{x}$ in Lemma \ref{rnsp21}, we get
\begin{equation}
\label{mge1}
\nm{\nabla \hat{x}-\nabla x}_{\ell^{2,1}}\lesssim \sigma_s(\nabla x)_{\ell^{2,1}}+\sqrt{s}\nm{A_1\nabla (\hat{x}-x)}_{\ell^{2,2}},
\end{equation}
\begin{equation}
\label{mge2}
\nm{\nabla \hat{x}-\nabla x}_{\ell^{2,2}}\lesssim \frac{\sigma_s(\nabla x)_{\ell^{2,1}}}{\sqrt{s}}+\nm{A_1\nabla (\hat{x}-x)}_{\ell^{2,2}}.
\end{equation}
For the second term of \eqref{mge1} and \eqref{mge2}, by the commuting property, we have, 
\begin{align*}\nm{A_1\nabla (\hat{x}-x)}_{\ell^{2,2}}&=\nm{(A_1\nabla_1(\hat{x}-x)~\cdots~A_1\nabla_d(\hat{x}-x))}_{\ell^{2,2}}\\
&=\nm{(\frac{1}{\sqrt{m}}P_{\Omega_1}F^{(d)}\nabla_1 (\hat{x}-x)~\cdots~\frac{1}{\sqrt{m}}P_{\Omega_1}F^{(d)}\nabla_d (\hat{x}-x))}_{\ell^{2,2}}\\
&\leq 2\nm{(A_1(\hat{x}-x)~\cdots~A_1(\hat{x}-x))}_{\ell^{2,2}}\\
&\leq 2\nm{(A(\hat{x}-x)~\cdots~A(\hat{x}-x))}_{\ell^{2,2}} \leq 4\sqrt{d}\eta.
\end{align*}
In the last step we used the fact that $\hat{x}$ and $x$ are feasible for \eqref{TVminprob}. Substituting this in \eqref{mge1} and \eqref{mge2} and recalling that $\nm{\cdot}_{\mathrm{TV}_i} =  \nm{\nabla \cdot}_{\ell^{2,1}}$ yields \eqref{uge22}and \eqref{gri}.
\end{proof}

\subsection{Image recovery for uniform random sampling}\label{ss:imageunif}

We now prove \eqref{unia} and \eqref{unii}. These proofs rely on a discrete Poincar\'e inequality (Lemma \ref{l:poincare}). To prove this, as well as several later results, we will follow ideas from \cite{NeedellWardTV2,NeedellWardTV1} and relate the TV semi-norm to the decay rate of Haar wavelet coefficients.  For notation and background on Haar wavelets, see \S \ref{ss:Haarwave}.

\begin{lemma}[Discrete Poincar\'e inequality]
\label{l:poincare}
 Let $x\in\mathbb{C}^{N^d}$ with $\sum_{i=1}^{N^d}x_{i}=0$. Then
\eas{
\nm{x}_{\ell^2} \leq \sqrt{N} \nm{x}_{\mathrm{TV}},\quad d = 1,
\qquad
\nm{x}_{\ell^2} \lesssim \frac{\nm{x}_{\mathrm{TV}a}}{2^{d/2-1}} \leq \frac{\sqrt{d}\nm{x}_{\mathrm{TV}i}}{2^{d/2-1}},\quad d \geq 2.
}
\end{lemma}
\begin{proof}
See \cite[Lem.\ 4.1]{PoonTV} for the $d = 1$ result. Now consider $d \geq 2$.
Let $x\in\mathbb{C}^{N^d}$ with mean zero and $f$ be its isometric embedding, i.e. $f(i/N)=N^{d/2}X_i$
where $i=(i_1,\ldots,i_d)\in \{0,\ldots,N-1\}^d$ and $x = \mathrm{vec}(X)$. Note that $\nm{f}_{L^2}=\nm{x}_{\ell^2}$ and $f$ also has mean zero. Let $c^{(e)}_{j,n}$ denote the Haar wavelet coefficient of $f$. Since $f$ has mean zero, we have $c^{(0)}_{0,0} = 0$. Write $c_{j,n} \in \bbC^{2^d-1}$ for the vector of coefficients $c^{(e)}_{j,n}$ with $e \in \{0,1\}^d \backslash \{0\}$. Then
Lemma \ref{bvtv} and Lemma \ref{cbbv} give that when $d\geq 2$, there exists a constant $C>0$ such that 
\begin{equation*}
\nmu{c_{(k)}}_{\ell^2} \leq C\frac{2^{j(d-2)/2}|f|_{BV}}{k},\qquad |f|_{BV}\leq N^{-d/2+1}\nm{x}_{\mathrm{TV}a}.
\end{equation*}
Since $2^j\leq N/2$ we have
\begin{equation}
\label{haardecay}
\nmu{c_{(k)}}_{\ell^2}  \leq C\frac{\nm{x}_{\mathrm{TV}a}}{k\cdot 2^{d/2-1}} \leq  C\frac{\sqrt{d}\nm{x}_{\mathrm{TV}i}}{k\cdot 2^{d/2-1}},
\end{equation}
where in the second inequality we use \eqref{TVairelate}. Therefore
\begin{align*}
\nm{x}_{\ell^2}&=\nm{f}_{L^2}=\sqrt{\sum^{\infty}_{k = 1} \nm{c_{(k)} }^2_{\ell^2} }\leq C\frac{ \nm{x}_{\mathrm{TV}a}}{ 2^{d/2-1}}\leq C\frac{\sqrt{d} \nm{x}_{\mathrm{TV}i}}{ 2^{d/2-1}},
\end{align*}
as required.
\end{proof}

This now gives the following:
\begin{lemma}
\label{mpi}
Let $A$ be the measurement matrix of Theorem \ref{t:TVuniform1D}. Then
\bes{
\nm{z}_{\ell^2}\leq \nm{Az}_{\ell^2}+\sqrt{N} \nm{z}_{\mathrm{TV}},\qquad \forall z\in\mathbb{C}^{N}.
}
If $A$ is the measurement matrix of Theorem \ref{t:TVuniformdD} then
\begin{align*}
\nm{z}_{\ell^2}&\leq \nm{Az}_{\ell^2}+2^{1-d/2}\nm{z}_{\mathrm{TV}_a} \leq\nm{Az}_{\ell^2}+2^{1-d/2}\sqrt{d}\nm{z}_{\mathrm{TV}_i},\qquad \forall z\in\mathbb{C}^{N^d}.
\end{align*}
\end{lemma}

\begin{proof}
Consider the case $d \geq 2$ first.
Let $z\in\mathbb{C}^{N^d}$ and define $\bar{z}=(\bar{z}_i)_{i=1}^{N^d}$ with $\bar{z}_i=z_i-\frac{1}{N^d}\sum_{j=1}^{N^d}z_j$. Then we have $\sum_{i=1}^{N^d}\bar{z}_i=0$ and applying the Poincar$\mathrm{\acute{e}}$ inequality gives
\begin{equation*}
\nm{\bar{z}}_{\ell^2}\lesssim 2^{1-d/2}\nm{\bar{z}}_{\mathrm{TV}_a}=2^{1-d/2} \nm{z}_{\mathrm{TV}_a}.
\end{equation*}
Since  $\sum_{j=1}^{N^d}z_j=(Fz)_0=\sqrt{m}A_2z$ and $m\leq N^d$ by assumption, we have
\begin{align*}
\nm{z}_{\ell^2}\leq \frac{1}{\sqrt{N^d}} \nm{(Fz)_0 }_{\ell^2}+2^{1-d/2} \nm{z}_{\mathrm{TV}_a}
&=\sqrt{\frac{m}{N^d}} \nm{A_2z}_{\ell^2}+2^{1-d/2} \nm{z}_{\mathrm{TV}_a}\\
&\leq \nm{Az}_{\ell^2}+2^{1-d/2} \nm{z}_{\mathrm{TV}_a}.
\end{align*}
This gives the first inequality. The second follows from \eqref{TVairelate}. When $d = 1$ we use the same arguments, replacing the Poincar\'e inequality by its one-dimensional version (Lemma \ref{l:poincare}).
\end{proof}

\begin{proof}[Proof of \eqref{TVUnif1Dsigerr}, \eqref{unia} and \eqref{unii}]
We use Lemma \ref{mpi} with $z=\hat{x}-x$. This gives
\begin{align*}
\nm{\hat{x}-x }_{\ell^2}&\leq  \nm{A(\hat{x}-x) }_{\ell^2}+2^{1-d/2} \nm{\hat{x}-x }_{\mathrm{TV}_a} \lesssim 2^{-d/2}\sigma_s(\nabla x)_{\ell^1}+(1+2^{-d/2}\sqrt{s}d)\eta,
\end{align*}
when $d \geq 2$, which yields \eqref{unia}. Here, for the second inequality, we use \eqref{uge2} and the fact that $\hat{x}$ and $x$ are feasible, so that $\nm{A (\hat{x} - x)}_{\ell^2} \leq 2 \eta$. This isotropic case \eqref{unii} is identical. To prove \eqref{TVUnif1Dsigerr}, we use Lemma \ref{mpi} with $d = 1$ and \eqref{TVUnif1Dgraderr}.
\end{proof}

\subsection{Image recovery for variable density Fourier sampling}\label{ss:imagevds}

We now consider variable density samples. We first show a strengthened Poincar\'e inequality for Haar-incoherent measurements, and then derive conditions under which this holds for variable density samples.

\begin{lemma}
[Poincar\'e inequality for Haar-incoherent measurements]
\label{pihim}
Let $W\in\mathbb{R}^{N^d\times N^d}$ be the matrix of the $d$-dimensional discrete Haar wavelet transform and $B\in\mathbb{C}^{m\times N^d}$. Suppose that $BW$ satisfies the RIP of order $5k$ with constant $\delta_{5k} < 1/3$. Then
\bes{
\nm{x}_{\ell^2}\lesssim \nm{Bx}_{\ell^2}+\frac{\sqrt{N} \nm{x}_{\mathrm{TV}}}{k},\qquad d = 1,
}
and
\begin{equation*}
\nm{x}_{\ell^2}\lesssim_d \nm{Bx}_{\ell^2}+\frac{\nm{x}_{\mathrm{TV}_a}}{\sqrt{k}}\log\left(\frac{N}{k}\right),\qquad d \geq 2.
\end{equation*}
\end{lemma}
\begin{proof}
We may assume without loss of generality that $x$ has mean zero. Let $A = B W$, $c = W^* x$ and $\Delta$ be the index set of the largest $k$ entries of $c$ in absolute value. Then Lemma \ref{l:NWprop} with the trivial choices $\gamma = 1$ and $\sigma = \nmu{P^{\perp}_{\Delta} c}_{\ell^1}$ gives
\bes{
\nm{x}_{\ell^2} = \nm{c}_{\ell^2} \lesssim \frac{\nmu{P^{\perp}_{\Delta} c}_{\ell^1}}{\sqrt{k}} + \nm{A c}_{\ell^2} = \frac{\nmu{P^{\perp}_{\Delta} c}_{\ell^1}}{\sqrt{k}} + \nm{B x}_{\ell^2}
}
Now, as in the proof of Lemma \ref{l:poincare}, let $c_{(k)} \in \bbC^{2^d-1}$ denote $k^{\rth}$ largest wavelet coefficient block in $c$. Then
\bes{
\nmu{c_{(k)}}_{\ell^2} \lesssim_d \left \{ \begin{array}{cc} \sqrt{N} \nm{x}_{\mathrm{TV}_a} / k^{3/2}  & d = 1  \\ \nm{x}_{\mathrm{TV}_i} / k & d \geq 2 \end{array} \right . .
}
Hence, when $d = 1$, we have 
$
\nmu{P^{\perp}_{\Delta} c}_{\ell^1} \lesssim \sqrt{N} \nm{x}_{\mathrm{TV}} \sum_{i > k} i^{-3/2} \lesssim \sqrt{N/k} \nm{x}_{\mathrm{TV}},
$
as required. When $d \geq 2$, since $\Delta$ contains the index set of the largest $k$ entries of $c$, we can bound $\nmu{P^{\perp}_{\Delta} c}_{\ell^1}$ by $\nmu{P^{\perp}_{\Delta'} c}_{\ell^1}$, where $\Delta'$ contains the indices of the largest $\lfloor k / (2^d-1) \rfloor$ of $c$. Hence
\bes{
\nmu{P^{\perp}_{\Delta} c}_{\ell^1} \lesssim_d \nm{x}_{\mathrm{TV}_a} \sum^{N}_{i = \lfloor k / (2^d-1) \rfloor+1 } i^{-1} \lesssim_d \nm{x}_{\mathrm{TV}_a} \log(N/k),
}
as required.
\end{proof}

\begin{lemma}
[The RIP for the Fourier--Haar matrix]
\label{ripfhm}
Let $0<\delta, \varepsilon<1$, $2\leq s\leq N^d$, $\Omega\subseteq\{-\frac{N}{2}+1,\ldots,\frac{N}{2}\}^d$ be a $d$-dimensional variable sampling pattern corresponding to a probability distribution $p = (p_{\omega})$, with $\Gamma(p)$ as in \eqref{Cpdef} and $D\in\mathbb{C}^{N^d\times N^d}$ be the  diagonal matrix with entries $D_{ii}=\frac{1}{\sqrt{p_{\varrho^{-1}(i)}}}$, where $\varrho = \varrho^{(d)}$ is the bijection \eqref{fourierd}. Suppose that
\begin{equation*}
m \gtrsim_d  \Gamma(p)\cdot s \cdot \log(\Gamma(p) s) \cdot \left ( \log(\Gamma(p) s) \cdot \log(N) + \log(\epsilon^{-1}) \right ).
\end{equation*}
Then, with probability at least $1-\varepsilon$, the matrix
\be{
\label{FHaarnormalized}
\frac{1}{\sqrt{mN^d}}P_{\Omega}DFW,
}
has the RIP of order $s$ with constant $\delta_s\leq 1/2$, where $F$ and $W$ are the discrete Fourier and Haar wavelet transforms respectively.
\end{lemma}

Note that the factor $1/2$ here is arbitrary. To prove this, we use the tools introduced in \S \ref{ss:BOS}. To this end, we first require an upper bound on the Fourier transform of the discrete Haar wavelet $\phi^{(e)}_{j,n}$. For this, we use Lemma \ref{fhmeb}.

\begin{proof}[Proof of Lemma \ref{ripfhm}]
Since $N^{-d/2} F W = U$ is unitary, the matrix \eqref{FHaarnormalized} is a randomly-subsampled unitary matrix in the sense of \S \ref{s:CSpreliminaries}. Hence it has the RIP of order $s$ provided \eqref{BOSRIPcond} holds, where $\Theta$ is as in \eqref{Thetadef}. In particular, it suffices to show that $\Theta \leq \sqrt{\Gamma(p)}$. Indeed, if this holds, the log factor in \eqref{BOSRIPcond} simplifies, since $s \geq 2$ and $\Gamma(p) \geq 1$ (this follows from \eqref{Cpdef} and the fact that $p$ is a probability distribution).
Using Lemma \ref{fhmeb} and tensor-product nature of the Fourier transform and Haar wavelets, we see that
\be{
\label{Thetabdtemp1}
\Theta \lesssim_{d} \max_{\substack{ \omega = (\omega_1,\ldots,\omega_d) \\ -N/2 < \omega_1,\ldots,\omega_d \leq N/2}} \max_{j = 0,\ldots,r-1} \left \{ \frac{1}{\sqrt{p_{\omega}} } \prod^{d}_{i=1}  \frac{2^{j/2}}{\max \{ \overline{\omega_i} , 2^j \} } \right \}.
}
Consider the product term on the right-hand side. Let $\pi : \{1,\ldots,d\} \rightarrow \{1,\ldots,d\}$ be a nonincreasing rearrangement of the the $\overline{\omega_i}$, and let $0 \leq l \leq d+1$ be such that
$
\overline{\omega_{\pi(l)}} \geq 2^j \geq \overline{\omega_{\pi(l+1)}}.
$
Note that if $l = 0 $ this means $2^j \geq \overline{\omega_{\pi(1)}}$ and if $l = d+1$ this means $\overline{\omega_{\pi(d)}} \geq 2^j$. Then
\eas{
\prod^{d}_{i=1} \frac{2^{j/2}}{\max \{  \overline{\omega_i} , 2^{j} \}} = \prod^{l}_{i=1} \frac{2^{j/2}}{\overline{\omega_{\pi(i)}}} \prod^{d}_{i=l+1} \frac{2^{j/2}}{2^j} = \frac{2^{j(l-d/2)}}{\overline{\omega_{\pi(1)}} \cdots \overline{\omega_{\pi(l)}} }.
}
Suppose first that $d$ is even. Then, since $\overline{\omega_{\pi(i)}} \geq 2^j$ for $i = 1,\ldots,l$, we can use the smallest $l-d/2$ such terms to bound the denominator, giving
\bes{
\prod^{d}_{i=1} \frac{2^{j/2}}{\max \{  \overline{\omega_i} , 2^{j} \}} \leq \frac{1}{\overline{\omega_{\pi(1)}} \cdots \overline{\omega_{\pi(d/2)}} }
}
If $d$ is odd, then by a similar argument we obtain
\bes{
\prod^{d}_{i=1} \frac{2^{j/2}}{\max \{  \overline{\omega_i} , 2^{j} \}} \leq \frac{1}{\overline{\omega_{\pi(1)}} \cdots \overline{\omega_{\pi((d-1)/2)}} \sqrt{\overline{\omega_{\pi((d+1)/2)}}} }.
}
Hence, recalling \eqref{qomdef1}--\eqref{qomdef2}, and returning to \eqref{Thetabdtemp1}, we deduce that
\bes{
\Theta \lesssim_d \max_{\substack{ \omega = (\omega_1,\ldots,\omega_d) \\ -N/2 < \omega_1,\ldots,\omega_d \leq N/2}} \left \{ \frac{1}{q_{\omega} \sqrt{p_{\omega}}} \right \} \leq \sqrt{\Gamma(p)},
}
as required.
\end{proof}

We now return to the final arguments. We first require the following:

\begin{lemma}
Under the conditions of Theorem \ref{t:TVVDS1D}, the following holds with probability at least $1-\varepsilon/2$:
\be{
\label{final10}
\nm{x}_{\ell^2}\lesssim \sqrt{\Gamma(p)}\nm{Ax}_{\ell^2}+\frac{\sqrt{N}\nm{x}_{\mathrm{TV}}}{s},\qquad \forall x \in \bbC^{N}.
}
Under the conditions of Theorem \ref{t:TVVDSdD}, the following holds with probability at least $1-\varepsilon/2$:
\begin{equation}
\label{final1}
\nm{x}_{\ell^2}\lesssim \sqrt{\Gamma(p)}\nm{Ax}_{\ell^2}+\frac{\nm{x}_{\mathrm{TV}_a}}{\sqrt{s}},\qquad \forall x \in \bbC^{N^d}.
\end{equation}
\end{lemma}
\begin{proof}
Consider the first case. The condition \eqref{m1DVDS} and Lemma  \ref{ripfhm} give that the matrix $BW=\frac{1}{\sqrt{mN^d}}P_{\Omega_2}DFW$ has the RIP of order $2s+1$ with constant $\delta_{2s+1} \leq 1/2$. Hence Lemma \ref{pihim} gives that
\bes{
\nm{x}_{\ell^2} \lesssim \nm{Bx}_{\ell^2}+\frac{\sqrt{N}\nm{x}_{\mathrm{TV}}}{s},\qquad \forall x \in \bbC^N.
}
For the second case, the condition \eqref{nomvds} and Lemma \ref{ripfhm} give that the matrix $BW=\frac{1}{\sqrt{mN^d}}P_{\Omega_2}DFW$ has the RIP of order $k$ with constant $\delta_{2k+1} \leq 1/2$, where $k=\lceil sd^2(\log N)^2\rceil$. Hence Lemma \ref{pihim} gives that
\begin{align*}
\nm{x}_{\ell^2}&\lesssim \nm{Bx}_{\ell^2}+\frac{\nm{x}_{\mathrm{TV}_a}}{\sqrt{s}},\qquad \forall x \in \bbC^{N^d}.
\end{align*}
Thus, it remains to show that $\nm{B x}_{\ell^2} \leq \sqrt{\Gamma(p)} \nm{A x}_{\ell^2}$.
Observe that $B=\frac{1}{\sqrt{mN^d}}P_{\Omega_2}DF=\frac{1}{\sqrt{N^d}}DA_2$. Therefore
\begin{align*}
\nm{Bx}_{\ell^2}&\leq \frac{1}{\sqrt{N^d}}\nm{D}_{\ell^2}\nm{Ax}_{\ell^2} \leq \frac{1}{\sqrt{N^d}\min_{\omega}\{\sqrt{p_{\omega}}\}}\nm{Ax}_{\ell^2}\leq\sqrt{\Gamma(p)}\nm{Ax}_{\ell^2}.
\end{align*}
Here, in the penultimate step we use \eqref{Cpdef} and the definition of $q_{\omega}$ to write
\bes{
\frac{1}{\sqrt{N^d} \sqrt{p_{\omega}}} \leq \sqrt{\Gamma(p)} \frac{q_{\omega}}{N^{d/2}} \leq \sqrt{\Gamma(p)}.
}
The result now follows.
\end{proof}

\begin{proof}[Proof of  \eqref{TVVDS1Dsigerr}, \eqref{srani} and \eqref{sri}]
We consider the case $d \geq 2$. The case $d = 1$ is identical.
As shown in \S \ref{ss:gradrecovproof}, the gradient error bounds \eqref{grani} and \eqref{gri} hold with probability at least $1-\varepsilon /2$. Hence, the bounds \eqref{grani}, \eqref{gri} and \eqref{final1} hold simultaneously with probability at least $1-\varepsilon$. We now apply \eqref{final1} to $\hat{x} - x$ to get
\eas{
\nmu{\hat{x}-x }_{\ell^2}&\lesssim\sqrt{\Gamma(p)} \nm{A(\hat{x}-x) }_{\ell^2}+\frac{ \nmu{\hat{x}-x }_{\mathrm{TV}_a}}{\sqrt{s}} \lesssim \sqrt{\Gamma(p)} \eta + \frac{\nmu{\hat{x}-x }_{\mathrm{TV}_a}}{\sqrt{s}} .
}
Hence \eqref{srani} follows from \eqref{grani}. For \eqref{sri}, we use \eqref{final1} and the inequality $\nmu{\hat{x} - x}_{\mathrm{TV}_a} \leq \sqrt{d} \nmu{\hat{x} - x}_{\mathrm{TV}_i}$ to get 
\eas{
\nmu{\hat{x}-x }_{\ell^2}&\lesssim\sqrt{\Gamma(p)} \nm{A(\hat{x}-x) }_{\ell^2}+\frac{ \nmu{\hat{x}-x }_{\mathrm{TV}_a}}{\sqrt{s}} \lesssim \sqrt{\Gamma(p)} \eta + \sqrt{d} \frac{\nmu{\hat{x}-x }_{\mathrm{TV}_i}}{\sqrt{s}} .
}
The result then follows from \eqref{gri}.
\end{proof}

\section{Proofs Part II: Theorem \ref{t:WalshTV}}\label{s:proofsII}

Since we no longer have the commuting property, our proof strategy is based on ideas from \cite{NeedellWardTV1}, see also \cite{KrahmerWardCSImaging}. In particular,  we first show the following result, which extends \cite[Thm.\ 6]{NeedellWardTV1} for $d = 2$ to $d \geq 2$ dimensions:
\thm{
\label{t:TVHaarinc2D}
Let $d \geq 2$, $N = 2^r \geq s \geq 2$, $W \in \bbR^{N^d \times N^d}$ be the matrix of the $d$-dimensional discrete Haar wavelet transform and $A \in \bbC^{m \times N^d}$.  Suppose that $A W $ has the RIP of order $t \gtrsim_d s \cdot \log(N) \cdot \log^2(N/s)$ with constant $\delta \leq 1/2$.  Then for every $x \in \bbC^{N^d}$ and $y = A x + e$, where $\nm{e}_{\ell^2} \leq \eta$ for some $\eta \geq 0$, any minimizer $\hat{x}$ of \eqref{TVminprob} satisfies
\bes{
\nm{\hat{x} - x}_{\ell^2} \lesssim_d \frac{\sigma_s(\nabla x)_{\ell^1}}{\sqrt{s \log(N)}} + \eta.
}
}

This result asserts that any measurement matrix which is incoherent with the Haar wavelet basis yields stable and robust recovery via TV minimization.  Hence, much as in the Fourier case, to derive guarantees for Walsh sampling we need to examine its incoherence with the Haar basis. Note that Theorem \ref{t:TVHaarinc2D} does not apply when $d = 1$ (which is the reason our results for Walsh sampling apply only when $d \geq 2$), since it relies crucially on the multi-dimensional Haar coefficient bound that follows from Lemmas \ref{bvtv} and \ref{cbbv}.

\prf{
[Proof of Theorem \ref{t:TVHaarinc2D}]
Since the proof is similar to that of \cite[Thm.\ 6]{NeedellWardTV1}, we omit some details.
Let $z = \hat{x} - x$ and $c = W^* z$ be its discrete Haar coefficients. We may assume $z$ is mean zero. Let $\pi : \{1,\ldots,N^d\} \rightarrow \{1,\ldots,N^d\}$ be a nonincreasing rearrangement of the entries of $c$ in absolute value.  We first show that
\be{
\label{Mtm}
\sum^{N^d}_{j = k+1} |c_{\pi(j)}| \leq C_d \log(N^d/t) \left ( \sum^{k}_{j=1} | c_{\pi(j)} | + \sigma_{s}(\nabla x)_{\ell^1} \right ),
}
where $C_d > 0$ and $k = (2^d - 1) l + 1$ is minimal such that $k \geq \tau_d s \log(N)$ for some constant $\tau_d \geq 1$ to be defined later. Observe that
\bes{
\sum_{j > k} | c_{\pi(j)} | \leq \sum_{i > t} \nmu{c_{(i)}}_{\ell^2},
}
where $c_{(i)} \in \bbC^{2^d-1}$ are the wavelet coefficient blocks, sorted in nonincreasing order. Hence Lemmas \ref{bvtv} and \ref{cbbv} give
\be{
\label{rearrangedecay}
\sum_{j > k} | c_{\pi(j)} | \lesssim \frac{\nm{\nabla z}_{\ell^1}}{2^{d/2}} \log(N^d/t).
}
Let $\Delta$ be the index set of the largest $s$ entries of $\nabla z$ in absolute value.  It is straightforward to show that
\be{
\label{purduebin}
\nmu{P^{\perp}_{\Delta} \nabla z}_{\ell^1} \leq 2 \sigma_{s}(\nabla x)_{\ell^1} + \nmu{P_{\Delta} \nabla z}_{\ell^1}.
}
Now consider $\nm{\nabla z}_{\ell^1}$. Write $\xi_1,\ldots,\xi_{N^d} \in \bbC^{N^d}$ for the discrete Haar basis and let $\Lambda = \{ j : (\nabla \xi_j)_i \neq 0\ \mbox{for some $i \in \Delta$} \}$
be the index set of those Haar wavelets that are nonconstant on $\Delta$.  It is straightforward to show that $|\Lambda | \lesssim_d s \log(N)$, thus we now let $\tau_d$ be such that $|\Lambda| \leq \tau_d s \log(N)$. Write $z = \sum_{j \in \Lambda} c_j \xi_j + \sum_{j \notin \Lambda} c_j \xi_j$.  Then $P_{\Delta} \nabla z = \sum_{j \in \Lambda} c_j P_{\Delta} \nabla \xi_j$ by construction, and therefore
\bes{
\nmu{P_{\Delta} \nabla z}_{\ell^1} \leq \sum_{j \in \Lambda} | c_j | \nmu{ \nabla \xi_j}_{\ell^1} \lesssim_{d} \sum_{j \in \Lambda} | c_j |.
}
Here, in the second step we use the fact that $\nmu{\nabla \xi_j}_{\ell^1} \lesssim_d 1$, which follows easily from the definition of the $\xi_j$.  Combining this with \eqref{purduebin}, we have
\bes{
\nm{\nabla z}_{\ell^1} \lesssim_d \sigma_{s}(\nabla x)_{\ell^1} + \sum_{j \in \Lambda} | c_j| \leq  \sigma_{s}(\nabla x)_{\ell^1} +  \sum^{k}_{j=1} | c_{\pi(j)} |,
}
where in the second step we use the definition of $\pi$ and the fact that  $|\Lambda| \leq \tau_d s \log(N) \leq k$.  Substituting this into \eqref{rearrangedecay} now yields \eqref{Mtm}.

To complete the proof we apply Lemma \ref{l:NWprop} to the matrix $A W$, with values $\gamma = \lceil C_d \log(N^d/t) \rceil$, $\sigma = \gamma \sigma_{s}(\nabla x)_{\ell^1}$ and $\Delta = \{ \pi(1),\ldots,\pi(k) \}$.  The matrix $A W$ satisfies the RIP of order $5 k \gamma^2 \lesssim_d s \cdot \log(N) \cdot \log^2(N/s)$.  Hence
\bes{
\nm{\hat{x} - x}_{\ell^2} = \nm{c}_{\ell^2} \lesssim \frac{\sigma}{\gamma \sqrt{k}} + \nm{A W c }_{\ell^2} \lesssim_d  \frac{\sigma_{s}(\nabla x)_{\ell^1}}{\sqrt{s \log(N)}} + \nm{A (\hat{x} - x)}_{\ell^2}.
}
The result now follows after noting that $\nm{A (\hat{x} - x)}_{\ell^2} \leq 2 \eta$.
}

\lem{
\label{l:WalshHaarRIP}
Let $0<\delta, \varepsilon<1$, $2\leq s\leq N^d$, $\Omega\subseteq\{0,\ldots,N-1\}^d$ be a $d$-dimensional variable sampling pattern corresponding to a probability distribution $p = (p_{i})$, with $\Gamma(p)$ as in \eqref{CpdefWalsh} and $D\in\mathbb{C}^{N^d\times N^d}$ be the  diagonal matrix with entries $D_{ii}=\frac{1}{\sqrt{p_{\varrho^{-1}(i)}}}$. Suppose that
\begin{equation*}
m \gtrsim_d  \Gamma(p)\cdot s \cdot \log(\Gamma(p) s) \cdot \left ( \log(\Gamma(p) s) \cdot \log(N) + \log(\epsilon^{-1}) \right ).
\end{equation*}
Then, with probability at least $1-\varepsilon$, the matrix
\be{
\label{WHaarnormalized}
\frac{1}{\sqrt{mN^d}}P_{\Omega}D HW,
}
has the RIP of order $s$ with constant $\delta_s\leq 1/2$, where $H$ and $W$ are the discrete Walsh--Hadamard and Haar wavelet transforms respectively.
}
\prf{
As in the Fourier case (see the proof of Lemma \ref{ripfhm}), the matrix $N^{-d/2} H W$ is unitary and therefore $A = \frac{1}{\sqrt{m N^d}} P_{\Omega} D H W$ is a randomly-subsampled unitary matrix. Hence it has the RIP of order $s$ whenever  \eqref{BOSRIPcond} holds with $\Theta$ is as in \eqref{Thetadef} for $U = N^{-d/2} H W$. Hence it suffices to show that $\Theta^2 \lesssim_d \Gamma(p)$.

Let $v_i$ denote the one-dimensional Walsh function on $[0,1)$ and $\psi^{(e)}_{j,n}$ be the one-dimensional Haar wavelet. Then
\be{
\label{WalshHaar}
\left | \ip{v_i}{\psi^{(0)}_{j,n}}_{L^2} \right | = \left \{ \begin{array}{cc} 2^{-j/2} & i < 2^j \\ 0 & \mbox{otherwise} \end{array} \right . ,\qquad \left | \ip{v_i}{\psi^{(1)}_{j,n}}_{L^2} \right | =\left \{ \begin{array}{cc} 2^{-j/2} & 2^j \leq i < 2^{j+1} \\ 0 & \mbox{otherwise} \end{array} \right . ,
}
See \cite[Thm.\ 6.8]{AAHWalshWavelet}. In particular, this implies that
\be{
\label{WalshHaarbd}
| \ip{v_i}{\psi^{(e)}_{j,n}}_{L^2} | \leq \left \{ \begin{array}{cc} 2^{-j/2} & i < 2^{j+1} \\ 0 & \mbox{otherwise} \end{array} \right . .
}
Let $\psi^{(e)}_{j,n}$ be the $d$-dimensional Haar wavelets on $[0,1]^d$ and $v_i = v_{i_1} \otimes \cdots \otimes v_{i_d}$ be the $d$-dimensional Walsh functions, where $i = (i_1,\ldots,i_d)$. Then
\bes{
\Theta = \max \left \{ \frac{1}{\sqrt{p_i}} \left | \ip{v_i}{\psi^{(e)}_{j,n}}_{L^2}  \right | \right \},
}
where the maximum is taken over all $i = (i_1,\ldots,i_d) \in \{0,\ldots,N-1\}^d$, $n = (n_1,\ldots,n_d)$ with $0 \leq n_k < 2^j$, $j = 0,\ldots,r-1$ and $e \in \{0,1\}^d$ with $e \neq 0$ unless $j = 0$. Using \eqref{WalshHaarbd}, we have
\bes{
\left | \ip{v_i}{\psi^{(e)}_{j,n}}_{L^2} \right | = \prod^{d}_{k=1} \left | \ip{v_{i_k}}{\psi^{(e_k)}_{j,n_k}}_{L^2} \right |  \leq \left \{ \begin{array}{cc} 2^{-jd/2} & \nmu{i}_{\ell^{\infty}} < 2^{j+1} \\ 0 & \mbox{otherwise} \end{array} \right . .
}
It follows that $ | \ip{v_i}{\psi^{(e)}_{j,n}}_{L^2}  | \lesssim_d  ( 1 + \nmu{i}^{d/2}_{\ell^{\infty}}  )^{-1}$
and therefore
\bes{
\Theta \lesssim_d \max_{i} \left \{ \frac{1}{\sqrt{p_i}} \left ( 1 + \nmu{i}^{d/2}_{\ell^{\infty}} \right )^{-1} \right \} \leq \sqrt{\Gamma(p)},
}
as required.
}

\prf{
[Proof of Theorem \ref{t:WalshTV}]
Let $A'$ be the matrix defined in \eqref{WHaarnormalized} of Lemma \ref{l:WalshHaarRIP}.  This lemma, the condition on $m$ and the fact that $\Gamma(p) \gtrsim \log(N)$ (Lemma \ref{l:GammaWalsh}) imply that $A'$ has the RIP of order $t \gtrsim_{d} s \log(N) \log^2(N/s)$.  To complete the proof, we cannot simply invoke Theorem \ref{t:TVHaarinc2D}, since the measurement matrix $A = \frac{1}{\sqrt{m}} P_{\Omega} H$ is not scaled in such a way for $A W$ to have the RIP.  Instead, we follow the same steps as its proof, making necessary adjustments.  Let $z = \hat{x} - x$, $c = W^* z$ be its Haar coefficients and $k$ be as in the proof.  Then \eqref{Mtm} holds (this property does not depend on the measurement matrix).  We now apply \cite[Prop.\ 3]{NeedellWardTV1} using the matrix $A' $ and the values $\gamma = \lceil C_d\log(N^d/t) \rceil $ and $\sigma = \gamma \sigma_{s}(\nabla x)_{\ell^1}$.  This gives
\eas{
\nm{\hat{x} - x}_{\ell^2} = \nm{d}_{\ell^2} &\lesssim \frac{\sigma_{s}(\nabla x)_{\ell^1}}{\sqrt{s \log(N)}} +\nmu{A' c}_{\ell^2}.
}
Now observe that
\bes{
\nmu{A' c}_{\ell^2} = \nmu{A' W^* (\hat{x} - x) }_{\ell^2} = \frac{1}{\sqrt{N^d} } \nm{D}_{\ell^2} \nm{A (\hat{x} - x) }_{\ell^2} \leq \frac{2}{\sqrt{N^d} \min_{i} \{ \sqrt{p_i} \} } \eta. 
}
Observe that
\bes{
\frac{1}{\sqrt{N^d} \sqrt{p_i}} \leq \sqrt{\Gamma(p)} \frac{\sqrt{1+\nmu{i}^{d}_{\ell^{\infty}} }}{\sqrt{N^{d}}} \leq \sqrt{2\Gamma(p)}.
}
Hence $\nm{A' d}_{\ell^2} \lesssim \sqrt{\Gamma(p)} \eta$, as required.
}

\appendix

\input{TVguarantees_supplement2}

\bibliographystyle{siamplain}
\bibliography{TVguarantees}

\end{document}

%% file: TVguarantees_shared.tex
\ifpdf
  \DeclareGraphicsExtensions{.eps,.pdf,.png,.jpg}
\else
  \DeclareGraphicsExtensions{.eps}
\fi

\setlist[enumerate]{leftmargin=.5in}
\setlist[itemize]{leftmargin=.5in}

\headers{TV minimization in compressive imaging}{B. Adcock, N. Dexter and Q. Xu}

\title{Improved recovery guarantees and sampling strategies for TV minimization in compressive imaging \thanks{Submitted to the editors DATE.
\funding{ND acknowledges the support of the PIMS Postdoctoral Fellowship program. This works was supported by the PIMS CRG ``High-dimensional Data
Analysis'', SFU's Big Data Initiative ``Next Big Question" Fund and NSERC
through grant R611675. }}}

\author{
Ben Adcock, Nick Dexter and Qinghong Xu\thanks{Simon Fraser University, 8888 University Drive
Burnaby, BC V5A 1S6, Canada (\email{ben\_adcock@sfu.ca}, \email{nicholas\_dexter@sfu.ca}, \email{qinghong\_xu@sfu.ca})}
}


%% file: TVguarantees_supplement2.tex
\input{TVguarantees_shared}
\ifpdf
\hypersetup{
  pdftitle={Improved TV guarantees},
  pdfauthor={B. Adcock, N. Dexter and Q. Xu}
}
\fi

\maketitle

\section{Preliminary results from compressed sensing}\label{s:CSpreliminaries}

Below we collect some standard compressed sensing results. For further information, see for instance \cite{FoucartRauhutCSbook}.

\subsection{Sparsity, rNSP and RIP}

Let $N \geq s \geq 2$. Recall that a vector $x \in \bbC^N$ is $s$-sparse if it has at most $s$ nonzero entries. We write $\Sigma_{s}$ for the set of $s$-sparse vectors. Let $D_s$ denote the set of all subsets $\Delta \subseteq \{1,\ldots,N\}$ for which $|\Delta| \leq s$. Thus, $x \in \Sigma_s$ if and only if its support $\mathrm{supp}(x) = \{ i : x_i \neq 0\}$ belongs to $D_s$.

\defn{[Robust Null Space Property]
\label{d:rNSP}
A matrix $A \in \bbC^{m \times N}$ satisfies the \textit{robust Null Space Property (rNSP)} of order $s$ with constants $0 < \rho < 1$ and $\gamma > 0$ if
\be{
\label{rNSP}
\nm{P_{\Delta} x}_{\ell^2} \leq \frac{\rho}{\sqrt{s}} \nmu{P^{\perp}_{\Delta} x}_{\ell^1} + \gamma \nm{A x}_{\ell^2},\qquad \forall x \in \bbC^N,\ \Delta \in D_s.
}
}

\lem{[rNSP implies $\ell^1$ and $\ell^2$ distance bounds]
\label{l:rNSPdistance}
Suppose that $A$ has the rNSP of order $s$ with constants $0 < \rho < 1$ and $\gamma > 0$.  Let $x,z \in \bbC^N$.  Then
\bes{
\nm{z - x}_{\ell^1} \leq \frac{1+\rho}{1-\rho} \left(  2 \sigma_{s}(x)_{\ell^1} + \nm{z}_{\ell^1} - \nm{x}_{\ell^1} \right ) + \frac{2 \gamma}{1-\rho} \sqrt{s} \nm{A(z-x)}_{\ell^2},
}
and
\bes{
\nm{x-z}_{\ell^2} \leq \frac{(3 \rho+1)(\rho+1)}{2(1-\rho)} \left( \frac{2 \sigma_{s}(x)_{\ell^1} + \nm{z}_{\ell^1} - \nm{x}_{\ell^1} }{\sqrt{s}} \right) + \frac{(3 \rho+5) \gamma}{2(1-\rho)} \nm{A (z-x)}_{\ell^2}.
}
}
Note that this result is a special case (corresponding to $M = 1$) of a result proved later, Lemma \ref{rnsp21}.

\defn{
The $s^{\rth}$ \textit{Restricted Isometry Constant (RIC)} $\delta_s$ of a matrix $A \in \bbC^{m \times N}$ is the smallest $\delta \geq 0$ such that
\be{
\label{RIP}
(1-\delta) \| x \|^2_{\ell^2} \leq \| A x \|^2_{\ell^2} \leq (1+\delta) \| x \|^2_{\ell^2},\quad \forall x \in \Sigma_{s}.
}
If $0 < \delta_s < 1$ then the matrix $A$ is said to have the \textit{Restricted Isometry Property (RIP)} of order $s$.
}

Note that the RIP implies the rNSP. For instance, if $A$ has the RIP of order $2s$ with constant $\delta_{2s} < 4 / \sqrt{41}$ then it has the rNSP of order $s$ with constants $\rho$ and $\gamma$ depending on $\delta_{2s}$ \cite[Thm.\ 6.13]{FoucartRauhutCSbook}.

For convenience, we now state one further result:

\begin{lemma}
\label{l:blockArNSP}
If $A \in \bbC^{m \times N}$ satisfies the rNSP of order $s$ with constants $\rho$ and $\gamma$, then 
\begin{equation*}
B=\left(\begin{array}{ccc} A& & \\ &\ddots& \\&&A \end{array}\right)\in\mathbb{C}^{dm\times dN},
\end{equation*}
 has the rNSP of order $s$ with constants $\rho' = \rho$ and $\gamma'=\sqrt{d}\gamma$. 
\end{lemma}
\begin{proof}
Consider any $x=(x^{\top}_1, \ldots, x^{\top}_d )\in\mathbb{C}^{dN}$ with $x_i\in\mathbb{C}^{N}$.  Let $\Lambda\subseteq\{1,\ldots,dN\}$ with $|\Lambda|=s$, and write $\Lambda=\Lambda_1\cup\cdots\cup\Lambda_d$ where $\Lambda_i\subseteq\{(i-1)N+1,\ldots,iN\}$. Since $|\Lambda_i|\leq s$ the rNSP for $A$ gives 
\begin{align*}
\nmu{P_{\Lambda}x }_{\ell^2}& \leq \nmu{P_{\Lambda_1}x_1}_{\ell^2}+\ldots+\nmu{P_{\Lambda_d}x_d}_{\ell^2}\\
&\leq \frac{\rho}{\sqrt{s}}(\nmu{P_{\Lambda_1}^{\perp}x_1}_{\ell^1}+\ldots+\nmu{P_{\Lambda_d}^{\perp}x_d}_{\ell^1})+\gamma(\nmu{A x_1}_{\ell^2}+\ldots+\nmu{A x_d}_{\ell^2})\\
&\leq\frac{\rho}{\sqrt{s}}\nmu{P_{\Lambda}^{\perp}x}_{\ell^1}+ \sqrt{d}\gamma \nmu{B x}_{\ell^2},
\end{align*}
as required.
\end{proof}

\subsection{Bounded orthonormal systems}\label{ss:BOS}

Let $\cD \subset \bbR^d$ be a domain with a probability measure $\nu$ and $\psi_1,\ldots,\psi_N$ be an orthonormal system of complex-valued functions on $\cD$. The system is a \textit{bounded orthonormal system} with constant $\Theta$ if
\bes{
\sup_{t \in \cD} | \psi_j(t) | \leq \Theta,\quad j = 1,\ldots,N.
}
Given such a system, draw $t_1,\ldots,t_m$ random and independently from $\nu$ and define the measurement matrix
\bes{
A = \frac{1}{\sqrt{m}} \left ( \psi_j(t_i) \right )^{m,N}_{i,j=1} \in \bbC^{m \times N}.
}
Let $0 < \delta , \epsilon < 1$ and $N \geq s \geq 2$. The following result was shown in \cite[Thm.\ 2.2]{ChkifaDownwardsCS} (we have slightly simplified the log factor below using the fact that $N \geq s \geq 2$). Suppose that
\be{
\label{BOSRIPcond}
m \gtrsim \delta^{-2} \cdot \Theta^2 \cdot s \cdot L,
\quad 
L = \log \left (\frac{\Theta^2 s}{\delta^2} \right ) \cdot \left [  \frac{1}{\delta^4}\log \left ( \frac{\Theta^2 s}{\delta^2}  \right ) \cdot \log(N) + \frac{1}{\delta}\log \left (\frac{1}{\delta \epsilon} \log \left ( \frac{\Theta^2 s}{\delta^2} \right ) \right ) \right ].
}
Then, with probability at least $1-\epsilon$ the matrix $A$ has the RIP of order $s$ with $\delta_{s} \leq \delta$.

Randomly-subsampled unitary matrices are important examples of the bounded orthonormal system framework. Let $U \in \bbC^{N \times N}$ be unitary and $p = (p_i)^{N}_{i=1}$ be a probability distribution on $\{1,\ldots,N\}$. Draw $t_1,\ldots,t_m$ independently and randomly from $p$ and consider the measurement matrix
\bes{
A  = \frac{1}{\sqrt{m}} P_{T} D U \in \bbC^{m \times N},\qquad D = \diag(1/\sqrt{p_1},\ldots,1/\sqrt{p_N}) \in \bbC^{N \times N},
}
where $T = \{ t_1,\ldots,t_m \}$ and $P_T$ is the row selector matrix. Now let $\cD = \{1,\ldots,N\}$, $\nu$ be the probability measure corresponding to $p$ and define $\phi_j(i) = \frac{1}{\sqrt{p_i}} u_{ij}$, where $U = (u_{ij})$. It is straightforward to verify that this is a bounded orthonormal system. The constant $\Theta$ is
\be{
\label{Thetadef}
\Theta = \max_{i,j=1,\ldots,N} \frac{|u_{ij}|}{\sqrt{p_i}}.
}

\subsection{Matrix recovery}\label{ss:matrixrecov}

We now need a more general version of the rNSP, see for instance, \cite[Defn.\ 4.1]{dexter2018mixed}:
\begin{definition}
\label{l22rNSP}
A matrix $A\in\mathbb{C}^{m\times N}$ satisfies the $\ell^{2,2}$-\textit{robust Null Space Property (rNSP)} of order $s$ with constants $0<\rho<1$ and $\gamma>0$ if
\bes{
\nm{P_{\Delta} X}_{\ell^{2,2}}\leq \frac{\rho}{\sqrt{s}}\nm{P_{\Delta}^{\perp} X}_{\ell^{2,1}}+\gamma\nm{A X}_{\ell^{2,2}},\qquad \forall X \in \bbC^{N \times M},\ \Delta \in D_s.
}
\end{definition}
As shown in \cite[Prop.\ 4.3]{dexter2018mixed}, if $A\in\mathbb{C}^{m\times N}$ satisfies the RIP of order $s$ with constant $\delta_{2s}<\frac{4}{\sqrt{41}}$, then $A$ satisfies the $\ell^{2,2}$-rNSP of order $s$ with constants $\rho$ and $\gamma$ depending on $\delta_{2s}$. We also have the following generalization of Lemma \ref{l:rNSPdistance}:
\begin{lemma}
[rNSP implies $\ell^{2,1}$ and $\ell^{2,2}$ distance bounds] 
\label{rnsp21}
Suppose that $A$ has the $\ell^{2,2}$-rNSP of order $s$ with constants $0<\rho<1$ and $\gamma>0$. Let $X,Z\in\mathbb{C}^{N\times M}$. Then
\begin{equation}
\label{rnsp21bd1}
\nm{Z-X}_{\ell^{2,1}}\leq\frac{1+\rho}{1-\rho}(2\sigma_s(X)_{\ell^{2,1}}+\nm{Z}_{\ell^{2,1}}-\nm{X}_{\ell^{2,1}})+\frac{2\gamma}{1-\rho}\sqrt{s}\nm{A(Z-X)}_{\ell^{2,2}},
\end{equation}
and
\be{
\label{rnsp21bd2}
\nm{Z-X}_{\ell^{2,2}} \leq \frac{(3 \rho+1)(\rho+1)}{2(1-\rho)} \left( \frac{2 \sigma_{s}(X)_{\ell^{2,1}} + \nm{Z}_{\ell^{2,1}} - \nm{X}_{\ell^{2,1}} }{\sqrt{s}} \right) + \frac{(3 \rho+5) \gamma}{2(1-\rho)} \nm{A (Z-X)}_{\ell^{2,2}}.
}
\end{lemma}
\begin{proof}
Consider \eqref{rnsp21bd1}. Let $V=Z-X$ and $\Delta\in D_s$ be such that $\nm{X_{\Delta}^{\perp}}_{\ell^{2,1}}=\sigma_s(X)_{\ell^{2,1}}$. Then we have
\begin{align*}
\nm{X}_{\ell^{2,1}}+\nm{P_{\Delta}^{\perp} V}_{\ell^{2,1}}&=\nm{X}_{\ell^{2,1}}+\nm{P_{\Delta}^{\perp}(Z-X)}_{\ell^{2,1}}\\
&\leq \nm{P_{\Delta} X}_{\ell^{2,1}}+2\nm{P_{\Delta}^{\perp} X}_{\ell^{2,1}}+\nm{P_{\Delta}^{\perp} Z}_{\ell^{2,1}}\\
&=2 \nm{P_{\Delta}^{\perp} X}_{\ell^{2,1}}+\nm{P_{\Delta} X}_{\ell^{2,1}}+\nm{Z}_{\ell^{2,1}}-\nm{P_{\Delta} Z}_{\ell^{2,1}}\\
&\leq 2\sigma_s(X)_{\ell^{2,1}}+\nm{P_{\Delta} V}_{\ell^{2,1}}+\nm{Z}_{\ell^{2,1}},
\end{align*}
which implies that
\begin{equation*}
\nm{P_{\Delta}^{\perp}V}_{\ell^{2,1}}\leq 2\sigma_s(X)_{\ell^{2,1}}+\nm{Z}_{\ell^{2,1}}-\nm{X}_{\ell^{2,1}}+\nm{P_{\Delta}V}_{\ell^{2,1}}.
\end{equation*}
Now consider $\nm{P_{\Delta} V}_{\ell^{2,1}}$. We have
\begin{equation*}
\nm{P_{\Delta}V}_{\ell^{2,1}}\leq \sqrt{s}\nm{P_{\Delta}V}_{\ell^{2,2}}\leq \rho \nm{P_{\Delta}^{\perp}V}_{\ell^{2,1}}+\sqrt{s}\gamma\nm{A V}_{\ell^{2,2}}.
\end{equation*}
Hence
\begin{equation*}
\nm{P_{\Delta}V}_{\ell^{2,1}}\leq \rho(2\sigma_s(X)_{\ell^{2,1}}+\nm{Z}_{\ell^{2,1}}-\nm{X}_{\ell^{2,1}}+\nm{P_{\Delta}V}_{\ell^{2,1}})+\sqrt{s}\gamma\nm{A V}_{\ell^{2,2}},
\end{equation*}
which gives
\begin{equation*}
\nm{P_{\Delta}V}_{\ell^{2,1}}\leq \frac{\rho}{1-\rho}(2\sigma_s(X)_{\ell^{2,1}}+\nm{Z}_{\ell^{2,1}}-\nm{X}_{\ell^{2,1}})+\sqrt{s}\frac{\gamma}{1-\rho}\nm{A V}_{\ell^{2,2}}.
\end{equation*}
Now we have
\begin{align*}
\nm{Z-X}_{\ell^{2,1}}& \leq \nm{P_{\Delta}V}_{\ell^{2,1}}+\nm{P_{\Delta}^{\perp}V}_{\ell^{2,1}}\\
&\leq 2\sigma_s(X)_{\ell^{2,1}}+\nm{Z}_{\ell^{2,1}}-\nm{X}_{\ell^{2,1}}+2\nm{P_{\Delta}V}_{\ell^{2,1}}\\
&\leq \frac{1+\rho}{1-\rho}(2\sigma_s(X)_{\ell^{2,1}}+\nm{Z}_{\ell^{2,1}}-\nm{X}_{\ell^{2,1}})+\frac{2\gamma}{1-\rho}\sqrt{s}\nm{A(Z-X)}_{\ell^{2,2}}.
\end{align*}
This gives \eqref{rnsp21bd1}.

For \eqref{rnsp21bd2}, notice that it suffice show that
\begin{equation}
\label{rnsp21bd3}
\nm{Z-X}_{\ell^{2,2}}\leq\frac{3\rho+1}{2}\frac{\nm{Z-X}_{\ell^{2,1}}}{\sqrt{s}}+\frac{3\gamma}{2}\nm{A(Z-X)}_{\ell^{2,2}}.
\end{equation}
Once this is shown, then \eqref{rnsp21bd2} follows immediately from \eqref{rnsp21bd1}. To show \eqref{rnsp21bd3}, let $V=Z-X$ and write $v_i \in \bbC^M$ for its $i^{\rth}$ row. Let $\Delta \subseteq \{1,\ldots,N\}$ be the index set of the largest $s$ entries of $( \nm{v_i}_{\ell^2} )^{N}_{i=1}$. Then
\begin{equation*}
\nm{P_{\Delta} V}_{\ell^{2,2}}=\sqrt{\sum_{i\in\Delta}\nm{v_i}_{\ell^2}^2}\geq \sqrt{s}\min_{i\in\Delta}\nm{v_i}_{\ell^2}\geq\sqrt{s}\max_{i\notin\Delta}\nm{v_i}_{\ell^2},
\end{equation*}
which implies that
\begin{align*}
\nmu{P_{\Delta}^{\perp} V}^2_{\ell^{2,2}}=\sum_{i\notin\Delta}\nm{v_i}^2_{\ell^2}&\leq \sum_{i\notin\Delta}\nm{v_i}_{\ell^2}\max_{i\notin\Delta}\nm{v_i}_{\ell^2}
\\
& \leq \sum_{i\notin\Delta}\nm{v_i}_{\ell^2}\frac{\nm{P_{\Delta} V}_{\ell^{2,2}}}{\sqrt{s}}=\frac{\nm{P_{\Delta} V}_{\ell^{2,2}}}{\sqrt{s}}\nmu{P_{\Delta}^{\perp} V}_{\ell^{2,1}}.
\end{align*}
Now, applying Young's inequality,  we deduce that
\begin{equation*}
\nm{P_{\Delta}^{\perp} V}_{\ell^{2,2}}\leq \frac{\nm{P_{\Delta} V}_{\ell^{2,2}}}{2}+\frac{\nm{P_{\Delta}^{\perp} V}_{\ell^{2,1}}}{2\sqrt{s}}.
\end{equation*}
Hence
\eas{
\nm{V}_{\ell^{2,2}} \leq\nm{P_{\Delta} V}_{\ell^{2,2}}+\nm{P_{\Delta}^{\perp} V}_{\ell^{2,2}}
&\leq\frac{3}{2}\nm{P_{\Delta} V}_{\ell^{2,2}}+\frac{\nm{P_{\Delta}^{\perp} V}_{\ell^{2,1}}}{2\sqrt{s}}
\\
& \leq\frac{3\rho+1}{2\sqrt{s}} \nm{P_{\Delta}^{\perp} V}_{\ell^{2,1}}+\frac{3\gamma}{2}\nm{AV}_{\ell^{2,2}}.
}
Since $\nm{P_{\Delta}^{\perp} V}_{\ell^{2,1}}\leq \nm{V}_{\ell^{2,1}}$ we obtain the desired result.
\end{proof}

\subsection{Miscellaneous results}

The following is essentially \cite[Prop.\ 3]{NeedellWardTV1}, although with a couple of minor modifications. Since the proof is identical, it is omitted.

\lem{
\label{l:NWprop}
Let $\gamma \in \bbN$ and suppose that $A \in \bbC^{m \times N}$ has the RIP of order $5 k \gamma^2$ with constant $\delta \leq 1/2$.  Let $c \in \bbC^N$ and suppose that there is a set $\Delta \subseteq \{1,\ldots,N\}$ with $|\Delta| \leq k$ such that
\bes{
\nmu{P^{\perp}_{\Delta} c}_{\ell^1} \leq \gamma \nm{P_{\Delta} c}_{\ell^1} + \sigma,
}
for some $\sigma \geq 0$.    Then
\bes{
\nm{c}_{\ell^2} \lesssim \frac{\sigma}{\gamma \sqrt{k}} + \nm{A c}_{\ell^2}.
}
}

\section{Haar wavelets}\label{ss:Haarwave}

\subsection{Definitions}
The Haar scaling function and mother wavelet are defined by
\bes{
\psi^{(0)}(t) = \left \{ \begin{array}{cc} 1 & 0 \leq t < 1  \\ 0 & \mbox{otherwise} \end{array} \right . ,\qquad \psi^{(1)}(t) = \left \{ \begin{array}{cc} 1 & 0 \leq t < 1/2 \\ -1 & 1/2 \leq t < 1 \\ 0 & \mbox{otherwise} \end{array} \right .
}
For $e \in \{0,1\}$, $j,n \in \bbZ$, define $\psi^{(e)}_{j,n}(t) = 2^{j/2} \psi^{(e)}(2^j t - n)$. Then the set
\bes{
\{ \psi^{(0)}_{0,0} \} \cup \{ \psi^{(1)}_{j,n} : n = 0,\ldots,2^{j-1},\ j = 0,1,\ldots \},
}
is an orthonormal basis of $L^2([0,1])$.

Next, consider $d \geq 2$ and for $e = (e_1,\ldots,e_d) \in \{0,1\}^d$, $j \in \bbZ$ and $n = (n_1,\ldots,n_d) \in \bbZ^d$ define the function
\bes{
\psi^{(e)}_{j,n} = \psi^{(e_1)}_{j,n_1} \otimes \cdots \otimes \psi^{(e_d)}_{j,n_d},
}
where $\otimes$ denotes the tensor product. Then
\bes{
\{ \psi^{(0)}_{0,0} \} \cup \{ \psi^{(e)}_{j,n} : e  \in \{0,1\}^d \backslash \{0\},\ n = (n_1,\ldots,n_d),\ n_1,\ldots,n_d = 0,\ldots,2^j-1,\ j = 0,1,\ldots \},
}
is an orthonormal basis of $L^2([0,1]^d)$. 

Given $f \in L^2([0,1]^d)$, we may write
\bes{
f = c^{(0)}_{0,0} \psi^{(0)}_{0,0} + \sum_{e \in \{0,1\}^d \backslash \{0\}} \sum^{\infty}_{j=0} \sum_{\substack{n = (n_1,\ldots,n_d) \\ 0 \leq n_1,\ldots,n_d  < 2^j}} c^{(e)}_{j,n} \psi^{(e)}_{j,n},
}
where $c^{(e)}_{j,n} = \ip{f}{\psi^{(e)}_{j,n}}$. For convenience, we define $c_{j,n} \in \bbC^{2^d-1}$ for the vector containing the values $c^{(e)}_{j,n}$, $e \in \{0,1\}^d \backslash \{0\}$.

Let $d \geq 1$, $N = 2^r$ and consider $\bbC^{N^d}$. Let
\bes{
\phi^{(e)}_{j,n} = \mathrm{vec}(\Phi^{(e)}_{j,n}) \in \bbR^{N^d}
}
where $\Phi^{(e)}_{j,n} \in \bbR^{N \times \cdots \times N}$ with
\bes{
(\Phi^{(e)}_{j,n})_{i} =N^{d/2} \psi^{(e)}_{j,n}(i_1/N,\ldots,i_d/N),\qquad i= (i_1,\ldots,i_d) \in \{0,\ldots,N-1\}^d,
}
is the (normalized) discretization of $\psi^{(e)}_{j,n}$ on an equispaced grid of $N^d$ points on $[0,1]^d$. Then the set
\bes{
\{ \phi^{(0)}_{0,0} \} \cup \{ \phi^{(e)}_{j,n} : e  \in \{0,1\}^d \backslash \{0\},\ n = (n_1,\ldots,n_d),\ n_1,\ldots,n_d = 0,\ldots,2^j-1,\ j = 0,\ldots,r-1 \},
}
is an orthonormal basis for $\bbC^{N^d}$, the discrete Haar basis. After selecting an ordering for this basis, write $W \in \bbR^{N^d \times N^d}$ for the orthogonal matrix whose columns consist of these vectors, i.e.\ the discrete Haar wavelet transform.

\subsection{Relation to the TV semi-norm}

In the following two lemmas, $BV([0,1]^d)$ is the space of functions of bounded variation on $[0,1]^d$, and $|\cdot|_{BV}$ is the usual $BV$ semi-norm, see, for example, \cite{NeedellWardTV2}. The following can be found in \cite[Lem.\ 7]{NeedellWardTV2}:

\begin{lemma}
\label{bvtv}
Let $x=\mathrm{vec}(X)\in\mathbb{C}^{N^d}$, where $X\in\mathbb{C}^{N\times \ldots\times N}$ and $f\in BV([0,1]^d)$ be its isometric embedding as a piecewise constant function, i.e. 
\begin{equation*}
f(i/N)=N^{d/2}X_i,
\end{equation*}
where $i=(i_1,\ldots,i_d)\in \{0,\ldots,N-1\}^d$.  If $|f|_{BV}$ is the BV semi-norm of $f$, then
\begin{equation*}
|f|_{BV}\leq N^{-d/2+1}\nm{x}_{\mathrm{TV}_a}.
\end{equation*}
\end{lemma}
The following result illustrates the relation between Haar coefficients and the BV semi-norm (see, for instance, \cite[Prop.\ 8]{NeedellWardTV2}):

\begin{lemma}
\label{cbbv}
Let $d\geq 2$. There exists a constant $C>0$ such that the following holds for all mean-zero $f\in BV([0,1]^d)$. Let $c^{(e)}_{j,n}$ be the Haar wavelet coefficients of $f$ and $c_{j,n} \in \bbC^{2^d-1}$ be the vector of values $c^{(e)}_{j,n}$, $e \in \{0,1\}^d \backslash \{0\}$. Let $c_{(1)},c_{(2)},\ldots$ be a reordering of these vectors so that $\nmu{c_{(1)}}_{\ell^2} \geq \nmu{c_{(2)}}_{\ell^2} \geq \ldots $. Then
\bes{
|c_{(k)}| \lesssim \frac{|f|_{BV}}{k^{3/2}},\qquad d = 1,
}
and
\begin{equation*}
\nmu{c_{(k)}}_{\ell^2}\lesssim \frac{2^{j_k(d-2)/2}|f|_{BV}}{k},\qquad d \geq 2,
\end{equation*}
where $j_k$ is the scale corresponding to $c_{(k)}$.
\end{lemma}

\subsection{The Fourier transform of a Haar wavelet}

Finally, we also give the following:

\begin{lemma}
\label{fhmeb}
Let $F$ be the one-dimensional DFT matrix, $\{\psi^{(e)}_{j,n}\}$ be the one-dimensional discrete Haar wavelet basis and $\varrho$ be defined as in \eqref{fourier1}. Then for $j=0,\ldots,r-1$, $n=0,\ldots,2^j-1$, $e\in\{0,1\}$ and any $\omega\in\{-N/2+1,\ldots,N/2\}$ we have
\begin{equation}
\label{painfulgeosum}
\frac{1}{\sqrt{N}}|(F\psi^{(e)}_{j,n})_{\varrho^{-1}(\omega)}|\lesssim\frac{1}{\sqrt{\overline{\omega}}}\min\left\{\left(\frac{2^j}{\overline{\omega}}\right)^{1/2},\left(\frac{\overline{\omega}}{2^j}\right)^{1/2+e}\right\}.
\end{equation}
In particular,
\be{
\label{annoyingcdef}
\frac{1}{\sqrt{N}}|(F\psi^{(e)}_{j,n})_{\varrho^{-1}(\omega)}| \lesssim \min\left\{\frac{2^{j/2}}{\overline{\omega}},\frac{\overline{\omega}^e}{2^{j(e+1/2)}}\right\}
 \lesssim \min\left\{ \frac{2^{j/2}}{\overline{\omega}},\frac{1}{2^{j/2}}\right\} = \frac{2^{j/2}}{\max \{ \overline{\omega} , 2^{j} \}}.
}
\end{lemma}

We recall here the definition $\bar{\omega} = \max\{ 1, |\omega|\}$, and that the rows of $F$ are indexed over $\{1,\ldots,N\}$; hence the use of the bijection $\varrho$. 
The calculations that lead to this lemma can be found in, for instance, \cite{DiscreteFourierHaar,KrahmerWardCSImaging}. For completeness we give the proof:

\prf{
[Proof of Lemma \ref{fhmeb}]
We proceed by direct calculation.  We have
\eas{
(F \psi^{(e)}_{j,n})_{\varrho^{-1}(\omega)} =& 2^{\frac{j-r}{2}}\sum_{n 2^{r-j} < t \leq (n+1/2) 2^{r-j} } \E^{-2 \pi \I \omega (t-1)/N} 
\\
& +(-1)^e 2^{\frac{j-r}{2}} \sum_{ (n+1/2) 2^{r-j} < t \leq (n+1) 2^{r-j} } \E^{-2 \pi \I \omega(t-1)/N}
\\
= & 2^{\frac{j-r}{2}} \E^{-2 \pi \I \omega n 2^{r-j}/N} \sum^{2^{r-j-1}-1}_{s=0} \E^{-2 \pi \I \omega s/N}
\\
& +(-1)^e 2^{\frac{j-r}{2}} \E^{-2 \pi \I \omega(n+1/2) 2^{r-j}/N}  \sum^{2^{r-j-1}-1}_{s=0} \E^{-2 \pi \I \omega s/N}.
}
Hence 
\bes{
(F \psi^{(e)}_{j,n})_{\varrho^{-1}(0)} = \left \{ \begin{array}{ll} 2^{\frac{r-j}{2}} & e = 0  \\ 0 & \mbox{otherwise} \end{array} \right . ,
}
and for $\omega \in \{-N/2+1,\ldots,N/2\} \backslash \{ 0\}$,
\be{
\label{psieFTstep}
(F \psi^{(e)}_{j,n})_{\varrho^{-1}(\omega)} = 2^{j/2-r/2} \E^{-2 \pi \I \omega n / 2^j} \left ( 1 +(-1)^e\E^{-2\pi \I \omega/ 2^{j+1}} \right ) \left ( \frac{1-\E^{-2 \pi \I \omega/ 2^{j+1}}}{1-\E^{-2 \pi \I \omega/ 2^r}} \right ).
}
Observe that \eqref{painfulgeosum} trivially holds when $\omega = 0$.  Hence we now consider $\omega \neq 0$.  By \eqref{psieFTstep},
\bes{
\frac{1}{\sqrt{N}} \left | (F \psi^{(e)}_{j,n})_{\varrho^{-1}(\omega)} \right | \leq 2^{j/2-r} \frac{| \sin(\pi \omega/2^{j+1}) |^{1+e}}{|\sin(\pi \omega/2^r) |}.
}
Suppose first that $1 \leq |\omega| < 2^j$.  Then, since $| \sin(\pi z) | \leq \pi | z |$, $\forall z \in \bbR$, and $| \sin(\pi z) | \geq 2 |z|$ for $|z| \leq 1/2$, we have
\bes{
\frac{1}{\sqrt{N}} \left | (F \psi^{(e)}_{j,n})_{\varrho^{-1}(\omega)} \right | \lesssim 2^{j/2-r} \frac{(|\omega|/2^j)^{1+e}}{|\omega|/2^r} = 2^{-j/2} \left(|\omega|/2^j \right)^{e}.
}
Conversely, if $2^j \leq |\omega| \leq 2^{r-1}$ then we use the bound $|\sin(\pi z)|\leq 1$, $\forall z \in \bbR$, to obtain
\bes{
\frac{1}{\sqrt{N}} \left | (F \psi^{(e)}_{j,n})_{\varrho^{-1}(\omega)} \right | \lesssim 2^{j/2-r} \frac{1}{|\omega|/2^r} = 2^{-j/2} \frac{2^j}{|\omega|}.
}
This gives the result.
}

\section{Proof of selected results from \S \ref{s:Foursamp} and \S \ref{s:BeyondFourierTV}}\label{s:proofsrest}

\prf{[Proof of Lemma \ref{l:optimalrule}]
Notice that $q_{\omega} = \overline{\omega_1}$ whenever $\omega = (\omega_1 , 0 ,\ldots, 0 )$. Hence
\bes{
\sum_{\omega} (q_{\omega})^{-2} \geq \sum^{N/2}_{t=1} \frac1t \gtrsim \log(N).
}
Since $p = (p_{\omega})$ is a probability distribution, i.e.\ $\sum_{\omega} p_{\omega} = 1$, we deduce that $\Gamma(p) \gtrsim \log(N)$. 

Now consider the upper bound. Suppose first that $d$ is even. Then there are $d!$ different nonincreasing rearrangements $\pi$. Hence
\eas{
\sum_{\omega} (q_{\omega})^{-2} & \lesssim_d \sum^{N/2}_{t_1=1} \sum^{t_1}_{t_2=1} \cdots \sum^{t_{d-1}}_{t_d=1} \frac{1}{(t_1 \cdots t_{d/2})^2}
 \leq \sum^{N/2}_{t_1=1} \sum^{t_1}_{t_2=1} \cdots \sum^{t_{d/2-1}}_{t_{d/2} =1} \frac{(t_{d/2})^{d/2}}{(t_1 \cdots t_{d/2})^2} = F_{d/2}(N/2),
}
where
\bes{
F_m(N) = \sum^{N}_{t_1=1} \sum^{t_1}_{t_2=1} \cdots \sum^{t_{m-1}}_{t_m=1} \frac{(t_m)^m}{(t_1 \cdots t_m)^2}.
}
Similarly, if $d$ is odd we have
\eas{
\sum_{\omega} (q_{\omega})^{-2} & \lesssim_d \sum^{N/2}_{t_1=1} \sum^{t_1}_{t_2=1} \cdots \sum^{t_{d-1}}_{t_d=1} \frac{1}{(t_1 \cdots t_{(d-1)/2})^2 t_{(d+1)/2} }
\\
& \leq \sum^{N/2}_{t_1=1} \sum^{t_1}_{t_2=1} \cdots \sum^{t_{(d-1)/2}}_{t_{(d+1)/2}=1} \frac{(t_{(d+1)/2})^{(d-1)/2}}{(t_1 \cdots t_{(d-1)/2})^2 t_{(d+1)/2} }
\\
& = F_{(d+1)/2}(N/2).
}
We now show that $F_m(N) \lesssim_m \log(N)$ for any $m \in \bbN$. When $m = 1$ the result is is trivial. Now consider $m \geq 2$. We have
\bes{
F_m(N) \lesssim_m \sum^{N}_{t_1=1} \sum^{t_1}_{t_2=1} \cdots \sum^{t_{m-2}}_{t_{m-1} = 1} \frac{(t_{m-1})^{m-1}}{(t_1 \cdots t_{m-1})^2} = F_{m-1}(N).
}
Hence the result follows by induction. Therefore, for either even or odd $d$, we have shown that
\bes{
\sum_{\omega} (q_{\omega})^{-2} \lesssim_d \log(N).
}
Since $p = (p_{\omega})$ is a probability distribution the result now follows.
}

\prf{
[Proof of Lemma \ref{l:radialsubopt}]
Since all norms are equivalent on $\bbR^d$, we may without loss of generality consider the $\ell^{\infty}$-norm. We first estimate the constant $C_{N,d,\alpha}$. Hence
\be{
\label{borederaser}
(C_{N,d,\alpha})^{-1} \asymp_{d} \sum^{N/2}_{t_1 = 1} \sum^{t_1}_{t_2 = 1} \ldots \sum^{t_{d-1}}_{t_d=1} \frac{1}{(t_d)^{\alpha}} \asymp_{d,\alpha} \sum^{N/2}_{t_1 =1} (t_d)^{d-1-\alpha} \asymp_{d,\alpha} \left \{ \begin{array}{cc} N^{d-\alpha} & \alpha < d \\ \log(N) & \alpha = d \\ 1 & \alpha > d  \end{array} \right . .
}
Next, observe that $\Gamma(p)$ is defined by
\bes{
\Gamma(p) C_{N,d,\alpha} =  \max_{\omega} \frac{(1+\nm{\omega}_{\ell^{\infty}})^\alpha }{(q_{\omega})^2} \asymp_{\alpha} \max_{\omega} \frac{(\overline{\omega_{\pi(1)}})^{\alpha}}{(q_{\omega})^2 }.
}
We now split into two cases. Suppose first that $\alpha \geq 2$. Then, using the definition of $q_{\omega}$, we see that the maximum is attained at $\omega = (N/2,0,\ldots,0)$, giving
\bes{
\Gamma(p) C_{N,d,\alpha} \asymp_{\alpha} N^{\alpha-2}.
}
Conversely, when $\alpha < 2$ the maximum is attained when $\omega = (0,\ldots,0)$, giving
\bes{
\Gamma(p) C_{N,d,\alpha} \asymp_{\alpha} 1.
}
We now combine these two estimates with \eqref{borederaser} to get the result.
}

\prf{
[Proof of Corollary \ref{c:HCnearoptimal}]
Observe that
\bes{
(C_{N,d})^{-1} = \sum_{\omega} \frac{1}{\prod^{d}_{j=1} \overline{\omega_j}} \asymp_{d} \left ( \sum^{N/2}_{t=1} \frac{1}{t} \right )^d \asymp_{d} \log^{d}(N).
}
Moreover, using the fact that $\overline{\omega_{\pi(j)}} \geq \sqrt{\overline{\omega_{\pi(j)}} \overline{\omega_{\pi(d/2+j)}}}$ when $d$ is even, and similarly for $d$ odd, we deduce that
\bes{
(q_{\omega})^{2} \geq \overline{\omega_1} \cdots \overline{\omega_d} = \frac{C_{N,d}}{p_{\omega}},
}
and therefore $\Gamma(p) \leq (C_{N,d})^{-1}$. Hence $\Gamma(p) \asymp_d \log^d(N)$, which gives the result.
}

\prf{[Proof of Lemma \ref{l:GammaWalsh}]
This follows immediately from \eqref{borederaser}.
}

%% file: TVguarantees_paper_arXiv.bbl
\begin{thebibliography}{10}

\bibitem{AdcockEtAlCISampStrat}
{\sc B.~Adcock, V.~Antun, R.~Bergman, and A.~C. Hansen}, {\em Effective
  sampling strategies for compressive imaging}, In preparation,  (2019).

\bibitem{AAHWalshWavelet}
{\sc B.~Adcock, V.~Antun, and A.~C. Hansen}, {\em Uniform recovery in
  infinite-dimensional compressed sensing and applications to structured binary
  sampling}, arXiv:1905.00126,  (2019).

\bibitem{BASBMKRCSwavelet}
{\sc B.~Adcock, S.~Brugiapaglia, and M.~King-Roskamp}, {\em Do log factors
  matter? {O}n optimal wavelet approximation and the foundations of compressed
  sensing}, arXiv:1905.10028,  (2019).

\bibitem{AHPRBreaking}
{\sc B.~Adcock, A.~C. Hansen, C.~Poon, and B.~Roman}, {\em Breaking the
  coherence barrier: A new theory for compressed sensing}, Forum Math. Sigma, 5
  (2017).

\bibitem{DiscreteFourierHaar}
{\sc B.~Adcock, A.~C. Hansen, and B.~Roman}, {\em A note on compressed sensing
  of structured sparse wavelet coefficients from subsampled {F}ourier
  measurements}, IEEE Signal Process. Letters, 23 (2016), pp.~732--736.

\bibitem{Antun16}
{\sc V.~Antun}, {\em Coherence estimates between {H}adamard matrices and
  {D}aubechies wavelets}, master's thesis, University of Oslo, 2016.

\bibitem{BastounisHansen}
{\sc A.~Bastounis and A.~C. Hansen}, {\em On the absence of uniform recovery in
  many real-world applications of compressed sensing and the restricted
  isometry property and nullspace property in levels}, SIAM J. Imaging Sci., 10
  (2017), pp.~335--371.

\bibitem{Becker2011}
{\sc S.~Becker, J.~Bobin, and E.~J. Cand{\`{e}}s}, {\em {NESTA: A Fast and
  Accurate First-Order Method for Sparse Recovery}}, SIAM Journal on Imaging
  Sciences, 4 (2011), pp.~1--39.

\bibitem{CaiXuTVGauss}
{\sc J.-F. Cai and W.~Xu}, {\em Guarantees of total variation minimization for
  signal recovery}, Inf. Inference, 4 (2015), pp.~328--353.

\bibitem{CandesRombergTao}
{\sc E.~J. Cand{\`e}s, J.~Romberg, and T.~Tao}, {\em Robust uncertainty
  principles: exact signal reconstruction from highly incomplete frequency
  information}, IEEE Trans. Inform. Theory, 52 (2006), pp.~489--509.

\bibitem{ChambolleEtAlTVDenoise}
{\sc A.~Chambolle, V.~Duval, G.~Peyr{\'e}, and C.~Poon}, {\em Geometric
  properties of solutions to the total variation denoising problem}, Inverse
  Problems, 33 (2016), p.~015002.

\bibitem{ChambolleEtAlTVImage}
{\sc A.~Chambolle, M.~Novaga, D.~Cremers, and T.~Pock}, {\em An introduction to
  total variation for image analysis}, in Theoretical Foundations and Numerical
  Methods for Sparse Recovery, M.~Fornasier, ed., vol.~9 of Radon Series in
  Computational and Applied Mathematics, de Gruyter, Berlin, 2010,
  pp.~263--340.

\bibitem{ChamPockAN}
{\sc A.~Chambolle and T.~Pock}, {\em An introduction to continuous optimization
  for imaging}, Acta Numer., 25 (2016), pp.~161--319.

\bibitem{ChkifaDownwardsCS}
{\sc A.~Chkifa, N.~Dexter, H.~Tran, and C.~G. Webster}, {\em Polynomial
  approximation via compressed sensing of high-dimensional functions on lower
  sets}, Math. Comp., 87 (2018), pp.~1415--1450.

\bibitem{dexter2018mixed}
{\sc N.~Dexter, H.~Tran, and C.~Webster}, {\em {A mixed $\ell_1$ regularization
  approach for sparse simultaneous approximation of parameterized PDEs}}, ESAIM
  Math. Model. Numer. Anal., 53 (2019), pp.~2025--2045.

\bibitem{FoucartRauhutCSbook}
{\sc S.~Foucart and H.~Rauhut}, {\em A Mathematical Introduction to Compressive
  Sensing}, Birkhauser, 2013.

\bibitem{GaussWalsh}
{\sc E.~Gauss}, {\em {Walsh Funktionen f\"ur Ingenieure und
  Naturwissenschaftler}}, Springer Fachmedien Wiesbaden, 1994.

\bibitem{GolubovEtAlWalsh}
{\sc B.~Golubov, A.~Efimov, and V.~Skvortsov}, {\em Walsh Series and
  Transforms: Theory and Applications}, Springer Netherlands, 1991.

\bibitem{KrahmerEtAlTVCS}
{\sc F.~Krahmer, C.~Kruschel, and M.~Sandbichler}, {\em Total variation
  minimization in compressed sensing}, in Compressed Sensing and Its
  Applications, Birkh\"auser, 2018.

\bibitem{KrahmerWardCSImaging}
{\sc F.~Krahmer and R.~Ward}, {\em Stable and robust sampling strategies for
  compressive imaging}, IEEE Trans. Image Process., 23 (2013), pp.~612--622.

\bibitem{LiAdcockRIP}
{\sc C.~Li and B.~Adcock}, {\em Compressed sensing with local structure:
  uniform recovery guarantees for the sparsity in levels class}, Appl. Comput.
  Harmon. Anal., 46 (2019), pp.~453--477.

\bibitem{MoshtThesis}
{\sc A.~Moshtaghpour}, {\em Computational Interferometry for Hyperspectral
  Imaging}, PhD thesis, Universit{\'e} catholique de Louvain, 2019.

\bibitem{MoshtEtAlHadamardHaar}
{\sc A.~Moshtaghpour, J.~B. Dias, and L.~Jacques}, {\em Close encounters of the
  binary kind: signal reconstruction guarantees for compressive {H}adamard
  sampling with {H}aar wavelet basis}, IEEE Trans. Inf. Theory (in press),
  (2020).

\bibitem{NeedellWardTV2}
{\sc D.~Needell and R.~Ward}, {\em Near-optimal compressed sensing guarantees
  for total variation minimization}, IEEE Trans. Image Process., 22 (2013),
  pp.~3941--3949.

\bibitem{NeedellWardTV1}
{\sc D.~Needell and R.~Ward}, {\em Stable image reconstruction using total
  variation minimization}, SIAM J. Imaging Sci., 6 (2013), pp.~1035--1058.

\bibitem{PoonTV}
{\sc C.~Poon}, {\em On the role of total variation in compressed sensing}, SIAM
  J. Imaging Sci., 8 (2015), pp.~682--720.

\bibitem{AsymptoticCS}
{\sc B.~Roman, A.~C. Hansen, and B.~Adcock}, {\em On asymptotic structure in
  compressed sensing}, arXiv:1406.4178,  (2014).

\bibitem{TemylakovHC}
{\sc V.~Temlyakov}, {\em Multivariate Approximation}, Cambridge University
  Press, 2018.

\end{thebibliography}
